\documentclass{article}

\usepackage{fullpage}
\usepackage{libertine}

\usepackage[verbose=true,letterpaper]{geometry} % to remove redundant margin of neurips sty with no option.

\usepackage{amsmath,amsthm,amssymb,mathtools}
\usepackage{algorithm}
\usepackage[hidelinks, hypertexnames=false, colorlinks=true, citecolor=Navy, linkcolor=Maroon, urlcolor=Orchid, bookmarksnumbered, unicode]{hyperref}
\usepackage[noend]{algpseudocode}
\usepackage[capitalize]{cleveref}

\crefname{line}{Line}{Lines}
\Crefname{line}{Line}{Lines}
\usepackage{autonum}
\makeatletter
\autonum@generatePatchedReferenceCSL{Cref}
\makeatother
\usepackage[numbers, compress]{natbib}
\usepackage{thm-restate}

\usepackage[utf8]{inputenc} % allow utf-8 input
\usepackage[T1]{fontenc}    % use 8-bit T1 fonts
\usepackage{url}            % simple URL typesetting
\usepackage{booktabs}       % professional-quality tables
\usepackage{amsfonts}       % blackboard math symbols
\usepackage{nicefrac}       % compact symbols for 1/2, etc.
\usepackage{microtype}      % microtypography
\usepackage[svgnames]{xcolor} % colors

\usepackage{graphicx} % colored texts do not appear with this option
\usepackage{bm}
\usepackage{comment}
\usepackage{mleftright}\mleftright
\usepackage{nicematrix}
\usepackage{subcaption}
\usepackage{wrapfig}
\usepackage{multirow}
\usepackage{siunitx}

\usepackage[color={red!100!green!33},colorinlistoftodos,prependcaption,textsize=small]{todonotes}

\usepackage{lipsum}

\theoremstyle{plain}
\newtheorem{thm}{Theorem}[section]
\newtheorem{lem}[thm]{Lemma}
\newtheorem{prop}[thm]{Proposition}

\theoremstyle{definition}

\theoremstyle{remark}

\theoremstyle{observation}
\newtheorem{obs}[thm]{Observation}

\theoremstyle{plain}
\newtheorem*{thm*}{Theorem}
\newtheorem*{lem*}{Lemma}
\newtheorem*{prop*}{Proposition}
\newtheorem*{cor*}{Corollary}
\newtheorem*{conj*}{Conjecture}

\theoremstyle{definition}
\newtheorem*{ass*}{Assumption}
\newtheorem*{dfn*}{Definition}

\theoremstyle{remark}
\newtheorem*{rem*}{Remark}

\newcommand{\R}{\mathbb{R}}
\newcommand{\N}{\mathbb{N}}

\newcommand{\ee}{\mathrm{e}}
\DeclareMathOperator*{\argmax}{arg\,max}

\newcommand{\step}[2]{\vphantom{{#1}^{i}}\smash{#1^{(#2)}}}
\newcommand{\nil}{\textsc{nil}}
\let\set\relax
\DeclarePairedDelimiter\set{\{}{\}}
\let\Set\relax
\DeclarePairedDelimiterX\Set[2]{\{}{\}}{\mspace{2mu}{#1}\;\delimsize|\;{#2}\mspace{2mu}}
\newcommand{\Ord}{\mathrm{O}}
\title{Lazy and Fast Greedy MAP Inference for \\Determinantal Point Process}

\author{%
  Shinichi Hemmi\\
  The University of Tokyo\\
  \href{mailto:hemmi.shinichi@gmail.com}{hemmi.shinichi@gmail.com}
  \and
  Taihei Oki\\
  The University of Tokyo\\
  \href{mailto:oki@mist.i.u-tokyo.ac.jp}{oki@mist.i.u-tokyo.ac.jp}\\
  \and
  Shinsaku Sakaue\\
  The University of Tokyo\\
  \href{mailto:sakaue@mist.i.u-tokyo.ac.jp}{sakaue@mist.i.u-tokyo.ac.jp} \\
  \and
  Kaito Fujii\\
  National Institute of Informatics\\
  \href{mailto:fujiik@nii.ac.jp}{fujiik@nii.ac.jp}\\
  \and
  Satoru Iwata\thanks{SI also belongs to Institute for Chemical Reaction Design and Discovery (WPI‑ICReDD).}\\
  The University of Tokyo\\
  \href{mailto:iwata@mist.i.u-tokyo.ac.jp}{iwata@mist.i.u-tokyo.ac.jp}\\
}
\date{}

\begin{document}

\maketitle

\begin{abstract}
  The maximum a posteriori (MAP) inference for determinantal point processes (DPPs) is crucial for selecting diverse items in many machine learning applications. Although DPP MAP inference is NP-hard, the greedy algorithm often finds high-quality solutions, and many researchers have studied its efficient implementation. One classical and practical method is the lazy greedy algorithm, which is applicable to general submodular function maximization, while a recent fast greedy algorithm based on the Cholesky factorization is more efficient for DPP MAP inference. This paper presents how to combine the ideas of ``lazy'' and ``fast'', which have been considered incompatible in the literature. Our lazy and fast greedy algorithm achieves almost the same time complexity as the current best one and runs faster in practice. The idea of ``lazy + fast'' is extendable to other greedy-type algorithms. We also give a fast version of the double greedy algorithm for unconstrained DPP MAP inference. Experiments validate the effectiveness of our acceleration ideas.
\end{abstract}

\section{Introduction}\label{section:introduction}
The determinantal point process (DPP) has been a popular diversification model in machine learning.
\citet{Macchi1975-yh} first used DPPs to represent repulsion in quantum physics, and later, DPPs have been used in various scenarios such as recommendation systems~\cite{Chen2018-aa}, document summarization~\cite{Gillenwater2012-uw,Kulesza2012-er}, and diverse molecule selection~\citep{Nakamura2022-nd}.
An important problem in the DPP applications is the \emph{maximum a posteriori} (MAP) inference, which asks to find an item subset with the highest probability.
Intuitively, if each item is associated with a vector whose length and direction represent its importance and feature, respectively, then the aim of DPP MAP inference is to select items whose vectors form the largest volume parallelotope, thus selecting important and diverse items.

In DPPs, exact MAP inference is NP-hard \citep{Ko1995-yx}.
Fortunately, however, the standard greedy algorithm (\textsc{Greedy}) for submodular function maximization~\citep{Nemhauser1978-dr} enjoys a ($1 - 1/\mathrm{e}$)-approximation guarantee in terms of the log-determinant function value under the monotonicity assumption, and it often finds high-quality solutions in practice.
A naive implementation of \textsc{Greedy}, however, incurs too much computation cost for large instances since evaluating the determinant of a $k\times k$ matrix takes $\Ord(k^\omega)$ time, where $\omega \in [2, 3]$ is the matrix-multiplication exponent (usually $\omega = 3$).
To overcome this issue, researchers have studied techniques for implementing efficient greedy algorithms.
One classical and powerful method is the so-called \emph{lazy} greedy algorithm (\textsc{LazyGreedy})~\citep{Minoux1978-jr}, which avoids the redundant computation of function values by making good use of the submodularity.
\textsc{LazyGreedy} is applicable to submodular function maximization, and it empirically runs much faster than naive \textsc{Greedy}, although the worst-case running time is not improved.
\citet{Chen2018-aa} proposed another notable Cholesky-factorization-based method, called the \emph{fast} greedy algorithm (\textsc{FastGreedy}).
Their algorithm is specialized for DPP MAP inference and provides the fastest $\Ord(knd)$-time implementation of \textsc{Greedy} for selecting $k$ out of $n$ items represented by $d$-dimensional vectors.
\citet{Chen2018-aa} also experimentally showed that \textsc{FastGreedy} can run faster than \textsc{LazyGreedy}.

Since the study of \citep{Chen2018-aa}, although a pre-processing method for customized DPP MAP inference~\citep{Han2020-ez} and fast parallel algorithms for submodular function maximization~\citep{Balkanski2018-dq,Balkanski2019-mj,Fahrbach2019-kz,Breuer2020-do,Ene2020-gj,Chen2021-vs,Kuhnle2021-wb} have been studied, no progress has been made that directly accelerates \textsc{Greedy} for DPP MAP inference.
Since real-world dataset sizes have been growing, a further speed-up of \textsc{Greedy} is eagerly awaited.

\textbf{Our main contribution} is to combine the two ideas, ``lazy'' and ``fast'', to develop an even faster implementation of \textsc{Greedy} for DPP MAP inference, which we call \textsc{LazyFastGreedy}.
In the literature, the two ideas have been thought to be incompatible; in fact, experiments in~\citep{Chen2018-aa} compared \textsc{FastGreedy} with \textsc{LazyGreedy} without considering their combination.
The core idea of ``lazy + fast'' is widely applicable to other greedy-type algorithms.
Below is a summary of our results.

\begin{enumerate}
  \item We present \textsc{LazyFastGreedy} for cardinality-constrained DPP MAP inference. It takes $\Ord(kn(d + \log n))$ time even in the worst case and is faster than \textsc{FastGreedy} in practice.
  We also extend the idea of ``lazy + fast'' to other greedy-type algorithms: \textsc{RandomGreedy}~\citep{Buchbinder2014-il}, \textsc{StochasticGreedy}~\citep{Mirzasoleiman2015-qo,Sakaue2020-ap}, and \textsc{InterlaceGreedy}~\citep{Kuhnle2019-io}.
  \item We present a ``fast'' version of \textsc{DoubleGreedy}~\citep{Buchbinder2015-rc} for unconstrained DPP MAP inference by extending the idea of \citep{Chen2018-aa} with Jacobi's complementary minor formula.
  \item Experiments on synthetic and real-world datasets validate the empirical effectiveness of our acceleration techniques for the greedy-type algorithms.
  In particular, our \textsc{LazyFastGreedy} runs up to about $17$ times faster than \textsc{FastGreedy} in real-world settings.
\end{enumerate}
The greedy variants \citep{Buchbinder2014-il,Buchbinder2015-rc,Mirzasoleiman2015-qo,Kuhnle2019-io,Sakaue2020-ap} enjoy approximation guarantees for non-monotone submodular function maximization, which are relevant since the log-determinant function is non-monotone in general.
In short, we accelerate various important greedy-type algorithms for DPP MAP inference.

\subsection{Related work}
The greedy algorithm (\textsc{Greedy}) is a popular approach to DPP MAP inference~\citep{Kulesza2012-er}.
\citet{Han2017-bv} gave a fast but inexact implementation of \textsc{Greedy}.
Later, \citet{Chen2018-aa} gave an exact implementation of \textsc{Greedy} with the same time complexity as that of \citep{Han2017-bv}.
\citet{Han2020-ez} studied the case where kernel matrices of DPPs are generated via customization (or re-weighting) of fixed feature vectors and developed a pre-processing method for accelerating \textsc{Greedy}.
A continuous-relaxation-based $1/4$-approximation algorithm for general down-closed constraints was also studied~\citep{Gillenwater2012-uw}.
Our implementation of \textsc{InterlaceGreedy} \citep{Kuhnle2019-io} yields a faster $1/4$-approximation algorithm for a special case with a cardinality constraint.
Approximation algorithms and inapproximability results of DPP MAP inference (without log) have also been extensively studied \citep{Civril2009-uj,Mahabadi2019-bc,Mahabadi2020-vh,Bhaskara2020-cg,Ohsaka2022-ng}.
Sampling is another important research subject in DPPs and has been widely studied \citep{Anari2016-hw,Li2016-ia,Derezinski2019-zg,Gillenwater2019-gq,Anari2020-cj,Calandriello2020-qh,Launay2020-af}.

Since the log-determinant function is known to have the submodularity~\citep{Fan1968-pg}, DPP MAP inference has a close connection to submodular function maximization \citep{Nemhauser1978-dr}.
Besides \textsc{LazyGreedy}, \textsc{StochasticGreedy} \citep{Mirzasoleiman2015-qo,Sakaue2020-ap} is a popular fast variant of \textsc{Greedy}, which we will discuss later.
A recent line of work \citep{Balkanski2018-dq,Balkanski2019-mj,Fahrbach2019-kz,Breuer2020-do,Ene2020-gj,Chen2021-vs,Kuhnle2021-wb} has studied \textit{adaptive} algorithms for submodular function maximization, where we are allowed to execute polynomially many queries in parallel to reduce the number of sequential rounds.
Those studies consider oracle models of submodular functions, whereas we focus on the log-determinant functions and develop fast algorithms without such parallelization.

\section{Background}\label{section:background}
Let $[n]$ denote the set $\{1,2,\dots,n\}$ for any $n \in \N$.
For any $S \subseteq [n]$, $\overline{S}$ denotes its complement $[n] \setminus S$.
We use ${\bm 0}$ and $O$ as all-zero vectors and matrices, respectively, and ${\bm 1}$ as all-ones vectors (their sizes will be clear from the context).
Let $\langle\cdot, \cdot\rangle$ denote the inner product.
For a matrix $L\in\R^{n\times n}$ and subsets $S, T \subseteq [n]$, $L[S,T]$ is the submatrix of $L$ indexed by $S$ in rows and $T$ in columns.
For brevity, we write $L_{i,j} = L[\{i\},\{j\}], L[S] = L[S,S], L[S,i] = L[S,\{i\}]$, and $L[i,S] = L[\{i\}, S]$ for any $i,j \in [n]$ and $S \subseteq [n]$.
The determinant of $L$ is denoted by $\det L$.
Set $\det L[\emptyset] =1$ by convention.
For any $M \in \R^{n \times n}$, we suppose that $M^\top M$ is computed in $\Ord(n^\omega)$ time, which implies that we can compute $\det M$ and $M^{-1}$ (if non-singular) in $\Ord(n^\omega)$ time (see, e.g., \citep[Chapter~2]{Bini1994-qa}).

\paragraph{DPP MAP inference.}
Let $L\in\R^{n\times n}$ be a positive semi-definite matrix.
A probability measure $\mathcal{P}$ on $2^{[n]}$ is called a \emph{determinantal point process}~(DPP) with a kernel matrix $L$ if $\mathcal{P}[X = S] \propto \det L[S]$ holds for all $S\subseteq [n]$.
MAP inference for DPPs is the problem of finding a subset $S \subseteq [n]$ with the largest $\det L[S]$ value.
We suppose that each $i \in [n]$ is associated with a vector ${\bm \phi_i}\in\R^{d}$ and that a kernel matrix $L\in\R^{n\times n}$ is given as $L=B^{\top}B$, where $B = \mleft[{\bm \phi_1}, {\bm \phi_2}, \dots, {\bm \phi_n}\mright] \in\R^{d\times n}$, i.e., $L_{i,j} =\mleft\langle {\bm \phi_{i}}, {\bm \phi_{j}}\mright\rangle$ for $i,j\in[n]$.
Note that $\sqrt{\det L[S]}$ represents the volume of the parallelotope spanned by $\{{\bm \phi}_i\mid i\in S\}$.
Hence, if the length and direction of ${\bm \phi_i}$ indicate $i$'s importance and feature, respectively, the larger value of $\det L[S]$ implies that $S$ contains more important and diverse items.
In many situations, we want to select a limited number of items; let $k \in \N$ denote the upper bound.
Therefore, MAP inference for DPPs with a cardinality constraint, or $k$-DPP \citep{Kulesza2011-pf}, is often considered.
Note that since $\mathop{\mathrm{rank}} L \le \min\set{n, d}$, we can assume $k \le \min\set{n, d}$ without loss of generality.

\paragraph{Submodular function maximization.}
For a set function $f\colon 2^{[n]}\to\R\cup\set{-\infty}$, the \emph{marginal gain} of $i\in[n]$ with respect to $S\subseteq[n]$ is defined by $f_i(S)=f(S\cup\{i\})- f(S)$.
A set function $f\colon 2^{[n]}\to\R\cup\set{-\infty}$ is called \emph{monotone} if $f_i(S)\ge 0$ for every $S\subseteq [n]$ and $i \in \overline{S}$.
It is called \emph{submodular} if it has the \emph{diminishing returns property}: $f_i(S) \geq f_i(T)$ for every $S\subseteq T\subseteq [n]$ and $i\in \overline{T}$.
Since $f(S) = \log \det L[S]$ is submodular, DPP MAP inference can be written as submodular function maximization: $\max_{S\in\mathcal{X}} f(S)$, where $\mathcal{X}\subseteq 2^{[n]}$ is a family of feasible subsets.
This paper mostly considers the cardinality-constrained setting, i.e., $\mathcal{X}=\{S\subseteq [n] \mid |S| \le k\}$ for given $k \in \N$;
in \cref{section:double-greedy}, we study the unconstrained setting, i.e., $\mathcal{X}=2^{[n]}$.
The log-determinant function $f$ is monotone if the smallest eigenvalue of $L$ is at least $1$, but this is not always the case in practice.

It is well-known that the greedy algorithm (\textsc{Greedy}) enjoys a $(1-1/\ee)$-approximation guarantee for cardinality-constrained monotone submodular function maximization with $f(\emptyset) \ge 0$~\citep{Nemhauser1978-dr}.
This approximation ratio is optimal under the evaluation oracle model~\citep{Nemhauser1978-vk}.
\textsc{Greedy} works as follows:
setting $\step{S}{0}=\emptyset$, in each $t$th step ($t = 1,\dots,k$), choose $j_t \in \argmax \Set{f_i(\step{S}{t-1})}{i \in \overline{\step{S}{t-1}}}$ and put $\step{S}{t} = \step{S}{t-1} \cup \set*{j_t}$.
We call a solution obtained in this way a \emph{greedy solution}.

Besides \textsc{Greedy}, many algorithms \citep{Buchbinder2015-rc,Buchbinder2014-il,Kuhnle2019-io,Sakaue2020-ap} achieve constant-factor approximations for cardinality-constrained/unconstrained submodular function maximization.
These results motivate us to apply submodular-function-maximization algorithms to DPP MAP inference, although constant-factor approximations of the log-determinant value do not imply those of the determinant value.

\subsection{Lazy greedy algorithm for submodular function maximization}

\textsc{LazyGreedy}~\cite{Minoux1978-jr} is an efficient implementation of \textsc{Greedy} for submodular function maximization.
As explained above, \textsc{Greedy} finds $j_t$ by computing marginal gains $f_i(\step{S}{t-1})$ for all $i \in \overline{\step{S}{t-1}}$.
\textsc{LazyGreedy} attempts to find $j_t$ more efficiently by keeping an upper bound $\rho_i$ on $f_i(\step{S}{t-1})$ for each $i \in \overline{\step{S}{t-1}}$, which is an \emph{old} marginal gain, i.e., $\rho_i=f_i(S^{(u_i)})$ for some $u_i \le t-1$.
In each iteration, \textsc{LazyGreedy} picks $i\in \overline{\step{S}{t-1}}$ with the largest $\rho_i$ value as the most promising element.
It then updates the $\rho_i$ value to the latest marginal gain $f_i\mleft(\step{S}{t-1}\mright)$.
If $\rho_i$ is still the largest among $\rho_{i^\prime}$ for all $i^\prime\in\overline{\step{S}{t-1}}$, the diminishing returns property guarantees $i \in \argmax \Set{f_{i^\prime}(\step{S}{t-1})}{i^\prime \in \overline{\step{S}{t-1}}}$, and hence it adds $j_t = i$ to $\overline{\step{S}{t-1}}$.
\textsc{LazyGreedy} thus constructs a greedy solution, deferring updates of upper bounds of unpromising elements.
If the upper bounds are managed by a priority queue, every single iteration computes the marginal gain only once and makes $\Ord(\log n)$ comparisons (note that \text{``iteration'' is distinguished from ``step'' of \textsc{Greedy}}).
With this contrivance, \textsc{LazyGreedy} runs much faster than \textsc{Greedy} in practice, even though it does not improve the worst-case complexity.

\subsection{Fast greedy algorithm for DPP MAP inference}
We turn to DPP MAP inference with a kernel matrix $L = B^\top B$ and $B \in \R^{d \times n}$.
If we apply \textsc{Greedy} to $f(S) = \log \det L[S]$, it takes $\Ord(k^{\omega+1}n)$ time since computing an $f$ value takes $\Ord(k^{\omega})$ time; in addition, computing $L$ at first takes $\Ord(\min \set{n^{\omega-1}d, n^2d^{\omega-2}})$ time.
\textsc{FastGreedy} \citep{Chen2018-aa} provides an $\Ord(knd)$-time implementation of \textsc{Greedy} for cardinality-constrained DPP MAP inference.

\cref{alg:FG_MAP_Inf} describes the procedure of \textsc{FastGreedy}, which is based on the Cholesky factorization of $L$ with maximum pivoting.
The Cholesky decomposition produces a matrix $V\in\R^{n\times n}$ such that $L=VV^{\top}$ and $PV$ is lower triangular for some permutation matrix $P$.
We call $V$ a \emph{Cholesky factor} of $L$.
In each $t$th step, \cref{line:fast-for-Sc,line:fast-compute-Vijt,line:compute-di} calculate the $t$th column of $PV$ via backward substitution (see \cref{fig:CZZ}).
The following \cref{prop:rel_bet_Gre_Chol} implies an important fact that once $V[i, \step{S}{t}]$ is filled, we can obtain $f_i(\step{S}{t})$ from $d_i^{(t)}$ computed in \cref{line:compute-di}, which is equal to the diagonal $V_{i, i}$ of the current Cholesky factor.
Although the proposition is already proved in \citep{Chen2018-aa}, we present a proof sketch since it would be helpful to understand the subsequent discussion (see~\cite{Chen2018-aa} for the complete proof).

\begin{restatable}[\cite{Chen2018-aa}]{prop}{RelBetGreChol}\label{prop:rel_bet_Gre_Chol}
  For $t = 0, 1, \dots, k-1$, it holds that $f_i\mleft(\step{S}{t}\mright) = 2\log \step{d_i}{t}$ for every $i\in\overline{\step{S}{t}}$.
\end{restatable}

\begin{proof}[Proof sketch]
  The proof is by induction on $t$.
  If $t=0$, it holds that $f_i\mleft(\step{S}{t}\mright) = f(\{i\})-f(\emptyset)  = \log L_{i,i} = 2\log d_i^{(0)}$ for $i\in [n]$.
  Given a Cholesky factor $V[\step{S}{t}]$ of $L[\step{S}{t}]$, for $i\in\overline{\step{S}{t}}$, we have
  \begin{equation}\label{eq:Cholesky_ith_step}
    L[\step{S}{t}\cup \{i\}]
  =
    \begin{bmatrix}
      L\mleft[\step{S}{t}\mright] & L\left[\step{S}{t},i\right]\\
      L\left[i,\step{S}{i}\right] & L_{i,i}
   \end{bmatrix}
  =
    \left[\begin{array}{cc}
        V[\step{S}{t}] & {\bm 0}\\
        V[i,\step{S}{t}] & d_{i}^{(t)}
     \end{array}
    \right]
    \left[\begin{array}{cc}
        {V[\step{S}{t}]}^{\top} & {V[i,\step{S}{t}]}^\top \\
        {\bm 0}^{\top} & d_{i}^{(t)}
     \end{array}
    \right].
\end{equation}
Hence, we have $\log\det L[\step{S}{t}\cup\{{i}\}]=\log {\left(\step{d_{i}}{t}\det V[\step{S}{t}]\right)}^2=2\log \step{d_{i}}{t}+\log\det L[\step{S}{t}]$, which implies $\log\det L[\step{S}{t}\cup\{i\}] - \log\det L[\step{S}{t}] =2\log \step{d_{i}}{t}$.
Thus, the statement holds.
\end{proof}

Therefore, by iteratively adding $j_{t} \in \argmax_{i\in \overline{\step{S}{t-1}}}\,d_{i}^{(t-1)}$ to $\step{S}{t-1}$ as in \cref{alg:FG_MAP_Inf}, we can obtain a greedy solution.
Since $L_{i, j_t} = \mleft\langle {\bm \phi_{i}}, {\bm \phi_{j_t}}\mright\rangle$ is computed in $\Ord(d)$ time and $d \ge k$, \cref{line:fast-compute-Vijt} takes $\Ord(d)$ time.
This is repeated $\Ord(kn)$ times, hence the total time complexity of $\Ord(knd)$.

\begin{algorithm}[tb]
  \caption{\textsc{FastGreedy} \citep{Chen2018-aa} for cardinality-constrained DPP MAP inference}\label{alg:FG_MAP_Inf}
  \begin{algorithmic}[1]
    \State $V\gets O,~d_i^{(0)} \leftarrow \sqrt{L_{i,i}}~(\forall i\in [n]),~S^{(0)} \leftarrow\emptyset$%
    \For{$t=1$ to $k$}
      \State Take $j_{t} \in \argmax_{i\in \overline{\step{S}{t-1}}}\,d_{i}^{(t-1)}$ \Comment{Terminate if ${d}^{(t-1)}_{j_t} \le 1$ (i.e., $f_{j_t}(\step{S}{t-1}) \le 0$)}
      \State $\step{S}{t}\leftarrow \step{S}{t-1}\cup\{j_{t}\}$ \Comment{$\step{S}{t} = \set*{j_1,\dots,j_t}$}
      \For{$i$ in $\overline{\step{S}{t}}$} {\label[line]{line:fast-for-Sc}}
      \Comment{Skip \cref{line:fast-for-Sc,line:fast-compute-Vijt,line:compute-di} (updates for the next step) if $t = k$}
        \State $V_{i,j_t} \leftarrow (L_{i,j_t}-\langle V[i,\step{S}{t-1}], V[j_t,\step{S}{t-1}]\rangle)/\step{d_{j_{t}}}{t-1}$ {\label[line]{line:fast-compute-Vijt}}
        \Comment{$V[i, \emptyset] = {\bm 0}$ for any $i \in [n]$}
        \State $d_i^{(t)} \leftarrow \sqrt{\left(\step{d_i}{t-1}\right)^2-{V_{i,j_t}}^2}$ {\label[line]{line:compute-di}}
      \EndFor
    \EndFor
    \State \Return $\step{S}{k}$ \Comment{Or $\step{S}{t-1}$ if terminated with $t<k$}
  \end{algorithmic}
\end{algorithm}

\section{Lazy and fast algorithms for cardinality-constrained DPP MAP inference}

\subsection{Lazy and fast greedy algorithm}\label{subsection:lazyfastgreedy}

\begin{figure}[tb]
	\centering
	\begin{minipage}[b]{.465\linewidth}
			\centering
			\includegraphics[width=.7\linewidth]{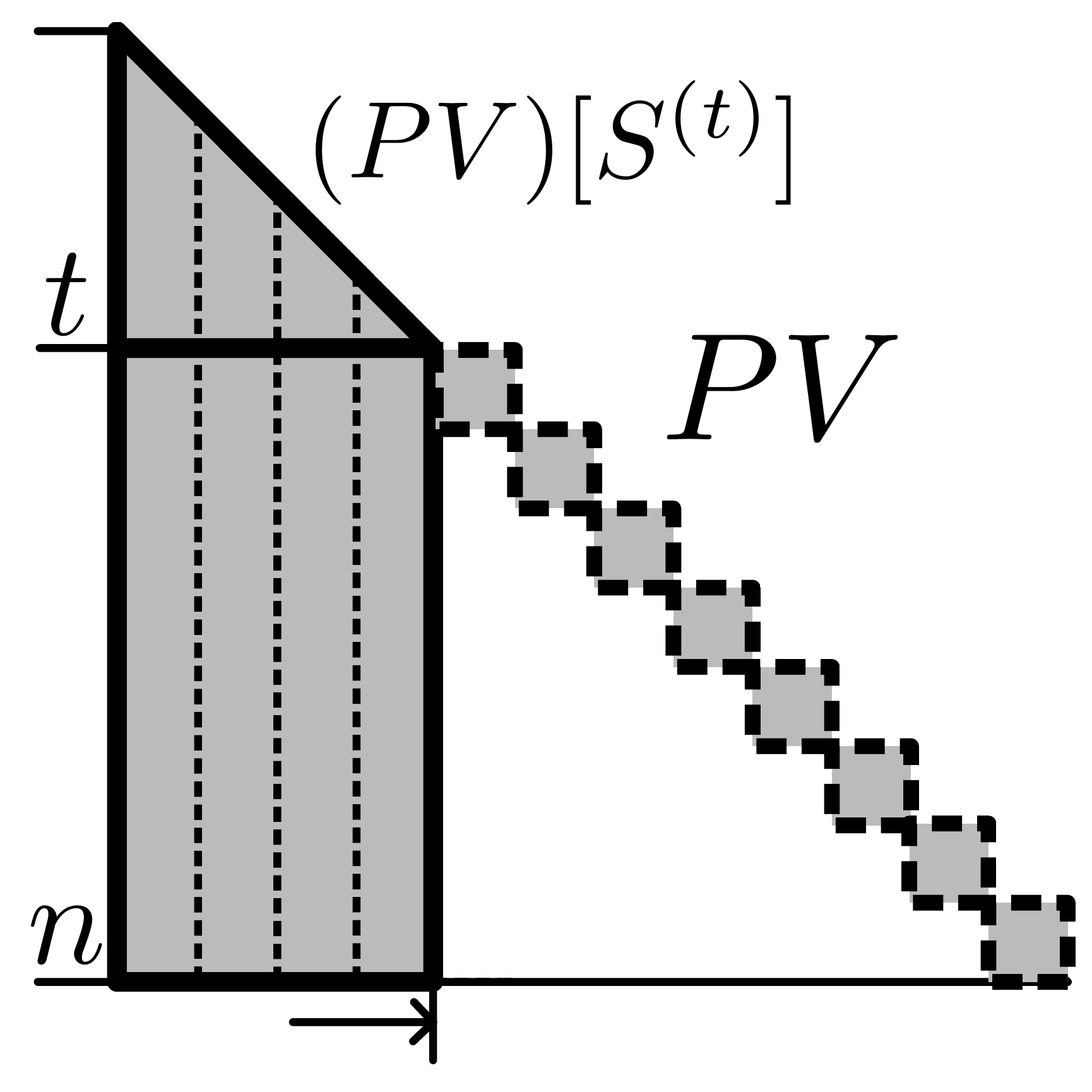}
			\subcaption{\textsc{FastGreedy}}\label{fig:CZZ}
	\end{minipage}
	\begin{minipage}[b]{.465\linewidth}
			\centering
			\includegraphics[width=.7\linewidth]{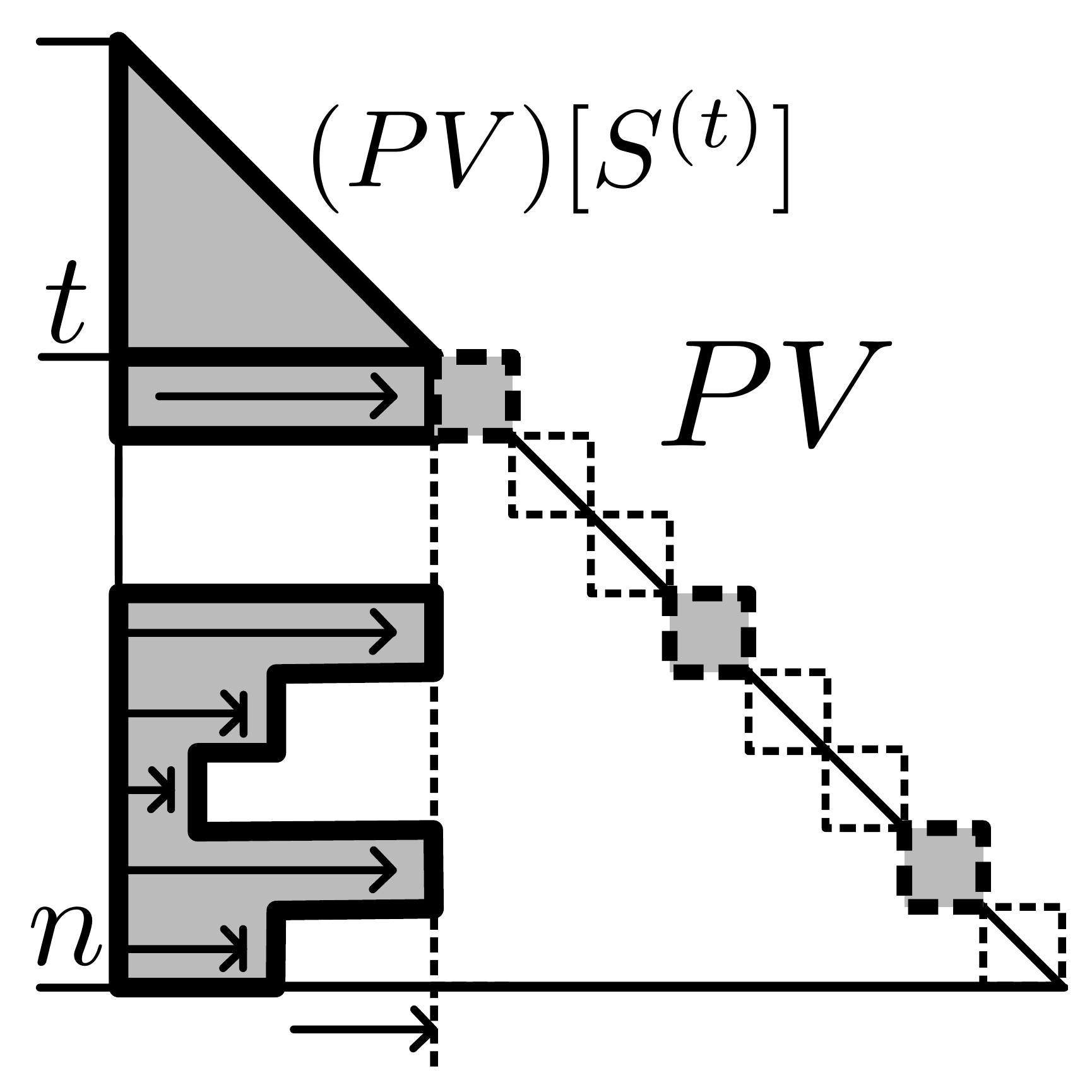}
			\subcaption{\textsc{LazyFastGreedy}}\label{fig:Lazy}
	\end{minipage}
	\caption{
		Images of Cholesky factors computed by \textsc{FastGreedy} and \textsc{LazyFastGreedy}.
    Shaded areas indicate the entries that are already computed.
    Diagonals with bold (thin) dashed lines indicate where the latest marginal gains are (not) available.
	}\label{fig:wise_update}
\end{figure}

% \begin{wrapfigure}[15]{r}[0pt]{.5\textwidth}
% 	\centering
% 	\begin{minipage}[b]{.465\linewidth}
% 			\centering
% 			\includegraphics[width=.785\linewidth]{new_img/image_of_CZZ.pdf}
% 			\subcaption{\textsc{FastGreedy}}\label{fig:CZZ}
% 	\end{minipage}
% 	\begin{minipage}[b]{.465\linewidth}
% 			\centering
% 			\includegraphics[width=.785\linewidth]{new_img/image_of_Lazy.pdf}
% 			\subcaption{\textsc{LazyFastGreedy}}\label{fig:Lazy}
% 	\end{minipage}
% 	\caption{
% 		Images of Cholesky factors computed by \textsc{FastGreedy} and \textsc{LazyFastGreedy}.
%     Shaded areas indicate the entries that are already computed.
%     Diagonals with bold (thin) dashed lines indicate where the latest marginal gains are (not) available.
% 	}\label{fig:wise_update}
% \end{wrapfigure}

% \begin{figure}[tb]
%   \centering
%   \begin{subfigure}{0.3\columnwidth}
%     \centering
%     \includegraphics[width=\columnwidth]{img/image_of_CZZ.jpeg}
%     \caption{\textsc{FastGreedy}}
%     \label{fig:CZZ}
%   \end{subfigure}
%   \begin{subfigure}{0.3\columnwidth}
%     \centering
%     \includegraphics[width=\columnwidth]{img/image_of_Lazy.jpeg}
%     \caption{\textsc{LazyFastGreedy}}
%     \label{fig:Lazy}
%   \end{subfigure}
%   \caption{Comparison of \textsc{FastGreedy} and \textsc{LazyFastGreedy}.
%   The diagonal entries surrounded by bold and dashed squares are where the marginal gain on $\step{S}{t}$(todo) is calculated, while these surrounded by real lines(todo) are where updates are deferred.\rr{todo: change indiecs in figs}}
%   \label{fig:wise_update}
% \end{figure}

We combine \textsc{LazyGreedy} and \textsc{FastGreedy} to obtain even faster \textsc{LazyFastGreedy}.

As described above, \textsc{LazyGreedy} is designed for submodular function maximization under the oracle model, where a marginal gain is computed for each $i \in \overline{\step{S}{t-1}}$.
Meanwhile, the core idea of \textsc{FastGreedy} is to obtain marginal gains of all $i \in \overline{\step{S}{t-1}}$ efficiently by computing a new column of a Cholesky factor, as in \cref{fig:CZZ}.
To combine these seemingly incompatible methods, we need to take a closer look at how \textsc{FastGreedy} updates the entries in $V$.
The following observation is obvious but elucidates the essential row-wise independence of the updates of Cholesky-factor entries.
\begin{obs}\label{obs:independency}
  In each $t$th step in \cref{alg:FG_MAP_Inf}, for each $i\in\overline{\step{S}{t}}$, $V_{i,j_t}$ in \cref{line:fast-compute-Vijt} is computed from $\step{d_{j_t}}{t-1}$, $L_{i,j_t}$, and $V[\{i,j_t\},\step{S}{t-1}]$.
  Thus, conditioned on $\step{S}{t}$, the $i$th row of $V[\overline{\step{S}{t}},\step{S}{t}]$ can be updated independently for each $i\in\overline{\step{S}{t}}$ in a sense that computing $V_{i,j_t}$ only requires $V_{i,j_{t^{\prime}}}$ for $t^{\prime} < t$ and $V[j_t,\step{S}{t-1}]$, which is included in the already fixed Cholesky factor $V[\step{S}{t}]$.
\end{obs}
That is, in \cref{fig:Lazy}, once $(PV)[\step{S}{t}]$ is fixed, we can update the $i$th row independently for each $i \in \overline{\step{S}{t}}$.
In the words of submodular function maximization, we can compute
$f_i\left(\step{S}{t}\right)$ if $f_{j_{t^\prime}}\left(\step{S}{t^{\prime}-1}\right)$ and $f_{i}\left(\step{S}{t^{\prime}-1}\right)$ for $t^{\prime}\le t$ are available.
This enables us to apply the idea of lazy updates to \textsc{FastGreedy}.

\cref{alg:LazyMapInf} describes our \textsc{LazyFastGreedy}, and \cref{fig:Lazy} illustrates how entries in $V$ are updated.
In each iteration, it picks the most promising element $i$ in \cref{line:lazyfast-argmax}, as with \textsc{LazyGreedy}.
Then, it calls \textsc{UpdateRow} to compute the entries of $V[i, S]$ via backward substitution.
Once $V[i, S]$ is filled, the resulting $d_i$ computed in \cref{line:updaterow-compute-di} satisfies $2\log d_i = f_i(S)$ by \cref{prop:rel_bet_Gre_Chol}, thus obtaining the latest marginal gain of $i$.
\cref{line:lazyfast-if-max,line:lazyfast-jsp,line:lazyfast-add-to-S} check whether the latest $f_i(S)$ is still the largest among the (old) marginal gains of $i^\prime \in \overline{S}$; if so, $j_{|S|+1} = i$ is added to $S$.
Note that the deferred updates of $d_{i^\prime}$ ($i^\prime \in \overline{S} \setminus \set*{i}$) do not matter when deciding whether to add $i$ to $S$ due to the submodularity.
Here, $\bm{d} = (d_1, \dotsc, d_n)$ plays the role of upper bounds $(\rho_1,\dots,\rho_n)$ of \textsc{LazyGreedy} and is maintained by a priority queue.
\textsc{LazyFastGreedy} thus finds a greedy solution by exactly mimicking the behavior of \textsc{LazyGreedy} while computing marginal gains efficiently as with \textsc{FastGreedy}.

The vector ${\bm u} = (u_1,\dots,u_n)$ keeps track of when the upper bounds are last updated; specifically, each $u_i \in \set*{0,1,\dots,k-1}$ indicates that the upper bound $d_i$ is last updated with respect to $\set*{j_1,\dots,j_{u_i}}$, i.e., $2\log d_i = f_i(\set*{j_1,\dots,j_{u_i}})$.
Since \textsc{UpdateRow} starts to fill $V[i, S]$ from the $(u_i + 1)$st entry, no off-diagonals $V_{i, j_t}$ are computed more than once in \cref{line:updaterow-compute-Vijt}.
Thus, the complexity of \cref{alg:LazyMapInf} depends on the number of computed off-diagonals.
We denote it by $U = \sum_{i\in[n]}u_i$ with $u_i$ values at the end of \cref{alg:LazyMapInf}.
The $U$ value changes in the range of $\left[ k(k-1)/2, (k-1)(n-k/2) \right]$ depending on how well the lazy update works and affects the overall time complexity as follows.

\begin{algorithm}[tb]
  \caption{\textsc{LazyFastGreedy} for cardinality-constrained DPP MAP inference}\label{alg:LazyMapInf}
  \begin{algorithmic}[1]
    \State{$V\gets O, \bm{d} \gets {\left(\sqrt{L_{i,i}}\right)}_{i \in [n]}, \bm{u} \gets {\bm 0}, S\gets \emptyset$}\Comment{${\bm d}$ is maintained by a priority queue}
   \While{$|S|<k$}
    \State{Take $i \in \argmax_{i^\prime\in \overline{S}} d_{i'}$} \Comment{Terminate if ${d}_{i} \le 1$ (i.e., $f_{i}(S) \le 0$)} {\label[line]{line:lazyfast-argmax}}
    \State \Call{UpdateRow}{$V, \bm{d}, \bm{u}; i, S, L$}\Comment{Nothing is done if $|S| = 0$}
    \If {$d_{i}\geq\max_{i^\prime\in \overline{S}} d_{i^\prime}$} {\label[line]{line:lazyfast-if-max}}
    \Comment{Otherwise insert $d_i$ into the priority queue}
    \State $j_{|S|+1}\leftarrow i$ {\label[line]{line:lazyfast-jsp}}
    \State $S\leftarrow S\cup\{j_{|S|+1}\}$ {\label[line]{line:lazyfast-add-to-S}}
    \EndIf
   \EndWhile
  \State \Return $S$
  \item[]
  \Function{UpdateRow}{$V, \bm{d}, \bm{u}; i, S, L$}
  \Comment{$S = \set*{j_1, j_2,\dots,j_{|S|}}$}
  \For{$t = u_i+1, u_i+2,\dots, |S|$}
    \State $V_{i,{j}_{t}} \leftarrow (L_{i,{j}_{t}}-\langle V[i, \step{S}{t-1}], V[j_t, \step{S}{t-1}]\rangle)/d_{j_t}$
    \Comment{$\step{S}{t-1} = \set*{j_1,\dots,j_{t-1}}$} {\label[line]{line:updaterow-compute-Vijt}}
    \State $d_i \leftarrow \sqrt{{d_i}^2-{V_{i,{j}_t}}^2}$ {\label[line]{line:updaterow-compute-di}}
  \EndFor
  \State{$u_i\leftarrow |S|$}\Comment{This line is not needed in \cref{alg:DoubleGreedy}}
  \EndFunction
  \end{algorithmic}
\end{algorithm}

\begin{thm}\label{thm:LazyGreedyCorrectness_Complexity}
  \Cref{alg:LazyMapInf} returns a greedy solution in $\Ord(nd + U(d+\log n))$ time.
  If the lazy update works best and worst, it runs in $\Ord((n + k^2)d)$ and $\Ord(kn(d + \log n))$ time, respectively.
\end{thm}
\begin{proof}
  \Cref{alg:LazyMapInf} returns a greedy solution as explained above.
  We below discuss the running time.

  At the beginning, we need $\Ord(nd)$ time to compute $L_{i,i} = \mleft\langle {\bm \phi_{i}}, {\bm \phi_{i}}\mright\rangle$ for $i=1,\dots,n$.
  In \textsc{UpdateRow}, an access to $L_{i, j_t}$ in \cref{line:updaterow-compute-Vijt} takes $\Ord(d)$ time, and the inner product takes $\Ord(k)$ ($\lesssim \Ord(d)$) time.
  This computation is done $U$ times, and thus the total computation time caused by \textsc{UpdateRow} is $\Ord(Ud)$.
  In \cref{line:lazyfast-if-max}, we need $\Ord(\log n)$ time to update the priority queue if $d_{i} < \max_{i^\prime\in \overline{S}} d_{i^\prime}$, which can hold only when at least one off-diagonal is computed in \textsc{UpdateRow}.
  Therefore, \cref{line:lazyfast-if-max} takes $\Ord(U\log n)$ time in total.
  Thus, the overall time complexity is $\Ord(nd + U(d+\log n))$.

  Let $S$ be the output of \cref{alg:LazyMapInf} and $P$ a permutation matrix such that $(PV)[S]$ is lower triangular.
  In the best case, \textsc{UpdateRow} is called up to $k$ times and $U = k(k-1)/2$ off-diagonals of $(PV)[S]$ are computed.
  Moreover, updates of the priority queue are done only up to $k$ times, taking $\Ord(k \log n)$ ($\lesssim \Ord(nd)$) time in total.
  Thus, it runs in $\Ord((n + k^2)d)$ time.
  In the worst case, \cref{alg:LazyMapInf} calculates the off-diagonals of $(PV)[S]$ and all the entries of $V[\overline{S}, S \setminus \set{j_k}]$; the total number of those entries is $U = k(k-1)/2 + (k-1)(n-k) = (k-1)(n-k/2)$.
  Hence, it takes $\Ord(kn(d + \log n))$ time.
\end{proof}

The best-case time complexity is better than $\Ord(knd)$ of \textsc{FastGreedy} if $k = \mathrm{o}(n)$.
Even in the worse case, it is as fast as \textsc{FastGreedy} if $d = \Omega(\log n)$.
Note that both $k = \mathrm{o}(n)$ and $d = \Omega(\log n)$ are true in most practical situations.
Experiments in \cref{section:experiments} demonstrate that \textsc{LazyFastGreedy} can run much faster than \textsc{FastGreedy} in practice.

\subsection{Extension to random, stochastic, and interlace greedy algorithms}
The core idea of \textsc{LazyFastGreedy} can be used for speeding up other greedy-type algorithms:
\textsc{RandomGreedy}~\citep{Buchbinder2014-il}, \textsc{StochasticGreedy}~\citep{Mirzasoleiman2015-qo,Sakaue2020-ap}, and \textsc{InterlaceGreedy}~\citep{Kuhnle2019-io}, which enjoy $1/\ee$-, $1/4$-, and $1/4$-approximation guarantees, respectively, for non-monotone submodular function maximization with a cardinality constraint.
Note that the guarantees for the non-monotone case are essential in DPP MAP inference since the log-determinant function is non-monotone in general.
Due to the space limitation, we present the details of those extensions in \cref{app-section:extension}.

\section{Fast double greedy algorithm for unconstrained DPP MAP inference}\label{section:double-greedy}
This section discusses unconstrained DPP MAP inference with a kernel matrix $L = B^\top B$, where $B \in \R^{d \times n}$.
In this setting, if $f(S) = \log \det L[S]$ is monotone, $S = [n]$ is a trivial optimal solution.
Thus, we suppose $f$ to be non-monotone.
We also assume $L$ to be positive definite since the algorithm of \citep{Buchbinder2015-rc} discussed below requires $f(S) > -\infty$ for any $S \subseteq [n]$.
Note that this implies $d \ge n$.

A famous algorithm for unconstrained submodular function maximization is \textsc{DoubleGreedy} \citep{Buchbinder2015-rc}, a randomized $1/2$-approximation algorithm.
Although it calls an evaluation oracle only $\Ord(n)$ times, its naive implementation is too costly for large DPP MAP inference instances since computing the log-determinant function value takes $\Ord(n^{\omega})$ time, which will lead to the total time of $\Ord(n^{\omega-1}d + n^{\omega+1})$.
We below extend the idea of \textsc{FastGreedy} \citep{Chen2018-aa} to \textsc{DoubleGreedy} and obtain its $\Ord(n^{\omega-1}d + n^3)$-time implementation for unconstrained DPP MAP inference.

\textsc{DoubleGreedy} maintains two subsets $S$ and $T$, which are initially set to $S=\emptyset$ and $T=[n]$.
For $i = 1,\dots,n$, it computes $a_i = \max\set*{f_i(S), 0}$ and $b_i = \max\set*{-f_i(T\setminus\set*{i}), 0}$, and then either adds $i$ to $S$ with probability $a_i/(a_i + b_i)$ or removes $i$ from $T$ with probability $b_i/(a_i + b_i)$.\footnote{The algorithm thus sequentially examines all elements, and hence there is no room for the lazy update.}
Note that $T = [n]\setminus ([i]\setminus S) = \overline{[i]\setminus S}$ always holds.
Finally, it returns $S$ (or equivalently $T = \overline{[n] \setminus S} = S$).

As for the growing subset $S$, we can efficiently compute marginal gains $f_i(S)$ by incrementally updating a Cholesky factor, as with \textsc{FastGreedy}.
When it comes to the shrinking subset $T$, however, we cannot directly use the efficient incremental update for computing $-f_i(T\setminus\set*{i})$.
If we naively compute it in each step, it takes $\Ord(n^\omega)$ time, resulting in the same time complexity as the naive implementation.
Our key idea for overcoming this difficulty is to use the following Jacobi's complementary minor formula (see, e.g.,~\cite{Brualdi1983-nl}).

\begin{prop}\label{prop:Jacobi}
  Let $L\in\R^{n\times n}$ be a non-singular matrix and $I, J\subseteq [n]$ be subsets with $|I|=|J|$.
  Then, it holds that $\det L[I,J] = {(-1)}^{\sum_{i\in I} i+\sum_{j\in J} j}\det L \det {L}^{-1}[\overline{I},\overline{J}]$.
\end{prop}

This formula provides a lemma that enables us to compute $-f_i(T\setminus\set*{i})$ via incremental updates.

\begin{lem}\label{lem:DoubleGreedyforLogDet}
   Let $L\in\R^{n\times n}$ be positive definite.
   Define $f(S)=\log\det L[S]$ and $g(S)=\log L^{-1}[S]$ for any $S \subseteq [n]$.
   Then, $g(S\cup\{i\}) - g(S) = f(\overline{S} \setminus \{i\}) - f(\overline{S})$ holds for any $S\subseteq [n]$ and $i\in[n]$.
\end{lem}

\begin{proof}
  By using \cref{prop:Jacobi}, we can prove the claim as follows:
  \begin{align}
    f(\overline{S} \setminus \{i\}) - f(\overline{S})
    &= \log\det{L[\overline{S} \setminus \{i\}]}-\log\det{L[\overline{S}]}\\
    &= \log\det L \det {L}^{-1}[[n] \setminus(\overline{S}\setminus\{i\})]-\log \det L \det {L}^{-1}[[n]\setminus\overline{S}]\\
    &= \log\det {L}^{-1}[S \cup \{i\}]-\log\det {L}^{-1}[S]\\
    &= g(S\cup\{i\}) - g(S).
    \qedhere
  \end{align}
\end{proof}

\begin{algorithm}[tb]
  \caption{\textsc{FastDoubleGreedy} for unconstrained DPP MAP inference}\label{alg:DoubleGreedy}
  \begin{algorithmic}[1]
    \State Compute $L = B^\top B$ and $L^{-1}$ {\label[line]{line:double-compute-L}}
    \State $V \gets O$, $W \gets O$, $\bm{d} \gets {\left(\sqrt{L_{i,i}}\right)}_{i \in [n]}$, $\bm{e} \gets {\Big(\sqrt{{\left(L^{-1}\right)}_{i,i}}\Big)}_{i \in [n]}, S\gets\emptyset$ {\label[line]{line:double-init}}
    \For {$i = 1$ to $n$} {\label[line]{line:double-for}}
      \State\Call{\textsc{UpdateRow}}{$V, \bm{d}, \bm{0}; i, S, L$}\Comment{$S$ is sorted in order of $1,2,\dots,n$} {\label[line]{line:double-updaterow1}}
      \State\Call{\textsc{UpdateRow}}{$W, \bm{e}, \bm{0}; i, [i] \setminus S, L^{-1}$}\Comment{$[i] \setminus S$ is sorted in order of $1,2,\dots,n$} {\label[line]{line:double-updaterow2}}
      \State $a_i\gets \max\set*{2\log d_i, 0}$, $b_i\gets \max\set*{2\log e_i, 0}$ {\label[line]{line:double-compute-ab}}
      \State $S\leftarrow S\cup\{i\}$ w.p. $a_i/(a_i + b_i)$ {\label[line]{line:double-add-S}} \Comment{Implicity update $T = \overline{[n]\setminus S}$ w.p. $b_i/(a_i + b_i)$}
    \EndFor
  \State \Return $S$
  %\Comment{Or equivalently return $T = \overline{[i]\setminus S}$}
  \end{algorithmic}
\end{algorithm}

\Cref{lem:DoubleGreedyforLogDet} implies $g(([i]\setminus S) \cup \set*{i}) - g([i]\setminus S) = f(\overline{([i]\setminus S)} \setminus \set*{i}) - f(\overline{[i]\setminus S}) = -f_i(T)$.
Therefore, by computing $L = B^\top B$ and $L^{-1}$ in $\Ord(n^{\omega-1}d)$ ($\gtrsim \Ord(n^\omega)$) time in advance, we can compute $-f_i(T\setminus\set*{i})$ in each step by incrementally updating a Cholesky factor of size $|[i]\setminus S|$.

\Cref{alg:DoubleGreedy} presents our \textsc{FastDoubleGreedy} based on this idea.
In each $i$th step, the first and second calls to \textsc{UpdateRow}, defined in \cref{alg:LazyMapInf}, fill $V[i, S]$ and $W[i, [i]\setminus S]$, respectively.
Hence, \cref{prop:rel_bet_Gre_Chol} and \cref{lem:DoubleGreedyforLogDet} imply $2\log d_i = f_i(S)$ and $2\log e_i = g_i([i]\setminus S) = -f_i(T)$.
Thus, \cref{alg:DoubleGreedy} exactly mimics the behavior of \textsc{DoubleGreedy} while efficiently computing marginal gains via incremental updates of Cholesky factors.
In \textsc{UpdateRow} defined in \cref{alg:LazyMapInf}, since $L$ and $L^{-1}$ are already computed, each $V_{i, j_t}$ is calculated in $\Ord(n)$ time; therefore, a single call to \textsc{UpdateRow} takes $\Ord(n^2)$ time.
Since \textsc{UpdateRow} is called $2n$ times, \cref{line:double-init,line:double-for,line:double-updaterow1,line:double-updaterow2,line:double-compute-ab,line:double-add-S} construct a solution in $\Ord(n^3)$ time.
In total, \cref{alg:DoubleGreedy} runs in $\Ord(n^{\omega-1}d + n^3)$ time.

\section{Experiments}\label{section:experiments}
We evaluate the effectiveness of our acceleration techniques on synthetic and real-world datasets.
\Cref{subsection:experiment-greedy} examines speed-ups of \textsc{Greedy} for cardinality-constrained DPP MAP inference, and \cref{subsection:experiment-double-greedy} focuses on \textsc{DoubleGreedy} for the unconstrained setting.
We present experimental results on \textsc{RandomGreedy}, \textsc{StochasticGreedy}, and \textsc{InterlaceGreedy} in \cref{app-section:experiments-variants}.

The algorithms are implemented in C++ with library Eigen~3.4.0 for matrix computations.
Experiments are conducted using a compiler GCC~10.2.0 on a computer with \SI{3.8}{GHz} Intel Xeon Gold CPU and \SI{800}{GB} RAM.

We use synthetic and two real-world datasets, Netflix Prize~\cite{Bennett2007-gr} and MovieLens~\cite{Harper2015-mm}.
Each dataset provides a matrix $B \in \R^{d \times n}$ consisting of column vectors ${\bm \phi_1}, {\bm \phi_2}, \dots, {\bm \phi_n} \in \R^d$ of $n$ items, which defines an $n \times n$ kernel matrix $L = B^\top B$.
Below we explain the item vectors of each dataset.

\paragraph{Synthetic datasets.}
We use the setting of~\cite{Gillenwater2012-uw}.
Each entry of ${\bm \phi}_i \in \R^d$ is independently drawn from the standard normal distribution, $\phi_{ij}\sim\mathcal{N}(0,1)$.
As a result, the kernel matrix $L$ conforms to a Wishart distribution with $n$ degrees of freedom and an identity covariance matrix, i.e., $L\sim\mathcal{W}(I,n,n)$.
We consider various $n$ values in the experiments below, and we always set the vector length $d$ to $n$.

\paragraph{Real-world datasets.}
Both Netflix Prize and MovieLens datasets contain users' ratings of movies from one to five, where we regard a movie as an item.
Following \citep{Chen2018-aa}, we binarize the ratings based on whether it is greater than or equal to four.
After that, we eliminate movies that result in all-zero vectors and users who result in all-zero ratings since those are redundant.
Consequently, the Netflix Prize dataset has $n = 17770$ movies and $d = 478615$ users with $56919190$ ratings; the MovieLens dataset has $n = 40858$ movies and $d = 162342$ users with $12452811$ ratings.

\paragraph{$B$- and $L$-input settings.}
It is important to care about whether $L = B^\top B$ is computed in advance or not.
In practice, a matrix $B \in \R^{d \times n}$ of item vectors is often given.
Then, computing $L = B^\top B$ in advance takes $\Ord(\min \set{n^{\omega-1}d, n^2d^{\omega-2}})$ time, which we should avoid when $n$ is large since the running time of \textsc{LazyFastGreedy} (and \textsc{FastGreedy}) increases only linearly in $n$.
On the other hand, we are sometimes given a pre-computed kernel matrix $L$, and we can access $L_{i,j}$ in $\Ord(1)$ time.

We below consider both settings, called $B$- and $L$-input settings, respectively.
In the $L$-input setting, we exclude the time to compute $L = B^\top B$ from consideration.
Under this condition, \textsc{FastGreedy} takes $\Ord(k^2n)$ time and \textsc{LazyFastGreedy} does $\Ord(n + U(k + \log n))$, where the first $\Ord(n)$ term is for constructing a priority queue.
By similar reasoning to that in the proof of \cref{thm:LazyGreedyCorrectness_Complexity}, it runs in $\Ord(n + k^3 + k\log n)$ and $\Ord(kn(k + \log n))$ time if the lazy update works best and worst, respectively.

\subsection{Greedy algorithm for cardinality-constrained DPP MAP inference}\label{subsection:experiment-greedy}

\begin{figure}[tb]
  \begin{minipage}[b]{0.333\textwidth}
    \centering
    \includegraphics[width=\linewidth]{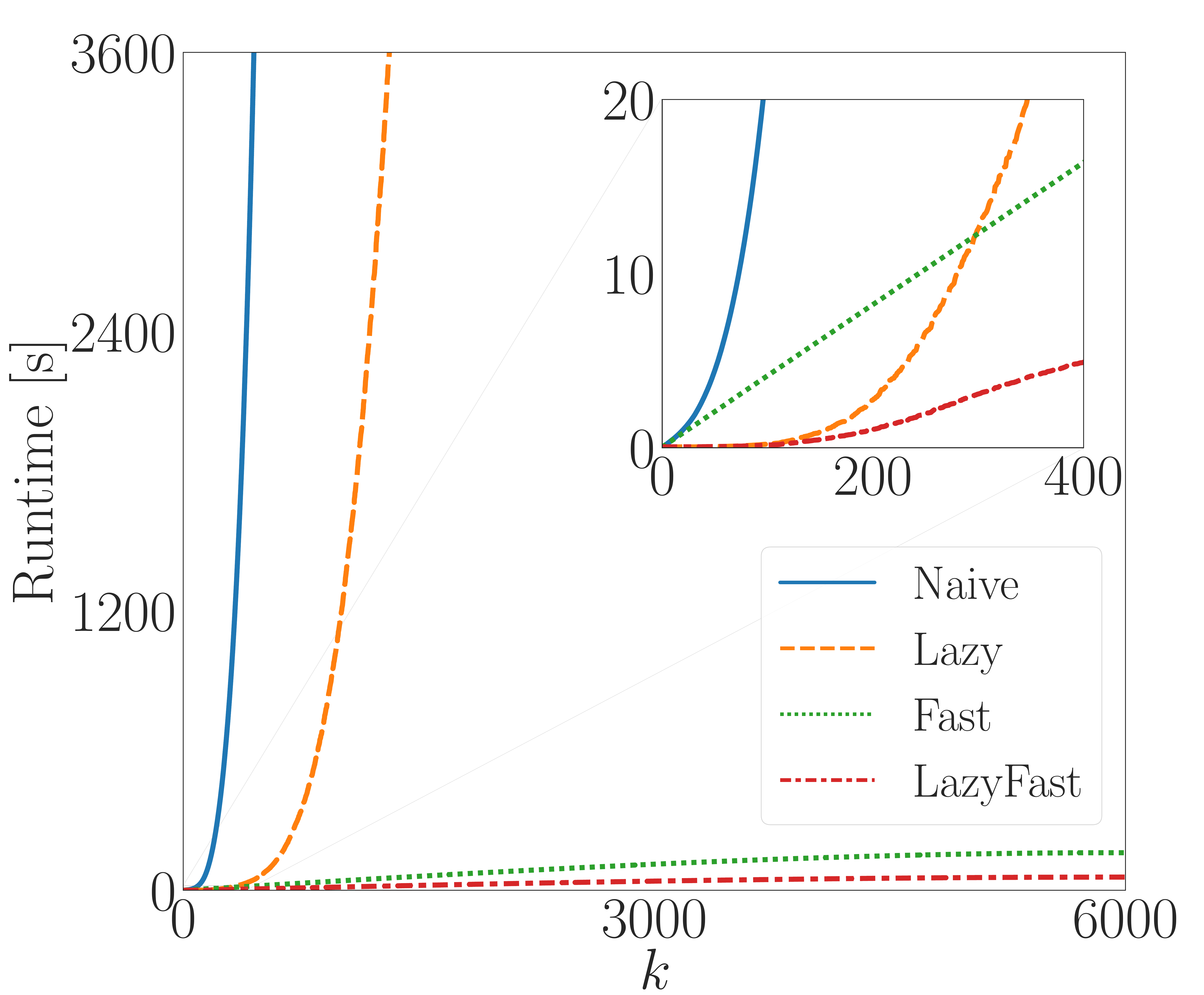}
    \subcaption{$n=6000$, $B$-input, Runtime}\label{subfig:synth-n-b-time}
  \end{minipage}%
  \begin{minipage}[b]{0.333\textwidth}
    \centering
    \includegraphics[width=\linewidth]{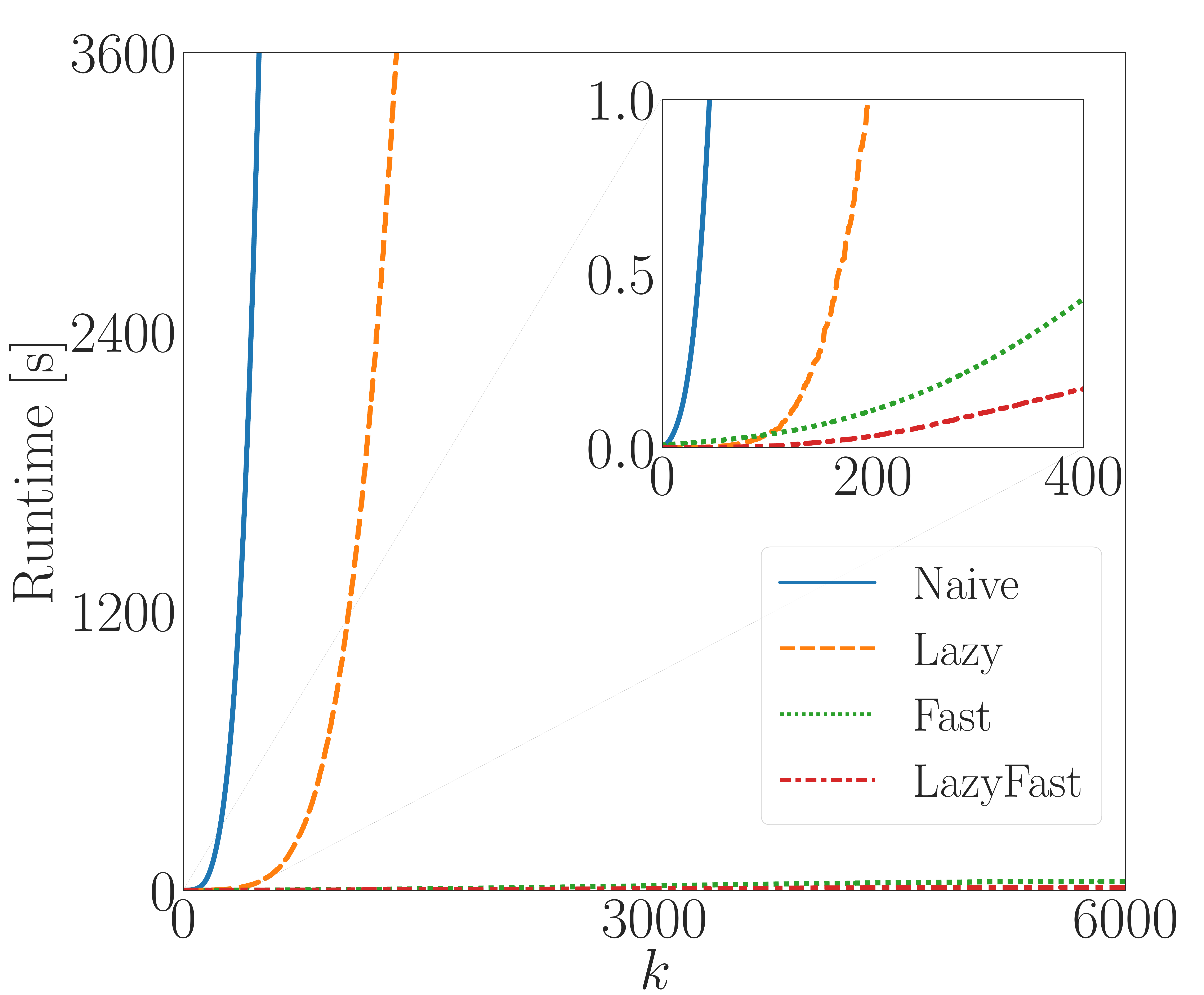}
    \subcaption{$n=6000$, $L$-input, Runtime}\label{subfig:synth-n-l-time}
  \end{minipage}%
  \begin{minipage}[b]{0.333\textwidth}
    \centering
    \includegraphics[width=\linewidth]{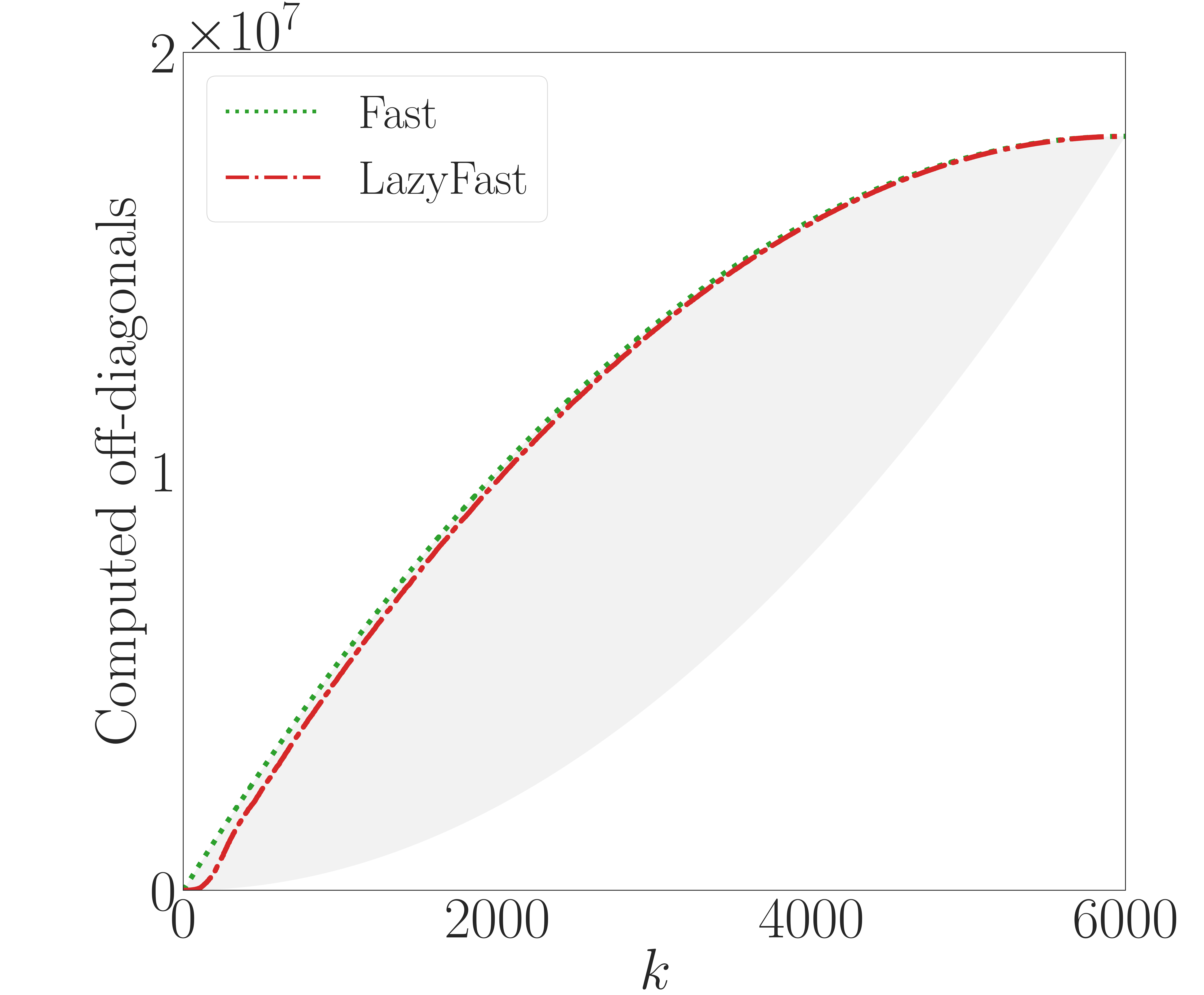}
    \subcaption{$n=6000$, Off-diagonals}\label{subfig:synth-n-offdiag}
  \end{minipage}
  \begin{minipage}[b]{0.333\textwidth}
    \vspace{10pt}
    \centering
    \includegraphics[width=\linewidth]{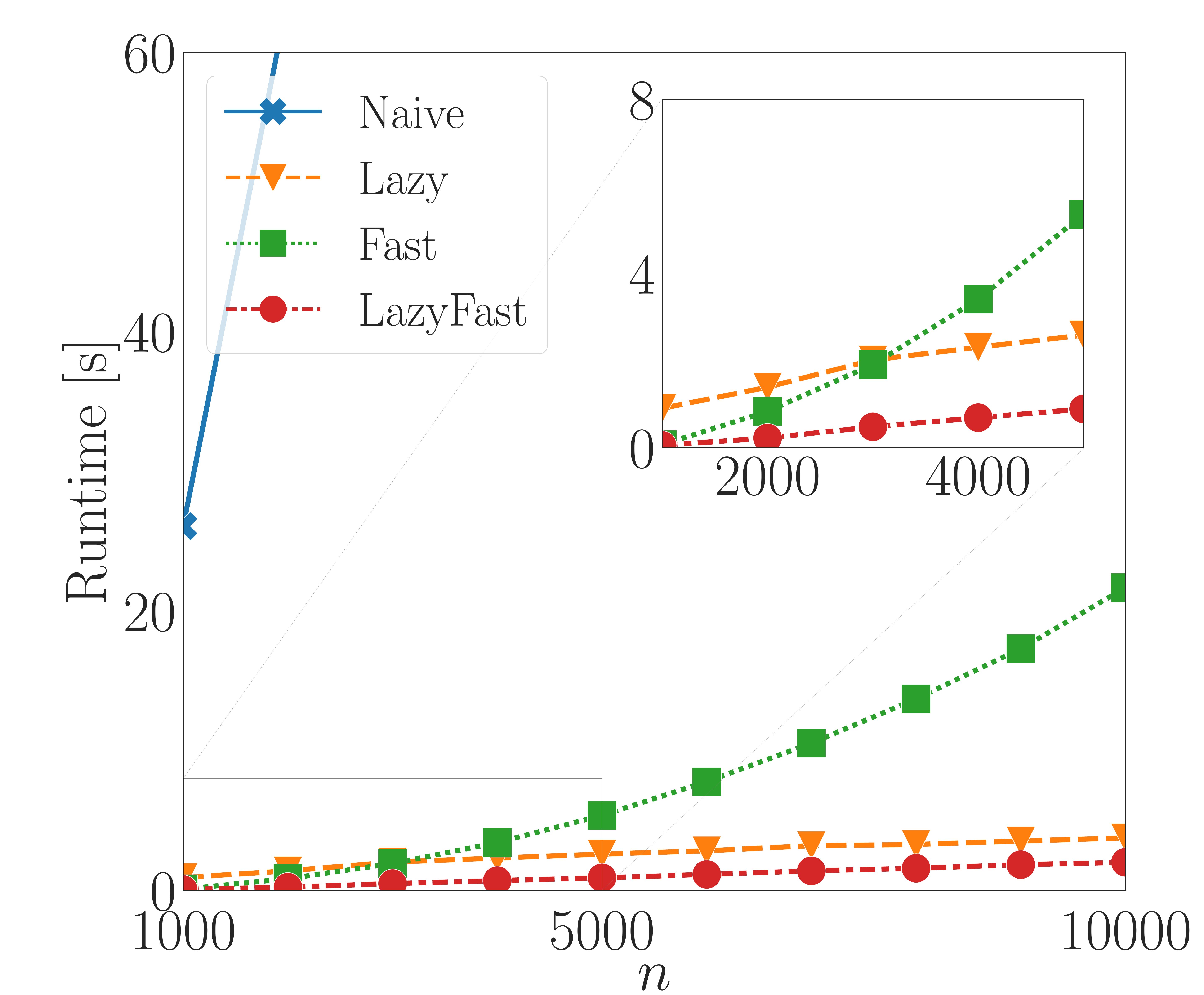}
    \subcaption{$k=200$, $B$-input, Runtime}\label{subfig:synth-k-b-time}
  \end{minipage}%
  \begin{minipage}[b]{0.333\textwidth}
    \vspace{10pt}
    \centering
    \includegraphics[width=\linewidth]{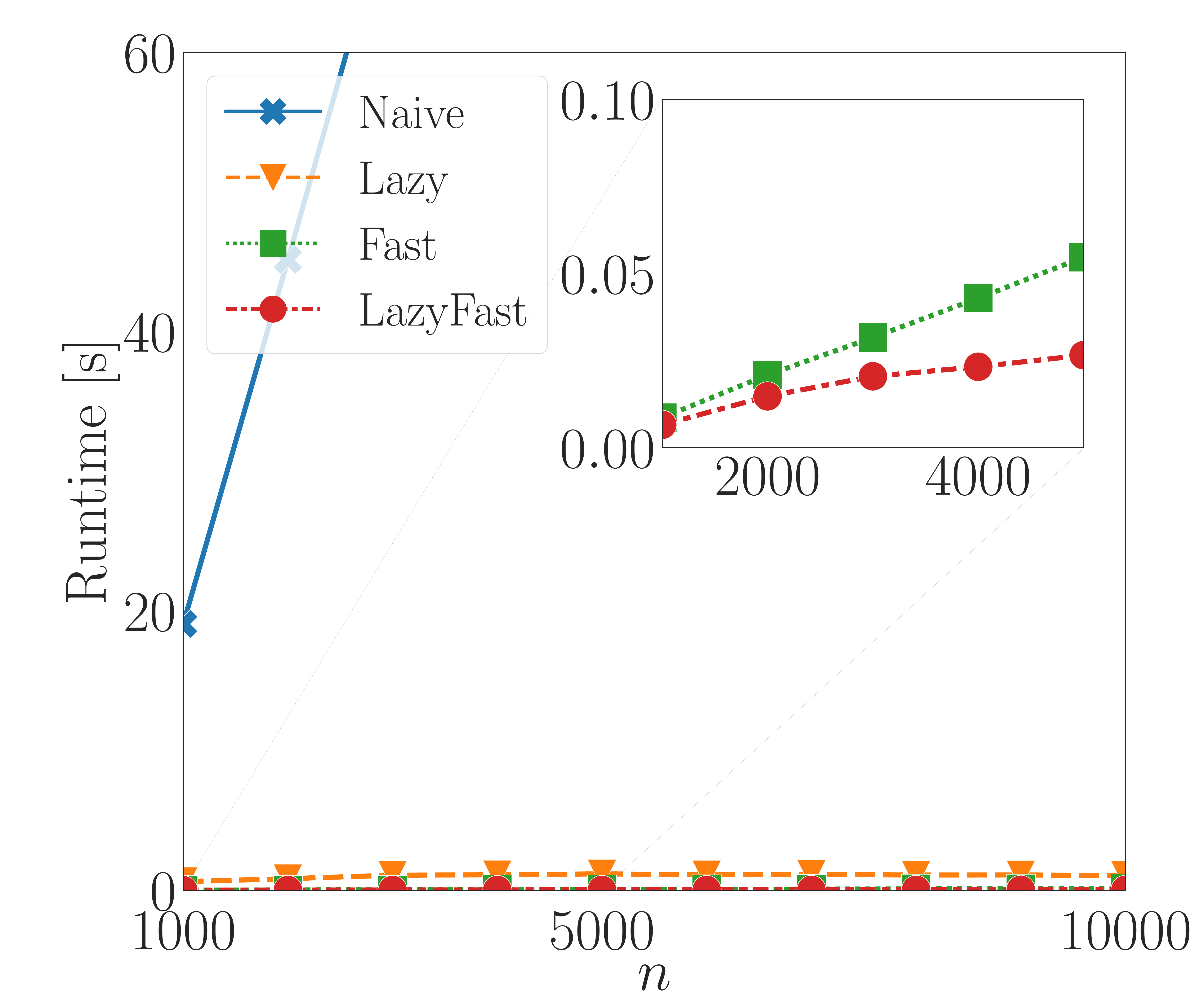}
    \subcaption{$k=200$, $L$-input, Runtime}\label{subfig:synth-k-l-time}
  \end{minipage}%
  \begin{minipage}[b]{0.333\textwidth}
    \vspace{10pt}
    \centering
    \includegraphics[width=\linewidth]{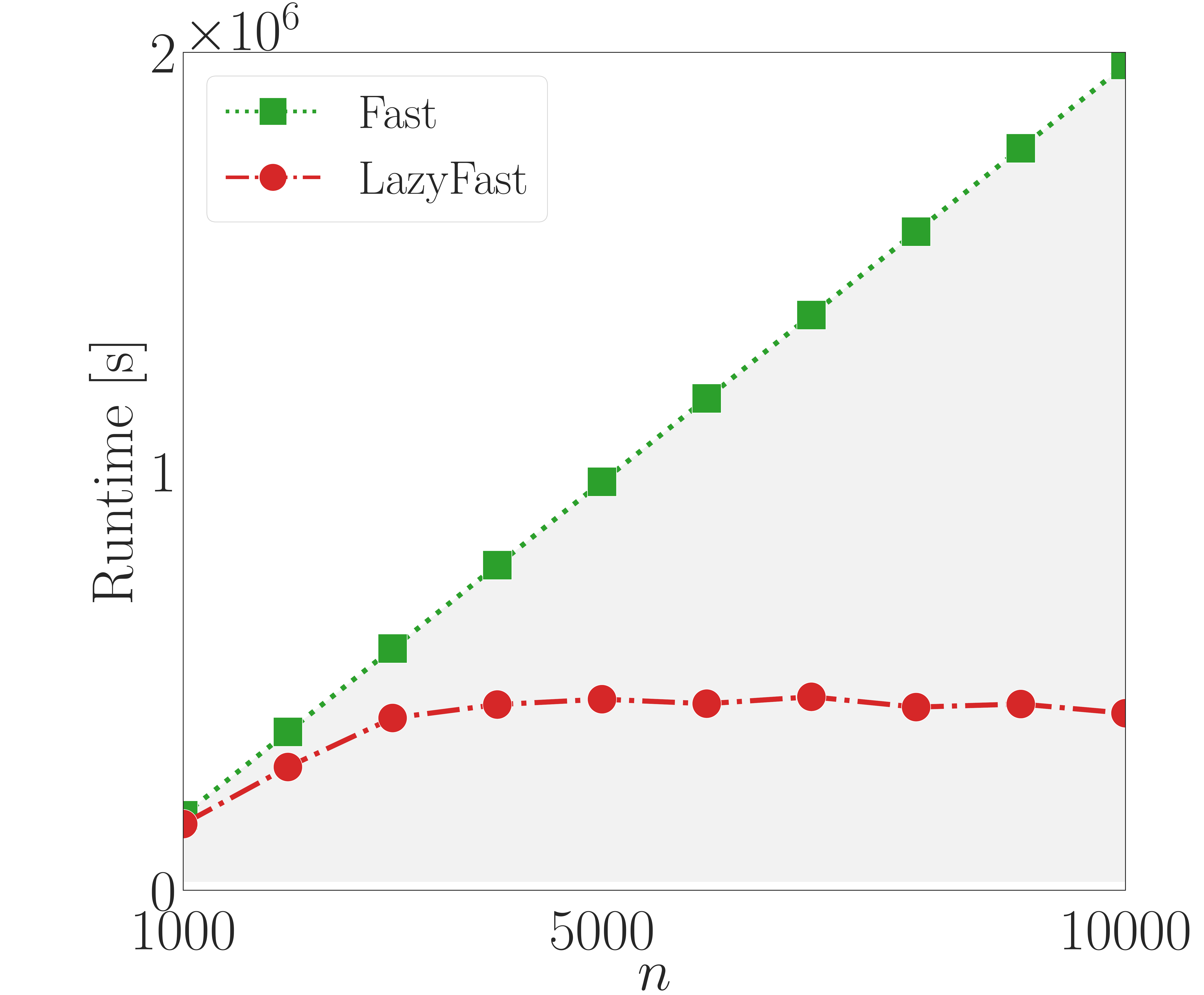}
    \subcaption{$k=200$, Off-diagonals}\label{subfig:synth-k-offdiag}
  \end{minipage}
  \caption{Results on synthetic datasets.
  In the four runtime figures, enlarged views of lower left parts are shown for visibility.
  In \cref{subfig:synth-n-offdiag,subfig:synth-k-offdiag}, the gray band indicates the range of the possible number of computed off-diagonals: $\left[k(k-1)/2, (k-1)(n-k/2)\right]$.}\label{fig:synth-greedy}
\end{figure}

\begin{figure}[tb]
  \begin{minipage}[b]{.333\linewidth}
    \centering
    \includegraphics[width=\textwidth]{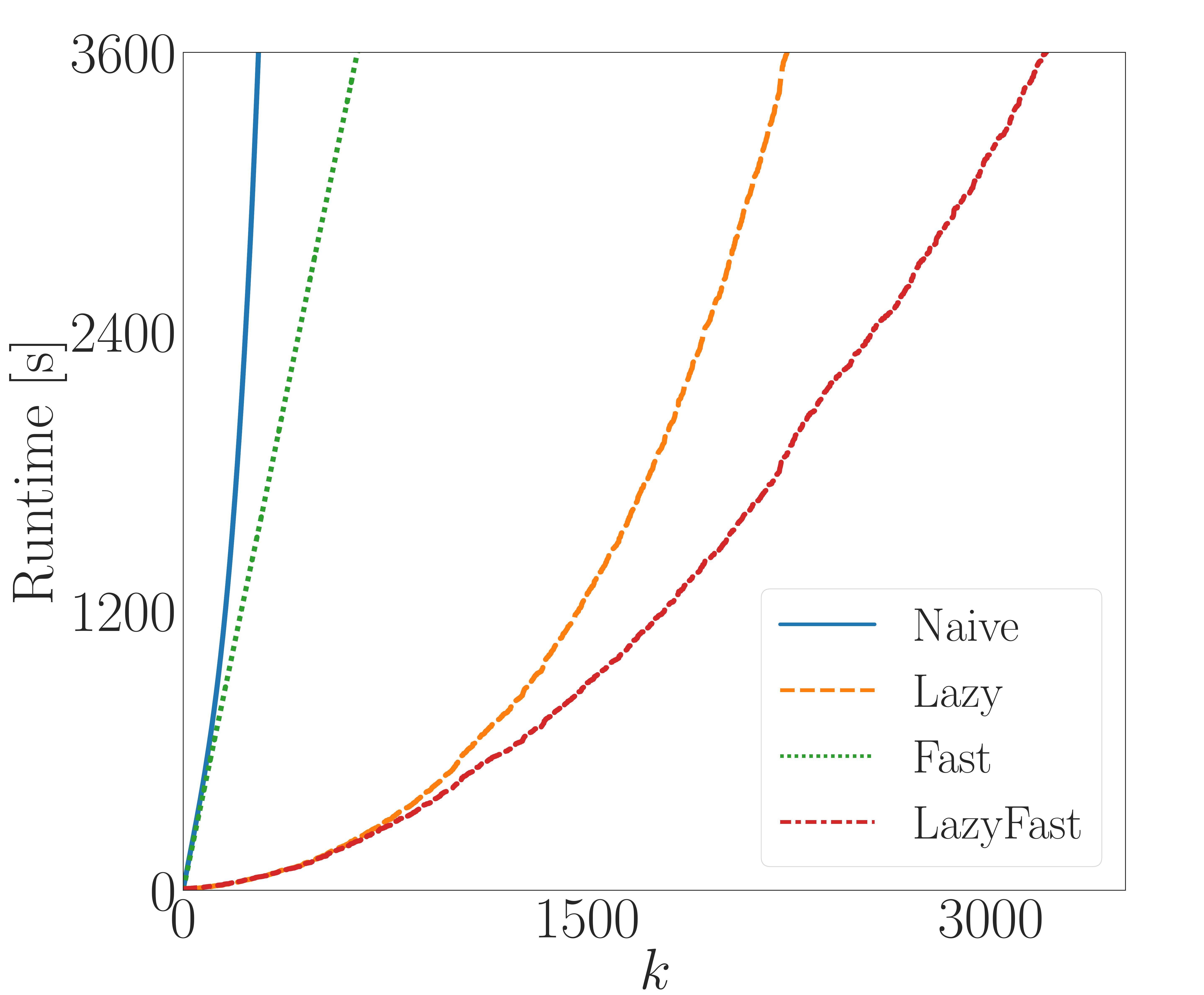}
    \subcaption{Netflix, $B$-input, Runtime}
  \end{minipage}%
  \begin{minipage}[b]{.333\linewidth}
    \centering
    \includegraphics[width=\textwidth]{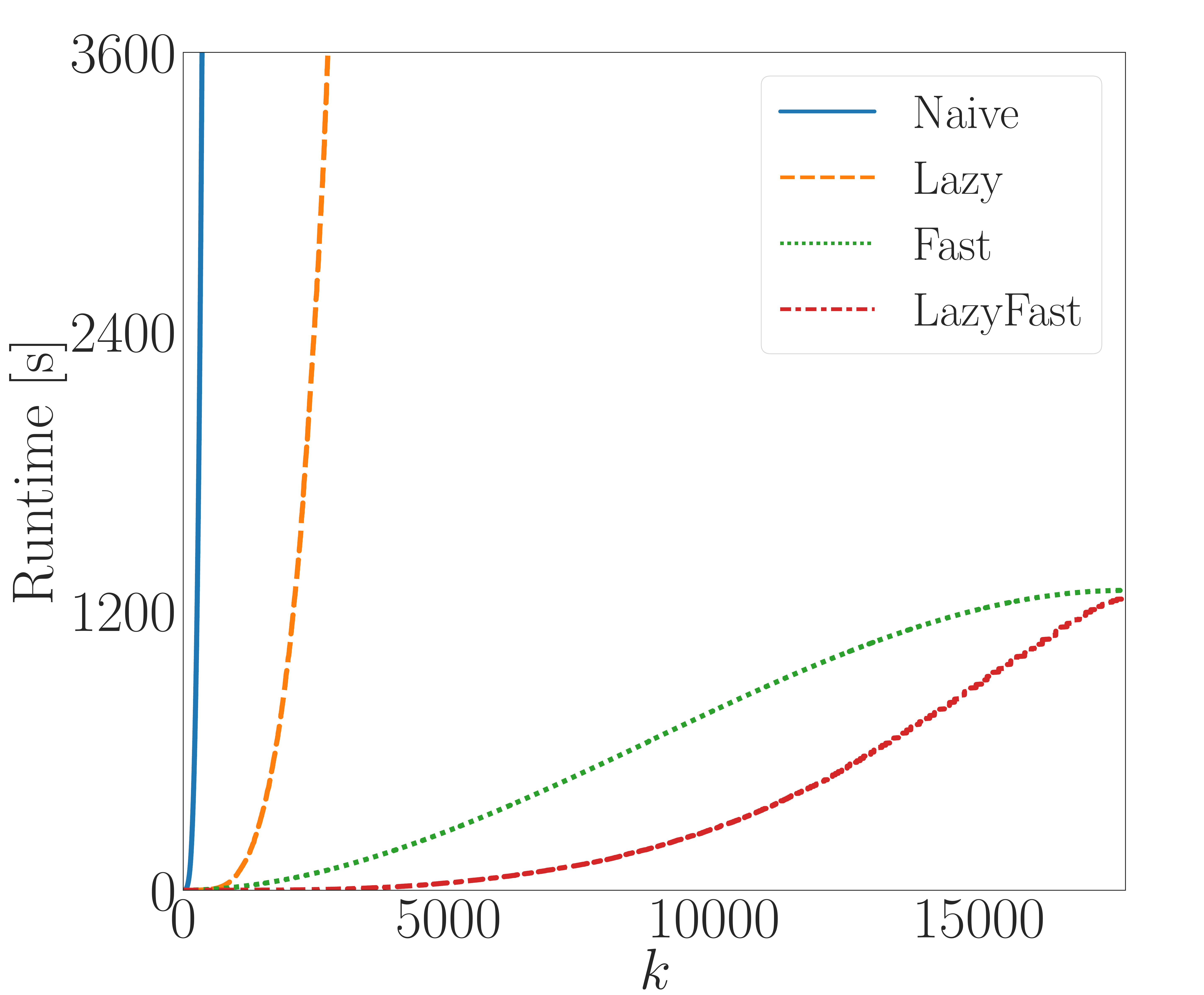}
    \subcaption{Netflix, $L$-input, Runtime}
  \end{minipage}%
  \begin{minipage}[b]{.333\linewidth}
    \centering
    \includegraphics[width=\textwidth]{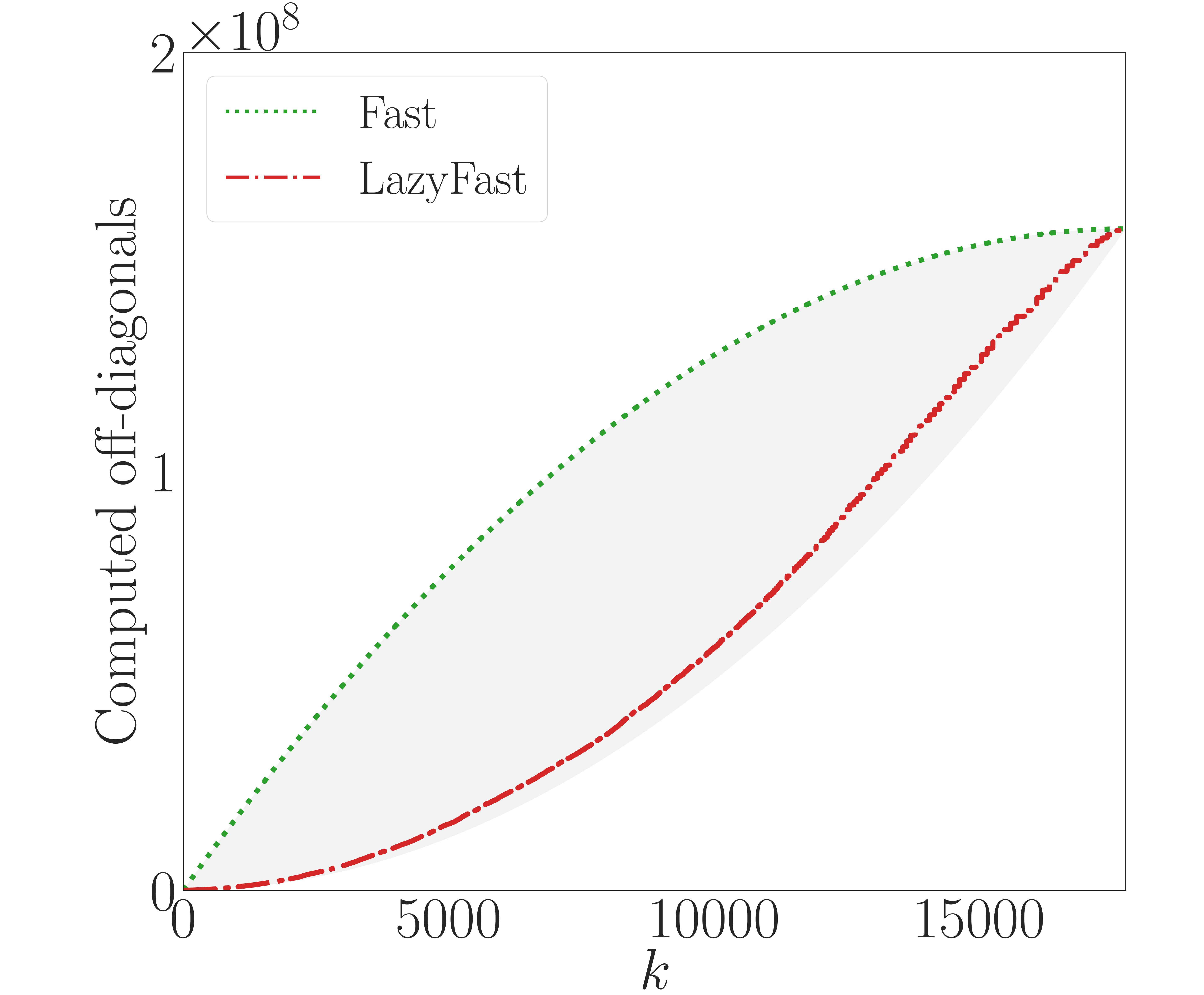}
    \subcaption{Netflix, Off-diagonals}\label{subfig:netflix-offdiag}
  \end{minipage}
  \begin{minipage}[b]{.333\linewidth}
    \vspace{10pt}
    \centering
    \includegraphics[width=\textwidth]{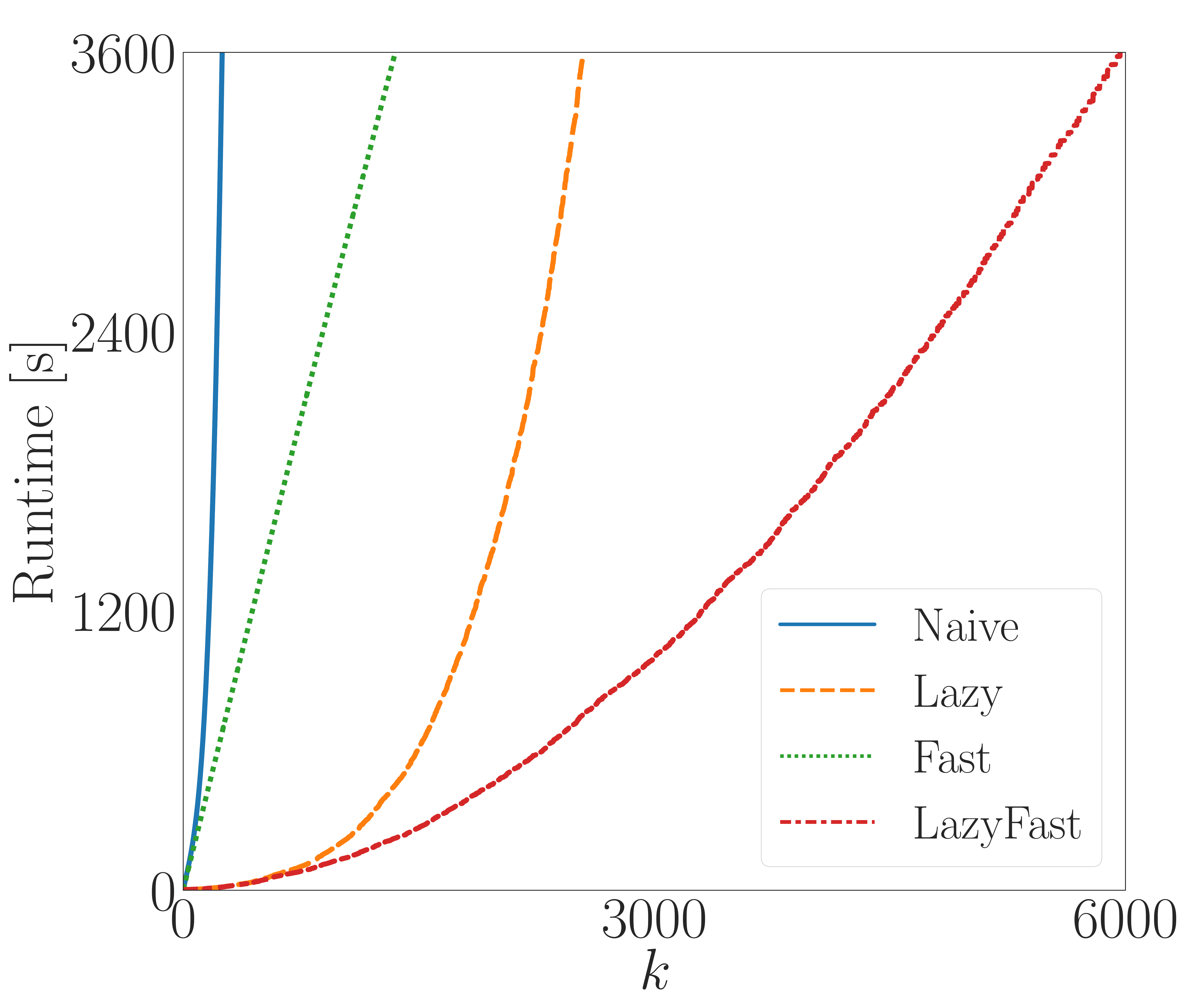}
    \subcaption{MovieLens, $B$-input, Runtime}
  \end{minipage}%
  \begin{minipage}[b]{.333\linewidth}
    \vspace{10pt}
    \centering
    \includegraphics[width=\textwidth]{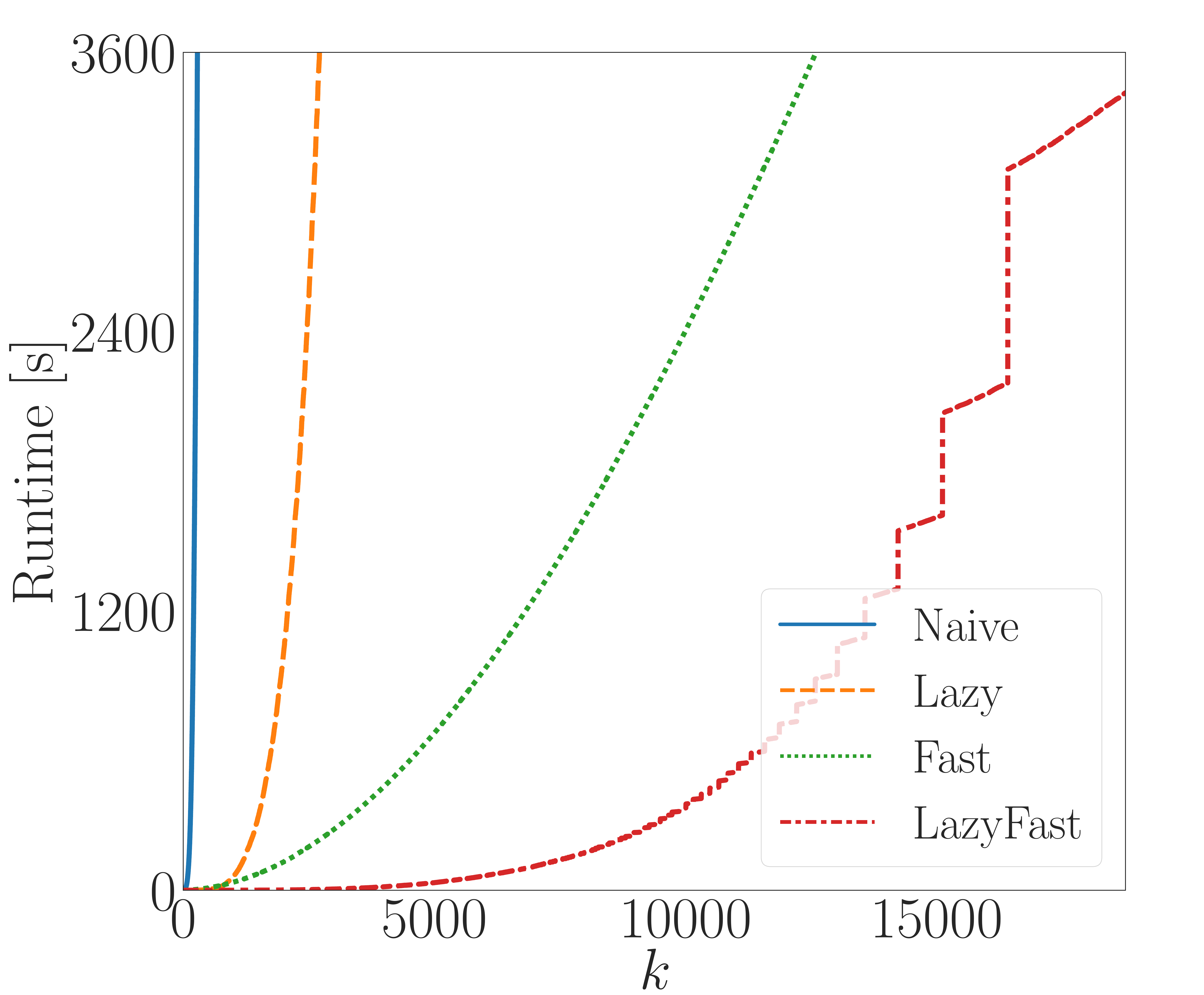}
    \subcaption{MovieLens, $L$-input, Runtime}
  \end{minipage}%
  \begin{minipage}[b]{.333\linewidth}
    \vspace{10pt}
    \centering
    \includegraphics[width=\textwidth]{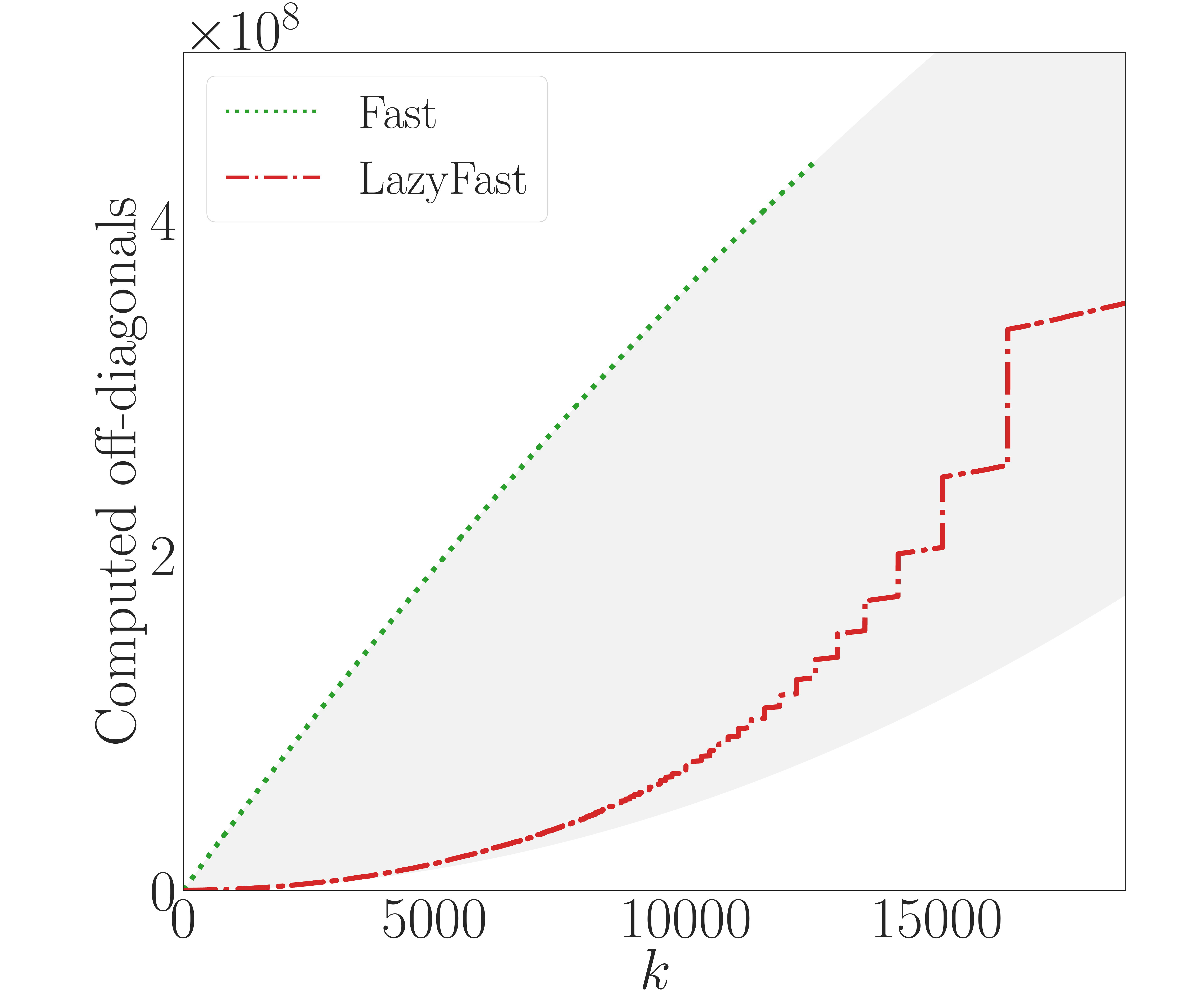}
    \subcaption{MovieLens, Off-diagonals}\label{subfig:movielens-offdiag}
  \end{minipage}
  \caption{Results on real-world datasets. In \cref{subfig:netflix-offdiag,subfig:movielens-offdiag}, the gray band indicates the range of the possible number of computed off-diagonals: $\left[k(k-1)/2, (k-1)(n-k/2)\right]$.}\label{fig:real-greedy}
\end{figure}

We compare the running time of \textsc{Greedy}, \textsc{LazyGreedy}, \textsc{FastGreedy}, and \textsc{LazyFastGreedy}, which we here call Naive, Lazy, Fast, and LazyFast, respectively, for short.
As regards synthetic datasets, we consider two settings that fix either $n$ or $k$:
(i) $n = 6000$ and $k = 1,2,\dots,n$,
and
(ii) $k = 200$ and $n = 1000, 2000,\dots,10000$.
We set the timeout periods of (i) and (ii) to $3600$ and $60$ seconds, respectively.
Regarding real-world datasets, $n$ is fixed as explained above and we increase $k = 1,2,\dots,n$.
We set the timeout period to $3600$ seconds.
With the Netflix (MovieLens) dataset, the objective value peaked with $k = 17762$ ($k = 18763$), and thus we stopped increasing $k$ after that.

\Cref{fig:synth-greedy,fig:real-greedy} present the results on synthetic and real-world datasets, respectively.
The curves in the runtime figures represent that faster algorithms can return greedy solutions to instances with larger $k$ or $n$ values within the timeout periods.
The rightmost figures present the number of off-diagonals of Cholesky factors computed by Fast and LazyFast.
As explained in \cref{subsection:lazyfastgreedy}, while Fast always computes all the off-diagonals of $V[[n], S]$, LazyFast does not due to the lazy update.

LazyFast was the fastest in all the settings.
In particular, in the synthetic~(ii) and real-world settings, LazyFast computed fewer off-diagonals than Fast, thus running the fastest.
In the synthetic setting~(i), although LazyFast computed almost all off-diagonals in $V[[n], S]$, it was still faster than Fast.
For example, for the $L$-input setting with $n = k = 6000$, Fast and LazyFast took $37.4$ and $13.7$ seconds, respectively.
This unexpected speed-up is caused by the cache efficiency of LazyFast.
Specifically, every call to \textsc{UpdateRow} computes off-diagonals from $V_{i, {j_{u_i+1}}}$ to $V_{i, j_{|S|}}$ by accessing entries only in $V[\set{i, {j_{u_i+1},\dots,j_{|S|}}}, \step{S}{|S|-1}]$.
This process virtually creates blocks in a Cholesky factor, enabling the cache-efficient computation of off-diagonals via blocking (see, e.g.,~\citep[Section~2.3]{Hennessy2011-ev}).
In fact, the cache miss rates of Fast and LazyFast for the above example were $71.3\%$ and $3.3\%$, respectively.

Another interesting observation in the synthetic (ii) and real-world settings is that while Fast was often faster than Lazy in the $L$-input setting, the opposite occurred in the $B$-input setting.
This is because computing an off-diagonal in the $B$-input setting is costly relative to the $L$-input setting.
As a result, avoiding the redundant computation of off-diagonals by the lazy update tends to be more effective than computing marginal gains efficiently via the Cholesky factorization.

\subsection{Double greedy algorithm for unconstrained DPP MAP inference}\label{subsection:experiment-double-greedy}

\begin{table}[tb]
  \caption{Running time [s] of \textsc{DoubleGreedy} and \textsc{FastDoubleGreedy}}\label{table:experiment-double}
  \vspace{10pt}
  \centering
  \begin{tabular}{lrrrrrrr}\toprule
    \multicolumn{1}{c}{\multirow{2}{*}[-3pt]{Dataset}} & \multicolumn{3}{c}{\textsc{DoubleGreedy}} & \multicolumn{4}{c}{\textsc{FastDoubleGreedy}} \\\cmidrule(r){2-4}\cmidrule(r){5-8}
     & \multicolumn{1}{c}{Product} & \multicolumn{1}{c}{Greedy} & \multicolumn{1}{c}{\textbf{Total}} & \multicolumn{1}{c}{Product} & \multicolumn{1}{c}{Inverse} & \multicolumn{1}{c}{Greedy} & \multicolumn{1}{c}{\textbf{Total}} \\\midrule
    Synthetic              &    $30.0$ &     $36233.1$  &    $36263.1$  &    $30.0$ & $59.8$ & $34.6$ & $124.4$ \\
    Netflix Prize          & $20561.6$ & $> 86400.0$ & \multicolumn{1}{c}{---} & $20561.6$ & $1706.7$ & $916.2$ & $23184.5$ \\
    MovieLens              & $37829.5$ & $> 86400.0$ & \multicolumn{1}{c}{---} & $37829.5$ & $21999.2$ & $6337.3$ & $66166.0$ \\\bottomrule
  \end{tabular}
\end{table}

We compare naive \textsc{DoubleGreedy} and our \textsc{FastDoubleGreedy} by applying them to unconstrained DPP MAP inference on three datasets: synthetic ($n = 6000$), Netflix Prize, and MovieLens.
Since kernel matrices $L$ in the real-world datasets are singular, we use kernel matrices computed as $L = 0.9 B^\top B + 0.1 I$ in this section, ensuring that the resulting matrices $L$ are positive definite.
In this section, we set the timeout period to one day ($86400$ seconds).

\Cref{table:experiment-double} presents the results.
Note that both algorithms require $L$ to be computed in advance.
Therefore, we measured the time of computing $B^\top B$ (Product) separately from the time of constructing solutions (Greedy).
Our \textsc{FastDoubleGreedy} additionally requires $L^{-1}$ to be computed in advance for accelerating the solution construction; therefore, we also measured the time of computing $L^{-1}$ (Inverse) separately.
As in \cref{table:experiment-double}, naive \textsc{DoubleGreedy} took so long for solution construction that it failed to return solutions to real-world instances in one day.
By contrast, our \textsc{FastDoubleGreedy} constructed solutions far faster and succeeded in returning solutions to all the instances.

Also, the computation of $B^\top B$ (Product) took a considerably long time relative to the running time of \textsc{LazyFastGreedy} in the previous section.
Therefore, as mentioned above, when a matrix $B$ of item vectors is given and our goal is to obtain a greedy solution for small $k = \mathrm{o}(n)$, we should avoid computing the kernel matrix $L = B^\top B$ in advance.

%Also, computing $B^\top B$ for $n = d = 6000$ took $30.0$ seconds, which was longer than the running time of \textsc{LazyFastGreedy} for $n = 6000$ and $k=2000$ in the $B$-input setting (\cref{subfig:synth-n-b-time}).
%Therefore, when a matrix $B$ of item vectors is given and our goal is to obtain greedy solutions for small $k$, we should avoid computing the kernel matrix $L = B^\top B$ in advance, as mentioned above.

\subsection*{Acknowledgements}
This work was supported by JST ERATO Grant Number JPMJER1903 and JSPS KAKENHI Grant Number JP22K17853.

\bibliographystyle{abbrvnat}
\bibliography{ref}

\begin{thebibliography}{42}
\providecommand{\natexlab}[1]{#1}
\providecommand{\url}[1]{\texttt{#1}}
\expandafter\ifx\csname urlstyle\endcsname\relax
  \providecommand{\doi}[1]{doi: #1}\else
  \providecommand{\doi}{doi: \begingroup \urlstyle{rm}\Url}\fi

\bibitem[Anari and Derezi{\'n}ski(2020)]{Anari2020-cj}
N.~Anari and M.~Derezi{\'n}ski.
\newblock Isotropy and log-concave polynomials: Accelerated sampling and
  high-precision counting of matroid bases.
\newblock In \emph{Proceedings of the 2020 IEEE 61st Annual Symposium on
  Foundations of Computer Science (FOCS 2020)}, pages 1331--1344. IEEE, 2020.

\bibitem[Anari et~al.(2016)Anari, Oveis~Gharan, and Rezaei]{Anari2016-hw}
N.~Anari, S.~Oveis~Gharan, and A.~Rezaei.
\newblock {Monte Carlo Markov} chain algorithms for sampling strongly
  {Rayleigh} distributions and determinantal point processes.
\newblock In \emph{Proceedings of the 29th Annual Conference on Learning Theory
  (COLT 2016)}, volume~49, pages 103--115. PMLR, 2016.

\bibitem[Balkanski et~al.(2018)Balkanski, Breuer, and Singer]{Balkanski2018-dq}
E.~Balkanski, A.~Breuer, and Y.~Singer.
\newblock Non-monotone submodular maximization in exponentially fewer
  iterations.
\newblock In \emph{Advances in Neural Information Processing Systems (NeurIPS
  2018)}, volume~31. Curran Associates, Inc., 2018.

\bibitem[Balkanski et~al.(2019)Balkanski, Rubinstein, and
  Singer]{Balkanski2019-mj}
E.~Balkanski, A.~Rubinstein, and Y.~Singer.
\newblock An exponential speedup in parallel running time for submodular
  maximization without loss in approximation.
\newblock In \emph{Proceedings of the 2019 Annual ACM-SIAM Symposium on
  Discrete Algorithms (SODA 2019)}, pages 283--302. SIAM, 2019.

\bibitem[Bennett and Lanning(2007)]{Bennett2007-gr}
J.~Bennett and S.~Lanning.
\newblock The {Netflix} prize.
\newblock In \emph{Proceedings of KDD Cup and Workshop}, 2007.
\newblock URL
  \url{https://www.kaggle.com/datasets/netflix-inc/netflix-prize-data}.

\bibitem[Bhaskara et~al.(2020)Bhaskara, Karbasi, Lattanzi, and
  Zadimoghaddam]{Bhaskara2020-cg}
A.~Bhaskara, A.~Karbasi, S.~Lattanzi, and M.~Zadimoghaddam.
\newblock Online {MAP} inference of determinantal point processes.
\newblock In \emph{Advances in Neural Information Processing Systems (NeurIPS
  2020)}, volume~33, pages 3419--3429. Curran Associates, Inc., 2020.

\bibitem[Bini and Pan(1994)]{Bini1994-qa}
D.~Bini and V.~Y. Pan.
\newblock \emph{Polynomial and Matrix Computations}, volume~1.
\newblock Birkh{\"a}user Boston, 1994.

\bibitem[Breuer et~al.(2020)Breuer, Balkanski, and Singer]{Breuer2020-do}
A.~Breuer, E.~Balkanski, and Y.~Singer.
\newblock The {FAST} algorithm for submodular maximization.
\newblock In \emph{Proceedings of the 37th International Conference on Machine
  Learning (ICML 2020)}, volume 119, pages 1134--1143. PMLR, 2020.

\bibitem[Brualdi and Schneider(1983)]{Brualdi1983-nl}
R.~A. Brualdi and H.~Schneider.
\newblock Determinantal identities: {Gauss, Schur, Cauchy, Sylvester,
  Kronecker, Jacobi, Binet, Laplace, Muir, and Cayley}.
\newblock \emph{Linear Algebra and Its Applications}, 52-53:\penalty0 769--791,
  1983.

\bibitem[Buchbinder et~al.(2014)Buchbinder, Feldman, Naor, and
  Schwartz]{Buchbinder2014-il}
N.~Buchbinder, M.~Feldman, J.~S. Naor, and R.~Schwartz.
\newblock Submodular maximization with cardinality constraints.
\newblock In \emph{Proceedings of the 2014 Annual ACM-SIAM Symposium on
  Discrete Algorithms (SODA 2014)}, pages 1433--1452. SIAM, 2014.

\bibitem[Buchbinder et~al.(2015)Buchbinder, Feldman, Naor, and
  Schwartz]{Buchbinder2015-rc}
N.~Buchbinder, M.~Feldman, J.~S. Naor, and R.~Schwartz.
\newblock A tight linear time ($1/2$)-approximation for unconstrained
  submodular maximization.
\newblock \emph{SIAM Journal on Computing}, 44\penalty0 (5):\penalty0
  1384--1402, 2015.

\bibitem[Calandriello et~al.(2020)Calandriello, Derezinski, and
  Valko]{Calandriello2020-qh}
D.~Calandriello, M.~Derezinski, and M.~Valko.
\newblock Sampling from a $k$-{DPP} without looking at all items.
\newblock In \emph{Advances in Neural Information Processing Systems (NeurIPS
  2020)}, volume~33, pages 6889--6899. Curran Associates, Inc., 2020.

\bibitem[Chen et~al.(2018)Chen, Zhang, and Zhou]{Chen2018-aa}
L.~Chen, G.~Zhang, and E.~Zhou.
\newblock Fast greedy {MAP} inference for determinantal point process to
  improve recommendation diversity.
\newblock In \emph{Advances in Neural Information Processing Systems (NeurIPS
  2018)}, volume~31. Curran Associates, Inc., 2018.

\bibitem[Chen et~al.(2021)Chen, Dey, and Kuhnle]{Chen2021-vs}
Y.~Chen, T.~Dey, and A.~Kuhnle.
\newblock Best of both worlds: Practical and theoretically optimal submodular
  maximization in parallel.
\newblock In \emph{Advances in Neural Information Processing Systems (NeurIPS
  2021)}, volume~34, pages 25528--25539. Curran Associates, Inc., 2021.

\bibitem[{\c C}ivril and Magdon-Ismail(2009)]{Civril2009-uj}
A.~{\c C}ivril and M.~Magdon-Ismail.
\newblock On selecting a maximum volume sub-matrix of a matrix and related
  problems.
\newblock \emph{Theoretical Computer Science}, 410\penalty0 (47):\penalty0
  4801--4811, 2009.

\bibitem[Derezinski et~al.(2019)Derezinski, Calandriello, and
  Valko]{Derezinski2019-zg}
M.~Derezinski, D.~Calandriello, and M.~Valko.
\newblock Exact sampling of determinantal point processes with sublinear time
  preprocessing.
\newblock In \emph{Advances in Neural Information Processing Systems (NeurIPS
  2019)}, volume~32. Curran Associates, Inc., 2019.

\bibitem[Ene and Nguyen(2020)]{Ene2020-gj}
A.~Ene and H.~Nguyen.
\newblock Parallel algorithm for non-monotone {DR}-submodular maximization.
\newblock In \emph{Proceedings of the 37th International Conference on Machine
  Learning (ICML 2020)}, volume 119, pages 2902--2911. PMLR, 2020.

\bibitem[Fahrbach et~al.(2019)Fahrbach, Mirrokni, and
  Zadimoghaddam]{Fahrbach2019-kz}
M.~Fahrbach, V.~Mirrokni, and M.~Zadimoghaddam.
\newblock Non-monotone submodular maximization with nearly optimal adaptivity
  and query complexity.
\newblock In \emph{Proceedings of the 36th International Conference on Machine
  Learning (ICML 2019)}, volume~97, pages 1833--1842. PMLR, 2019.

\bibitem[Fan(1968)]{Fan1968-pg}
K.~Fan.
\newblock An inequality for subadditive functions on a distributive lattice,
  with application to determinantal inequalities.
\newblock \emph{Linear Algebra and Its Applications}, 1\penalty0 (1):\penalty0
  33--38, 1968.

\bibitem[Gillenwater et~al.(2012)Gillenwater, Kulesza, and
  Taskar]{Gillenwater2012-uw}
J.~Gillenwater, A.~Kulesza, and B.~Taskar.
\newblock Near-optimal {MAP} inference for determinantal point processes.
\newblock In \emph{Advances in Neural Information Processing Systems (NeurIPS
  2012)}, volume~25. Curran Associates, Inc., 2012.

\bibitem[Gillenwater et~al.(2019)Gillenwater, Kulesza, Mariet, and
  Vassilvtiskii]{Gillenwater2019-gq}
J.~Gillenwater, A.~Kulesza, Z.~Mariet, and S.~Vassilvtiskii.
\newblock A tree-based method for fast repeated sampling of determinantal point
  processes.
\newblock In \emph{Proceedings of the 36th International Conference on Machine
  Learning (ICML 2019)}, volume~97, pages 2260--2268. PMLR, 2019.

\bibitem[Han and Gillenwater(2020)]{Han2020-ez}
I.~Han and J.~Gillenwater.
\newblock {MAP} inference for customized determinantal point processes via
  maximum inner product search.
\newblock In \emph{Proceedings of the 23rd International Conference on
  Artificial Intelligence and Statistics (AISTATS 2020)}, volume 108, pages
  2797--2807. PMLR, 2020.

\bibitem[Han et~al.(2017)Han, Kambadur, Park, and Shin]{Han2017-bv}
I.~Han, P.~Kambadur, K.~Park, and J.~Shin.
\newblock Faster greedy {MAP} inference for determinantal point processes.
\newblock In \emph{Proceedings of the 34th International Conference on Machine
  Learning (ICML 2017)}, volume~70, pages 1384--1393. PMLR, 2017.

\bibitem[Harper and Konstan(2015)]{Harper2015-mm}
F.~M. Harper and J.~A. Konstan.
\newblock The {MovieLens} datasets: History and context.
\newblock \emph{ACM Transactions on Interactive Intelligent Systems},
  5\penalty0 (4):\penalty0 1--19, 2015.
\newblock URL \url{https://grouplens.org/datasets/movielens/25m/}.

\bibitem[Hennessy and Patterson(2011)]{Hennessy2011-ev}
J.~L. Hennessy and D.~A. Patterson.
\newblock \emph{Computer Architecture: A Quantitative Approach}.
\newblock Elsevier, 6th edition, 2011.

\bibitem[Ko et~al.(1995)Ko, Lee, and Queyranne]{Ko1995-yx}
C.-W. Ko, J.~Lee, and M.~Queyranne.
\newblock An exact algorithm for maximum entropy sampling.
\newblock \emph{Operations Research}, 43\penalty0 (4):\penalty0 684--691, 1995.

\bibitem[Kuhnle(2019)]{Kuhnle2019-io}
A.~Kuhnle.
\newblock Interlaced greedy algorithm for maximization of submodular functions
  in nearly linear time.
\newblock In \emph{Advances in Neural Information Processing Systems (NeurIPS
  2019)}, volume~32. Curran Associates, Inc., 2019.

\bibitem[Kuhnle(2021)]{Kuhnle2021-wb}
A.~Kuhnle.
\newblock Nearly linear-time, parallelizable algorithms for non-monotone
  submodular maximization.
\newblock In \emph{Proceedings of the 35th AAAI conference on artificial
  intelligence (AAAI 2021)}, volume~35, pages 8200--8208. AAAI Press, 2021.

\bibitem[Kulesza and Taskar(2011)]{Kulesza2011-pf}
A.~Kulesza and B.~Taskar.
\newblock $k$-{DPP}s: Fixed-size determinantal point processes.
\newblock In \emph{Proceedings of the 28th International Conference on Machine
  Learning (ICML 2011)}, pages 1193--1200. Omnipress, 2011.

\bibitem[Kulesza and Taskar(2012)]{Kulesza2012-er}
A.~Kulesza and B.~Taskar.
\newblock Determinantal point processes for machine learning.
\newblock \emph{Foundations and Trends\textregistered{} in Machine Learning},
  5\penalty0 (2--3):\penalty0 123--286, 2012.

\bibitem[Launay et~al.(2020)Launay, Galerne, and Desolneux]{Launay2020-af}
C.~Launay, B.~Galerne, and A.~Desolneux.
\newblock Exact sampling of determinantal point processes without
  eigendecomposition.
\newblock \emph{Journal of Applied Probability}, 57\penalty0 (4):\penalty0
  1198--1221, 2020.

\bibitem[Li et~al.(2016)Li, Sra, and Jegelka]{Li2016-ia}
C.~Li, S.~Sra, and S.~Jegelka.
\newblock Fast mixing {Markov} chains for strongly {Rayleigh} measures, {DPPs},
  and constrained sampling.
\newblock In \emph{Advances in Neural Information Processing Systems (NeurIPS
  2016)}, volume~29. Curran Associates, Inc., 2016.

\bibitem[Macchi(1975)]{Macchi1975-yh}
O.~Macchi.
\newblock The coincidence approach to stochastic point processes.
\newblock \emph{Advances in Applied Probability}, 7\penalty0 (1):\penalty0
  83--122, 1975.

\bibitem[Mahabadi et~al.(2019)Mahabadi, Indyk, Gharan, and
  Rezaei]{Mahabadi2019-bc}
S.~Mahabadi, P.~Indyk, S.~O. Gharan, and A.~Rezaei.
\newblock Composable core-sets for determinant maximization: A simple
  near-optimal algorithm.
\newblock In \emph{Proceedings of the 36th International Conference on Machine
  Learning (ICML 2019)}, volume~97, pages 4254--4263. PMLR, 2019.

\bibitem[Mahabadi et~al.(2020)Mahabadi, Razenshteyn, Woodruff, and
  Zhou]{Mahabadi2020-vh}
S.~Mahabadi, I.~Razenshteyn, D.~P. Woodruff, and S.~Zhou.
\newblock Non-adaptive adaptive sampling on turnstile streams.
\newblock In \emph{Proceedings of the 52nd Annual ACM SIGACT Symposium on
  Theory of Computing (STOC 2020)}, pages 1251--1264. ACM, 2020.

\bibitem[Minoux(1978)]{Minoux1978-jr}
M.~Minoux.
\newblock Accelerated greedy algorithms for maximizing submodular set
  functions.
\newblock In \emph{Optimization Techniques}, pages 234--243. Springer Berlin
  Heidelberg, 1978.

\bibitem[Mirzasoleiman et~al.(2015)Mirzasoleiman, Badanidiyuru, Karbasi,
  Vondrak, and Krause]{Mirzasoleiman2015-qo}
B.~Mirzasoleiman, A.~Badanidiyuru, A.~Karbasi, J.~Vondrak, and A.~Krause.
\newblock Lazier than lazy greedy.
\newblock In \emph{Proceedings of the 29th AAAI Conference on Artificial
  Intelligence (AAAI 2015)}, volume~29, pages 1812--1818. AAAI Press, 2015.

\bibitem[Nakamura et~al.(2022)Nakamura, Sakaue, Fujii, Harabuchi, Maeda, and
  Iwata]{Nakamura2022-nd}
T.~Nakamura, S.~Sakaue, K.~Fujii, Y.~Harabuchi, S.~Maeda, and S.~Iwata.
\newblock Selecting molecules with diverse structures and properties by
  maximizing submodular functions of descriptors learned with graph neural
  networks.
\newblock \emph{Scientific Reports}, 12\penalty0 (1):\penalty0 1124, 2022.

\bibitem[Nemhauser and Wolsey(1978)]{Nemhauser1978-vk}
G.~L. Nemhauser and L.~A. Wolsey.
\newblock Best algorithms for approximating the maximum of a submodular set
  function.
\newblock \emph{Mathematics of Operations Research}, 3\penalty0 (3):\penalty0
  177--188, 1978.

\bibitem[Nemhauser et~al.(1978)Nemhauser, Wolsey, and Fisher]{Nemhauser1978-dr}
G.~L. Nemhauser, L.~A. Wolsey, and M.~L. Fisher.
\newblock An analysis of approximations for maximizing submodular set
  functions---{I}.
\newblock \emph{Mathematical Programming}, 14\penalty0 (1):\penalty0 265--294,
  1978.

\bibitem[Ohsaka(2022)]{Ohsaka2022-ng}
N.~Ohsaka.
\newblock Some inapproximability results of {MAP} inference and exponentiated
  determinantal point processes.
\newblock \emph{Journal of Artificial Intelligence Research}, 73:\penalty0
  709--735, 2022.

\bibitem[Sakaue(2020)]{Sakaue2020-ap}
S.~Sakaue.
\newblock Guarantees of stochastic greedy algorithms for non-monotone
  submodular maximization with cardinality constraints.
\newblock In \emph{Proceedings of the 23rd International Conference on
  Artificial Intelligence and Statistics (AISTATS 2020)}, volume 108, pages
  11--21. PMLR, 2020.

\end{thebibliography}

\appendix

\clearpage

%\onecolumn
\begin{center}
	{\fontsize{18pt}{0pt}\selectfont \bf Appendix}
\end{center}

% \section{Section title}\label{app-section:label}
%
% \begin{algorithm}[H]
%   \caption{Lazy Greedy for submodular maximization with cardinality constraints}
%   %\label{alg:Minoux1978-jr}
%   \begin{algorithmic}[1]
%   \State \textbf{Input} : submodular function $f\colon 2^{[n]}\to\R$,~cardinality constraint $k\in[n]$
%   \State \textbf{Initialize} : $S^{(0)}\leftarrow\emptyset,~t\leftarrow 1,~\rho_l\leftarrow f(\{l\})~(\forall l\in[n])$
%    \While{$t\leq k$}
%    \State $j\leftarrow \mathrm{arg\,max\,}_{l\in \overline{\step{S}{t-1}}}\,\rho_l$
%    \State $\rho_{j}\leftarrow f_{j}\mleft(\step{S}{t-1}\mright)$
%    \If{$\rho_{j}\geq\max_{l\in\overline{\step{S}{t-1}}} \rho_l$}
%    \State $e_t\leftarrow j$
%     \State $\step{S}{t}\leftarrow \step{S}{t-1}\cup\{e_t\}$
%     \State $t\leftarrow t+1$
%     \EndIf
%    \EndWhile
%   \State \Return $S^{(k)}$
%   \end{algorithmic}
% \end{algorithm}

\section{Extension to random, stochastic, and interlace greedy algorithms}\label{app-section:extension}

We consider cardinality-constrained DPP MAP inference and explain how to extend our ``lazy + fast'' idea to other greedy-type algorithms: \textsc{RandomGreedy}~\citep{Buchbinder2014-il}, \textsc{StochasticGreedy}~\citep{Mirzasoleiman2015-qo,Sakaue2020-ap}, and \textsc{InterlaceGreedy}~\citep{Kuhnle2019-io}.
Roughly speaking, those algorithms have in common the process of finding an element with the largest marginal gain, which we can do efficiently with our ``lazy + fast'' technique.
We also present experimental results on those algorithms.

In the analysis of \textsc{RandomGreedy}~\citep{Buchbinder2014-il}, \textsc{StochasticGreedy}~\citep{Sakaue2020-ap}, and \textsc{InterlaceGreedy}~\citep{Kuhnle2019-io}, $n \ge 2k$, $n \ge 3k$, and $n \ge 4k$, respectively, are assumed.
Thus, we below make the same assumptions.

\subsection{Lazy and fast random greedy algorithm}\label{app-subsection:random}
We consider \textsc{RandomGreedy}~\citep{Buchbinder2014-il}, a $1/\mathrm{e}$-approximation algorithm for non-monotone submodular function maximization with a cardinality constraint.
It works as follows.
First, we add $2k$ dummy elements that do not affect the function value to the ground set in advance.
In each step, \textsc{RandomGreedy} finds a set $M$ of $k$ elements with the top-$k$ marginal gains, draws an element $i$ from $M$ uniformly at random, and then adds $i$ to the solution.
This procedure is repeated $k$ times.

For convenience, we consider another equivalent algorithm.
In each step of the above original version, an element with the  $l$th largest marginal gain is added to a current solution, where $l$ is drawn from $[k]$ uniformly at random.
Since there always exist at least $k$ remaining dummy elements, the added element always has a non-negative marginal gain in every step.
Considering the above, each step of the original \textsc{RandomGreedy} can be equivalently performed without dummy elements as follows.
Let $S \subseteq [n]$ be a current solution.
We draw $l$ from the uniform distribution on $[k]$. ($\overline{S}$ has at least $l$ elements since $n \ge 2k$ implies $|\overline{S}| \ge n - k \ge k \ge l$.)
If an element with the $l$th largest marginal gain, denoted by $i\in\overline{S}$, satisfies $f_i(S) \ge 0$, we add $i$ to $S$.
Otherwise, we add nothing to $S$, corresponding to adding a dummy element in the original algorithm.
Note that each step of this algorithm requires finding elements with the top-$l$ marginal gains.

We below explain how to use the ``lazy + fast'' idea for finding $M$, a set of elements with the top-$l$ marginal gains.
First, we compute $d_i = \sqrt{L_{i,i}}$ for every $i \in [n]$ and push them into a priority queue.
In each $t$th step, given a current solution $S$, $l\in[k]$, and $M = \emptyset$, we find an element $i$ with the largest $f_i(S)$ among $\overline{S \cup M}$ by calling \textsc{UpdateRow} (defined in \cref{alg:LazyMapInf}) and add $i$ to $M$.
We repeat this until $|M| = l$ or $d_i < 1$ holds, where $d_i < 1$ implies that the $l$th largest marginal gain is negative.
After updating the current solution $S$, we insert back the elements that have not been added into the priority queue.
\Cref{alg:Lazy_Random} presents a formal description of this algorithm.

Let $U = \sum_{i \in [n]} u_i$ be the number of computed off-diagonals of $V$ at the end of the algorithm, which affects the time complexity of \cref{alg:Lazy_Random} as follows.

\begin{thm}\label{thm:LazyRandomCorrectness_Complexity}
  \Cref{alg:Lazy_Random} returns a solution obtained by \textsc{RandomGreedy} in $\Ord(nd + U(d+\log n))$ time.
  If the lazy update works best and worst, it runs in $\Ord(nd + k^2(d+\log n))$ and $\Ord(kn(d + \log n))$ time, respectively.
\end{thm}
\begin{proof}
  The correctness of \cref{alg:Lazy_Random} and the time complexity depending on $U$ follow from similar arguments to those in \cref{subsection:lazyfastgreedy}.
  We below discuss the best and worst cases.
  For ease of deriving upper bounds on the time complexity, we suppose $l = k$ to hold in every step, which is the worst scenario for the randomness of the algorithm.

  If the lazy update works best, in each step, it picks $k-1$ elements that have belonged to $M$ one step before and a new element whose row in a Cholesky factor has not been computed.
  Then, in each $t$th step ($t=2,3,\cdots,k$), $(k-1)+(t-1)$ off-diagonals are computed, hence $U=\sum_{t=2}^{k}((k-1)+(t-1)) = \Ord(k^2)$.
  In addition, since $k - 1$ elements are inserted in each step, updating the priority queue takes $\Ord(k^2\log n)$ in total.
  Thus, the overall running time is $\Ord(nd + k^2(d+\log n))$.

  We then turn to the worst case.
  In the $t$th step ($t = 2,3,\dots, k$), all the off-diagonals in $V[\overline{S}, j_{t-1}]$ are calculated to determine the top-$k$ element.
  Hence, it computes as many off-diagonals as \textsc{FastGreedy}, i.e., $U = (k-1)(n-k/2)$.
  Therefore, it takes $\Ord(kn(d + \log n))$ time.
\end{proof}

In the above proof, we have assumed the worst scenario of $l = k$, which is pessimistic in practice.
We here discuss the optimistic case of $l = 1$ to derive the range of possible $U$ values.
In this case, \cref{alg:Lazy_Random} selects an element with the largest marginal gain in each step as with \textsc{LazyFastGreedy}.
Hence, $U = k(k-1)/2$ off-diagonals of $(PV)[S]$ are computed, which is the lower bound on $U$.
The upper bound on $U$ is $(k-1)(n-k/2)$ as shown in the above proof.
Therefore, $U$ can take a value in the range of $[k(k-1)/2, (k-1)(n-k/2)]$.

\begin{algorithm}[tb]
  \caption{Lazy and fast \textsc{RandomGreedy} for DPP MAP inference}\label{alg:Lazy_Random}
  \begin{algorithmic}[1]
   \State{$V\gets O$, $\bm{d} \gets {\left(\sqrt{L_{i,i}}\right)}_{i \in [n]}$, ${\bm u}\gets {\bm 0}$, $S\gets \emptyset$}
   \State{Construct a priority queue of ${\bm d}$}
   \For{$t = 1$ to $k$}
   \State Draw $l$ from the uniform distribution on $[k]$ \Comment{$l \le |\overline{S}|$ always holds since $n \ge 2k$}
     \State $M\leftarrow\emptyset$
     \While {$|M|<l$}
       \State Take $i \in \mathrm{arg\,max\,}_{i^\prime\in \overline{S \cup M}} d_{i'}$\Comment{Pop the largest one from the priority queue}
       \State \Call{UpdateRow}{$V, \bm{d}, \bm{u}; i, S, L$}\Comment{$S = \set*{j_1, j_2,\dots,j_{|S|}}$}
       \If {$d_{i}\geq\max_{i^\prime\in \overline{S \cup M}} d_{i^\prime}$}\Comment{Otherwise insert $d_i$ into the priority queue}
       \If{$d_i < 1$}
         \State{\textbf{break}}\Comment{The $l$th largest marginal gain must be negative}
       \EndIf
       \State $M \leftarrow M \cup \{i\}$
       \If {$|M| = l$}
      \State $j_{|S|+1} \gets i$ \Comment{$i$ has the $l$th largest non-negative marginal gain}
      \State $S \leftarrow S \cup \{j_{|S|+1}\}$
    \EndIf
       \EndIf
    \EndWhile
    \State Insert $\set{d_i}_{i \in M \setminus S}$ into the priority queue
   \EndFor
  \State \Return $S$
  \end{algorithmic}
\end{algorithm}

\subsection{Lazy and fast stochastic greedy algorithm}
We then discuss \textsc{StochasticGreedy} \citep{Mirzasoleiman2015-qo}, which was initially proposed as a simple and fast variant of \textsc{Greedy} for maximizing a monotone submodular function under a cardinality constraint.
A recent study \citep{Sakaue2020-ap} showed that \textsc{StochasticGreedy} also enjoys a nearly $1/4$-approximation guarantee even for non-monotone submodular functions if $k \ll n$.

\textsc{StochasticGreedy} has a parameter $\varepsilon \in (0, 1)$ that controls the trade-off between the running time and approximation ratio (see \citep{Mirzasoleiman2015-qo,Sakaue2020-ap} for details).
In each step of \textsc{StochasticGreedy}, given a current solution $S$, it obtains $R \subseteq \overline{S}$ by sampling $s = \mleft\lceil \frac{n}{k}\log\frac{1}{\varepsilon}\mright\rceil$ elements uniformly at random from $\overline{S}$, selects $j_{t} \in \argmax_{i^\prime \in R} f_{i^\prime}(S)$, and adds $j_{t}$ to $S$ if its marginal gain is positive.
This is repeated for $t=1,\dots,k$.
Note that we have $\log\frac{1}{\varepsilon} \le k$ since $\mleft\lceil \frac{n}{k}\log\frac{1}{\varepsilon}\mright\rceil = s \le n$ (otherwise we let $R = \overline{S}$).

We can use the idea of ``lazy + fast'' to speed up the process of finding an element with the largest marginal gain from $R$.
To implement the idea, we use an array that maintains the latest $d_i$ values of each element $i \in [n]$.
In each step, we construct a priority queue that maintains $(d_i)_{i \in R}$, where $d_i$ values are given by those stored in the array.
We then find an element with the largest marginal gain among $R$ by iteratively executing \textsc{UpdateRow} with the priority queue.
We provide a formal description of this algorithm in \Cref{alg:Lazy_Stochastic}.

\begin{algorithm}[tb]
  \caption{Lazy and fast \textsc{StochasticGreedy} for DPP MAP inference}\label{alg:Lazy_Stochastic}
  \begin{algorithmic}[1]
   \State{$V\gets O$, $\bm{d} \gets {(+\infty)}_{i \in [n]}$, ${\bm u}\gets {\bm 0}$, $S\gets \emptyset$, $s\gets\mleft\lceil \frac{n}{k}\log\frac{1}{\varepsilon}\mright\rceil$}\Comment{Initializaiton of $\bm{d}$ is deferred}
   \For {$t = 1$ to $k$}
    \State Get $R$ by sampling $s$ elements from $\overline{S}$ \Comment{$R = \overline{S}$ if $|\overline{S}| \le s$}
	\State Construct a priority queue with ${(d_i)}_{i \in R}$ \Comment{Let $d_i \gets \sqrt{L_{i,i}}$ if $d_i = +\infty$}
	\While{true}
    \State Take $i \in \mathrm{arg\,max\,}_{i^\prime\in R} d_{i'}$ \Comment{Pop the largest one from the priority queue}
    \State \Call{UpdateRow}{$V, \bm{d}, \bm{u}; i, S, L$}
	\If{$d_{i} \ge \max_{i^\prime\in R} d_{i^\prime}$}
    \Comment{Otherwise insert $d_i$ into the priority queue}
    \State \textbf{break}
	\EndIf
	\EndWhile
    \If {$d_i > 1$}
    \State $j_{|S|+1}\leftarrow i$
    \State $S\leftarrow S\cup\{j_{|S|+1}\}$
    \EndIf
   \EndFor
  \State \Return $S$
  \end{algorithmic}
\end{algorithm}

Letting $U$ denote the number of computed off-diagonals as above, the time complexity of \Cref{alg:Lazy_Stochastic} is represented as follows.
\begin{thm}\label{thm:LazyStochasticCorrectness_Complexity}
  \Cref{alg:Lazy_Stochastic} returns a solution obtained by \textsc{StochasticGreedy} in $\Ord\big(nd + U\big(d + \log \big(\frac{n}{k}\log\frac{1}{\varepsilon}\big) \big) \big)$ time.
  If the lazy update works best and worst, it takes $\Ord\big((n + k^2)d \big) \big)$ and $\Ord\big(knd + n \log\big( \frac{1}{\varepsilon} \big) \log \big(\frac{n}{k}\log\frac{1}{\varepsilon}\big)\big)$ time, respectively.
\end{thm}
\begin{proof}
  The correctness of \cref{alg:Lazy_Random} follows from a similar argument to that in \cref{subsection:lazyfastgreedy}.

  %Note that \cref{alg:Lazy_Random} does not always compute all diagonals $L_{i,i}$, unlike other greedy-type algorithms.
  %We instead compute diagonals necessary for constructing a priority queue of sampled elements in each step.
  Since we compute at most $n$ diagonals, the time for computing diagonals is $\Ord(nd)$ in total.
  The time required for constructing priority queues of $s$ sampled elements is $\Ord(sk)$ in total, which is bounded by $\Ord(nd)$ since $s\leq n$ and $k\leq d$.
  %$s = \mleft\lceil \frac{n}{k}\log\frac{1}{\varepsilon}\mright\rceil$ and $k \le d$ imply $\Ord(sk) \lesssim \Ord\big(n\log\frac{1}{\varepsilon}\big) \lesssim \Ord(nk) \lesssim \Ord(nd)$.

  In addition, as in the proof of \cref{thm:LazyGreedyCorrectness_Complexity}, the total computation time caused by \textsc{UpdateRow} is $\Ord(Ud)$.
  Furthermore, we do at most $\Ord(U)$ operations on priority queues in total, taking $\Ord(U \log s)$ time.
  Thus, the overall time complexity is $\Ord\big(nd + U\big(d + \log \big(\frac{n}{k}\log\frac{1}{\varepsilon}\big) \big) \big)$.

  We below discuss the best and worst cases.
  As with the proof of \cref{thm:LazyRandomCorrectness_Complexity}, we assume the worst scenario for the randomness of the algorithm to derive upper bounds on the time complexity.

  In the best case, the element on top of the priority queue is added to $S$ in every step.
  Consequently, $k(k-1)/2$ off-diagonals of $V$ are computed.
  In addition, updates of priority queues occur $k$ times, which takes $\Ord(k \log s)$ ($\lesssim \Ord(nd)$) time in total.
  Therefore, the best-case time complexity is written as $\Ord\mleft(nd + \frac{k(k-1)}{2}d \mright) =  \Ord\mleft((n + k^2)d \mright)$.

  The worst case occurs if $R \subseteq \overline{S}$ sampled in each step minimizes $\sum_{i \in R'} u_i$ among all $R' \subseteq \overline{S}$ with $|R'| = s$, where $u_i$ values are those of the current step, and if the marginal gains of all elements in $R$ are updated.
  Let $q$ and $r$ be the quotient and the remainder, respectively, of $n$ divided by $s$, i.e., $n = qs+r$ and $0 \le r < s$.
  Then, we have
  \begin{align}
    U = \sum_{t=k-q+1}^k s(t-1) + \left(r-\frac{k-q}{2}\right)(k-q-1) = \left(n-\frac{k}{2}\right)(k-q-1)+ \frac{kq}{2} = \Ord\left(nk\right).
  \end{align}
  Moreover, since priority queues are updated at most $sk$ times, the total time taken for operations on priority queues is $\Ord(sk \log s) \lesssim \Ord\big( n \log\big( \frac{1}{\varepsilon} \big) \log \big(\frac{n}{k}\log\frac{1}{\varepsilon}\big) \big)$.
  Thus, the overall time complexity is $\Ord\big(knd + n \log\big( \frac{1}{\varepsilon} \big) \log \big(\frac{n}{k}\log\frac{1}{\varepsilon}\big)\big)$.
\end{proof}

Note that by using $\log\frac{1}{\varepsilon} \le k$, we can confirm that the time complexity bounds in \cref{thm:LazyStochasticCorrectness_Complexity} are at least as good as those of \textsc{LazyFastGreedy} in \cref{thm:LazyGreedyCorrectness_Complexity}.

\subsection{Lazy and fast interlace greedy algorithm}

The idea of ``lazy + fast'' is also applicable to \textsc{InterlaceGreedy}~\citep{Kuhnle2019-io}, which is a $1/4$-approximation algorithm for non-monotone submodular maximization with a cardinality constraint.
Note that \textsc{InterlaceGreedy} is a deterministic algorithm, unlike \textsc{RandomGreedy} and \textsc{StochasticGreedy}.

\textsc{InterlaceGreedy} sequentially updates two solution sets.
Starting from $A = \emptyset$ and $B = \emptyset$, in each step, it adds
$j^A \in \argmax_{i \in \overline{A \cup B}} f_i(A)$ to $A$ and
$j^B \in \argmax_{i \in \overline{(A \cup \set{j^A}) \cup B}} f_i(B)$ to $B$, until $|A| = |B| = k$ holds (or the largest marginal gain turns out to be negative).
Then, it repeats the same procedure with different initial solutions $C$ and $D$ consisting of the first element added to $A$.
Finally, it returns a set with the largest function value among all sets that have appeared as $A$, $B$, $C$, or $D$ at some step during the execution.

The idea of ``lazy + fast'' can be applied to \textsc{InterlaceGreedy} by using priority queues for $A$, $B$, $C$, and $D$.
First, we initialize two priority queues for $A$ and $B$ with ${\left(\sqrt{L_{i,i}}\right)}_{i \in [n]}$.
In each step, by iteratively using \textsc{UpdateRow}, we can efficiently find an element with the largest marginal gain for $A$ and $B$, as with \textsc{LazyFastGreedy}.
Here, since elements already added to $A$ cannot be added to $B$, if $i$ is added to $A$, we remove $i$ from the priority queue for $B$ and vice versa.
We thus construct sequences of solutions, ${\left(\step{A}{t}\right)}_{t=0}^k$ and ${\left(\step{B}{t}\right)}_{t=0}^k$.
Then, we obtain ${\left(\step{C}{t}\right)}_{t=1}^k$ and ${\left(\step{D}{t}\right)}_{t=1}^k$ in the same manner except that $C$ and $D$ are initially set to $\step{A}{1}$.
We provide a formal description of this algorithm in \cref{alg:Lazy_Interlace}.

\begin{algorithm}[tb]
  \caption{Lazy and fast \textsc{InterlaceGreedy} for DPP MAP inference}\label{alg:Lazy_Interlace}
  \begin{algorithmic}[1]
    \State ${\left(\step{A}{t}\right)}_{t=0}^k, {\left(\step{B}{t}\right)}_{t=0}^k$ $\gets$\Call{GetInterlacedSets}{$\nil, L$}
    \State ${\left(\step{C}{t}\right)}_{t=0}^k, {\left(\step{D}{t}\right)}_{t=0}^k$ $\gets$\Call{GetInterlacedSets}{$j^A_1, L$}\Comment{$\step{A}{1} = \set{j_1^A}$}
    \State{\Return $S \in \argmax \{ \log \det L[S'] \mid S' = \text{$\step{X}{t} \, (X \in \set{A, B, C, D}, t = 0, \dotsc, k)$} \}$}
    \item[]
    \Function{GetInterlacedSets}{$j_1, L$}
    \State $V^S \gets O$, $V^T \gets O$, $\bm{d}^S \gets {\left(\sqrt{L_{i,i}}\right)}_{i \in [n]}$, $\bm{d}^T \gets {\left(\sqrt{L_{i,i}}\right)}_{i \in [n]}$, ${\bm u}^S \leftarrow {\bm 0}$, ${\bm u}^T\leftarrow{\bm 0}$
    \State Construct priority queues of ${\bm d}^S$ and ${\bm d}^T$
    \State{$\step{S}{0} \gets \emptyset$, $\step{T}{0} \gets \emptyset$, $t_0 \gets |\{j_t\}|+1$}\Comment{$|\{\textsc{nil}\}| =0$}
    \If{$j_1 \ne \nil$}
      \State{$j^S_1 \gets j_1$, $j^T_1 \gets j_1$, $\step{S}{1} \gets \set{j^S_1}$, $\step{T}{1} \gets \set{j^T_1}$}
      \State{Remove $j_1$ from the priority queues of $\bm{d}^S$ and $\bm{d}^T$}
    \EndIf
    \For{$t = t_0$ to $k$}
      \State{$i \gets$ \Call{GetArgMax}{$V^S, \bm{d}^S, \bm{u}^S; \step{S}{t-1}, \step{T}{t-1}, L$}}
      \If{$i \ne \nil$ and $d^S_i \ge 1$}{\label[line]{line:interlace-repeat-start}}
        \State{$j_{t}^S \gets i$, $\step{S}{t} \gets \step{S}{t-1} \cup \set{j_{t}^S}$}
        \State{Remove $j_{t}^S$ from the priority queue of $\bm{d}^T$}
      \Else
        \State{$\step{S}{t} \gets \step{S}{t-1}$}
      \EndIf{\label[line]{line:interlace-repeat-end}}
      \State{$i \gets$ \Call{GetArgMax}{$V^T, \bm{d}^T, \bm{u}^T; \step{T}{t-1}, \step{S}{t}, L$}}
      \State{Repeat Lines~\ref{line:interlace-repeat-start}--\ref{line:interlace-repeat-end} interchanging $S$ and $T$}
   \EndFor
  \State\Return ${\left(\step{S}{t}\right)}_{t=0}^k$, ${\left(\step{T}{t}\right)}_{t=0}^k$
  \EndFunction
  \item[]
  \Function{GetArgMax}{$V, \bm{d}, \bm{u}; X, Y, L$}
    \If{$\max_{i^\prime\in \overline{X \cup Y}} d_{i^\prime} < 1$}
      \State{\Return \nil}
    \EndIf
    \While{true}
      \State Pop $i$ with the largest $d_i$ from the priority queue of ${\bm d}$
      \State{\Call{UpdateRow}{$V, \bm{d}, \bm{u}; i, X, L$}}
      \If{$d_{i} \ge \max_{i^\prime\in \overline{X \cup Y}} d_{i^\prime}$}\Comment{Otherwise insert $d_i$ into the priority queue}\label[line]{line:interlace-if-max-S}
        \State{\Return $i$}
      \EndIf
    \EndWhile
  \EndFunction
  \end{algorithmic}
\end{algorithm}

For each $X \in \set{A, B, C, D}$, let $V^X$ be a Cholesky factor of $L[\step{X}{k}]$ (partially) computed in \cref{alg:Lazy_Interlace}, $P^X$ be a permutation matrix such that $P^X V^X$ is lower triangular, $U^X$ be the total number of computed off-diagonals in $P^X V^X$, and $U = U^A + U^B + U^C + U^D$.
Then, \cref{alg:Lazy_Interlace} enjoys the following guarantee.

\begin{thm}\label{thm:LazyInterlaceCorrectness_Complexity}
  \Cref{alg:Lazy_Interlace} returns a solution obtained by naive \textsc{InterlaceGreedy} in $\Ord(nd + U(d+\log n))$ time.
  If the lazy update works best and worst, it runs in $\Ord((n + k^2)d)$ and $\Ord(kn(d + \log n))$ time, respectively.
\end{thm}
\begin{proof}
  The correctness of \cref{alg:Lazy_Random} and the time complexity depending on the $U$ value follow from similar arguments to those in \cref{subsection:lazyfastgreedy}.
  We below discuss the best and worst cases.

  In the best case, \textsc{UpdateRow} is called up to $4k$ times and for $X \in \set{A, B, C, D}$, $U^X = k(k-1)/2$ off-diagonals of $\left(P^X V^X\right)[\step{X}{k}]$ are computed.
  Therefore, $4\cdot k(k-1)/2 = 2k(k-1) = \Ord(k^2)$ off-diagonals are computed, taking $\Ord(k^2 d)$ time in total.
  Moreover, the priority queues are updated at most $4k$ times, taking $\Ord(k \log n)$ ($\lesssim \Ord(nd)$) time in total.
  Thus, it runs in $\Ord((n + k^2)d)$ time.

  We then discuss the worst case.
  In the $t$th step ($t = 2,3,\dots, k$) of \textsc{GetInterlacedSets}, the off-diagonals in $V^S[\overline{\step{S}{t-1}\cup \step{T}{t-1}}, j_{t-1}^S]$ and $V^T[\overline{\step{S}{t}\cup \step{T}{t-1}}, j_{t-1}^T]$ are computed.
  Therefore, the total number of computed off-diagonals is
  \begin{align}
    U = \sum_{t=2}^{k}((n-2t+2)+(n-2t+1)+(n-2t+3)+(n-2t+2)) = 4(k-1)(n-k).
  \end{align}
  Hence, \cref{alg:Lazy_Interlace} runs in $\Ord(kn(d + \log n))$ time.
\end{proof}

\subsection{Experimental results on random, stochastic, and interlace greedy algorithms}\label{app-section:experiments-variants}
We experimentally compare the running time of the four versions, Naive, Lazy, Fast, and LazyFast, for each of \textsc{RandomGreedy}, \textsc{StochasticGreedy}, and \textsc{InterlaceGreedy}.
We also compare objective values achieved by the algorithms.

We use the same experimental setups as those in \cref{section:experiments}.
We here focus on the $B$-input setting; as we will see below, the results were similar to those of \textsc{Greedy}, which we confirmed to be true in both $B$- and $L$-input settings.
As mentioned above, \textsc{RandomGreedy}, \textsc{StochasticGreedy}, and \textsc{InterlaceGreedy} require $n\ge 2k$, $n\ge 3k$, and $n\ge 4k$, respectively.
Therefore, when increasing the $k$ value, we set the upper bound to $\lfloor n/4 \rfloor$ to satisfy all the requirements.
Since the procedures of \textsc{RandomGreedy} and \textsc{StochasticGreedy} depend on the $k$ value as in \cref{alg:Lazy_Random,alg:Lazy_Stochastic}, respectively, we cannot measure their running time incrementally for $k = 1,2,\dots,\lfloor n/4 \rfloor$.
Therefore, we select evaluation points in increments of $200$ and $500$ for \textsc{RandomGreedy} and \textsc{StochasticGreedy}, respectively.
We set the parameter $\varepsilon$ of \textsc{StochasticGreedy} to $0.5$.

\Cref{fig:random-greedy,fig:stochastic-greedy,fig:interlace-greedy} summarize the results of \textsc{RandomGreedy}, \textsc{StochasticGreedy}, and \textsc{InterlaceGreedy}, respectively.
Similar to the results of \textsc{Greedy}, our LazyFast version was the fastest in all cases.
In the synthetic setting with $k = 200$ and real-world settings, our LazyFast ran faster than Fast by avoiding the redundant computation of off-diagonals, while it still ran faster in the synthetic setting with $n = 6000$ due to the cache efficiency discussed in \cref{subsection:experiment-greedy}.
%Note that since the procedures of the algorithms are different from \textsc{Greedy}, so are the ranges of the possible number of computed off-diagonals.

\Cref{fig:objective-values} summarizes the running times and objective values of the LazyFast version of the algorithms, where ratios relative to those of \textsc{LazyFastGreedy} are shown.
Although \textsc{LazyFastGreedy} empirically achieved the best objective values, the other algorithms performed comparably.
Note that while the objective values of \textsc{StochasticGreedy} were the worst, its LazyFast version ran the fastest in many cases.
Therefore, when the running time matters more than the solution quality, we can use the LazyFast version of \textsc{StochasticGreedy} instead of \textsc{LazyFastGreedy}.

\begin{figure}[tb]
  \begin{minipage}[b]{0.25\linewidth}
    \centering
    \includegraphics[width=\textwidth]{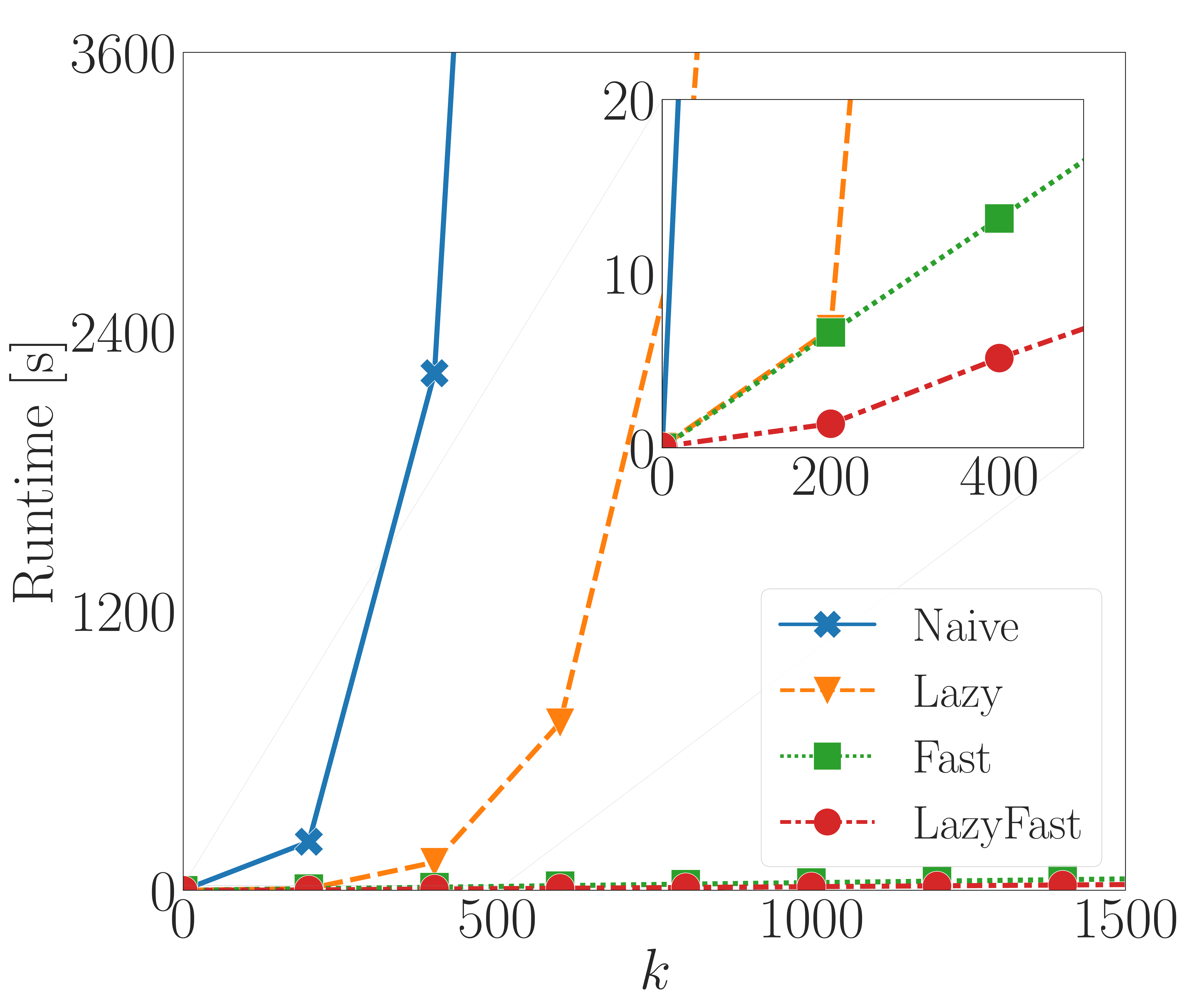}
    \subcaption{$n=6000$, Runtime}\label{subfig:random-synth-k-time}
  \end{minipage}%
  \begin{minipage}[b]{0.25\linewidth}
    \centering
    \includegraphics[width=\textwidth]{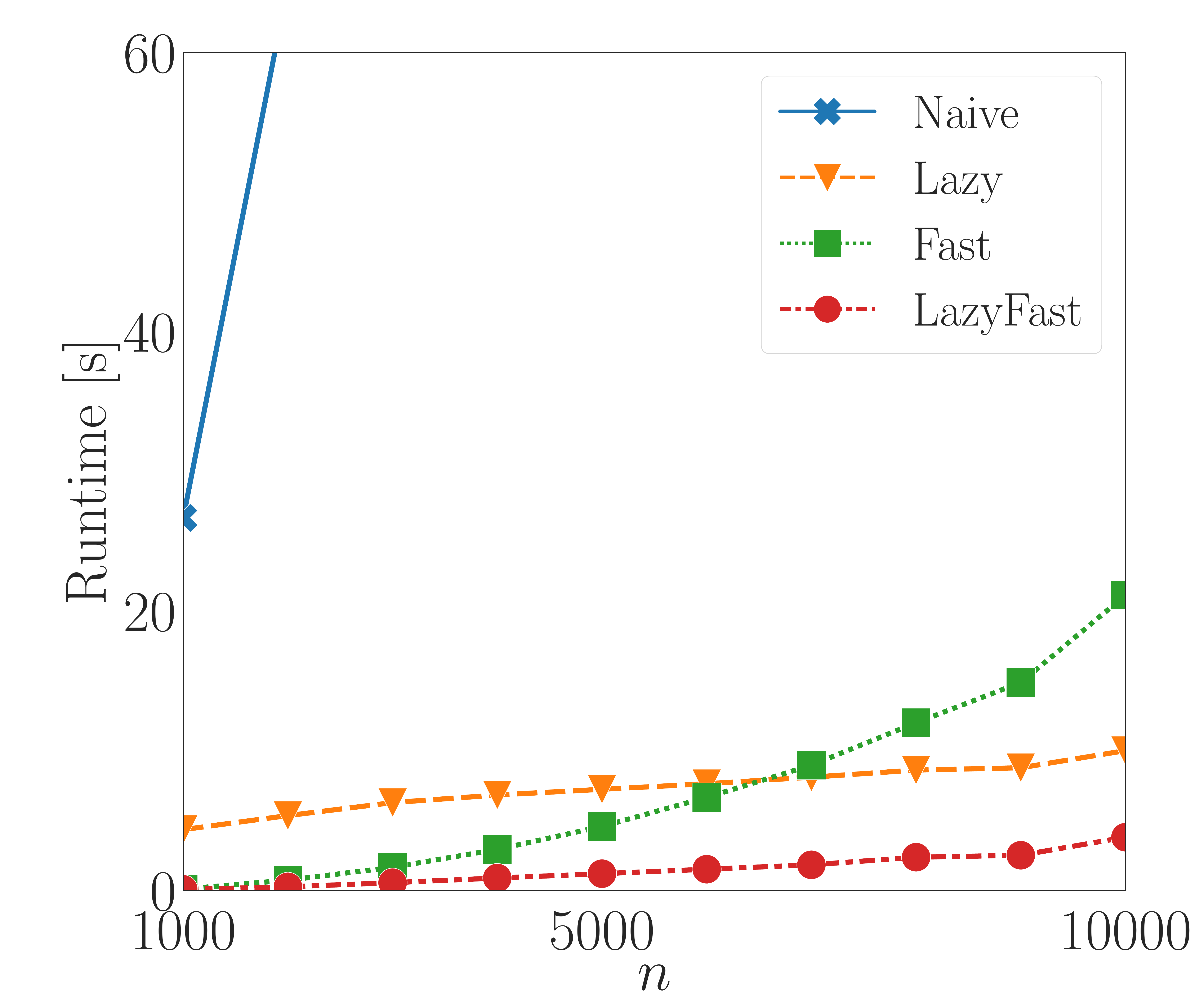}
    \subcaption{$k=200$, Runtime}\label{subfig:random-synth-n-time}
  \end{minipage}%
  \begin{minipage}[b]{0.25\linewidth}
    \centering
    \includegraphics[width=\textwidth]{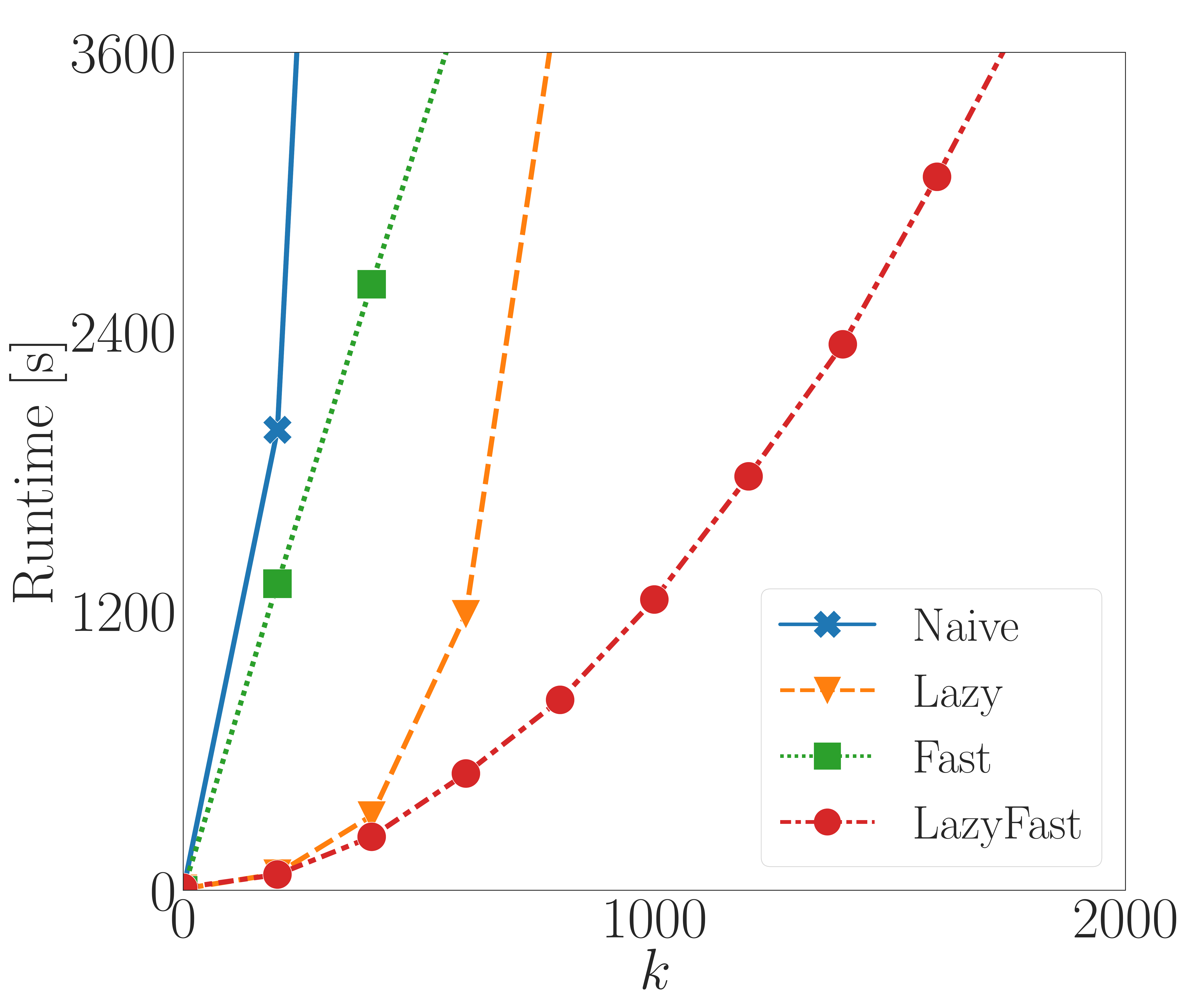}
    \subcaption{Netflix, Runtime}\label{subfig:random-movielens-k-time}
  \end{minipage}%
  \begin{minipage}[b]{0.25\linewidth}
    \centering
    \includegraphics[width=\textwidth]{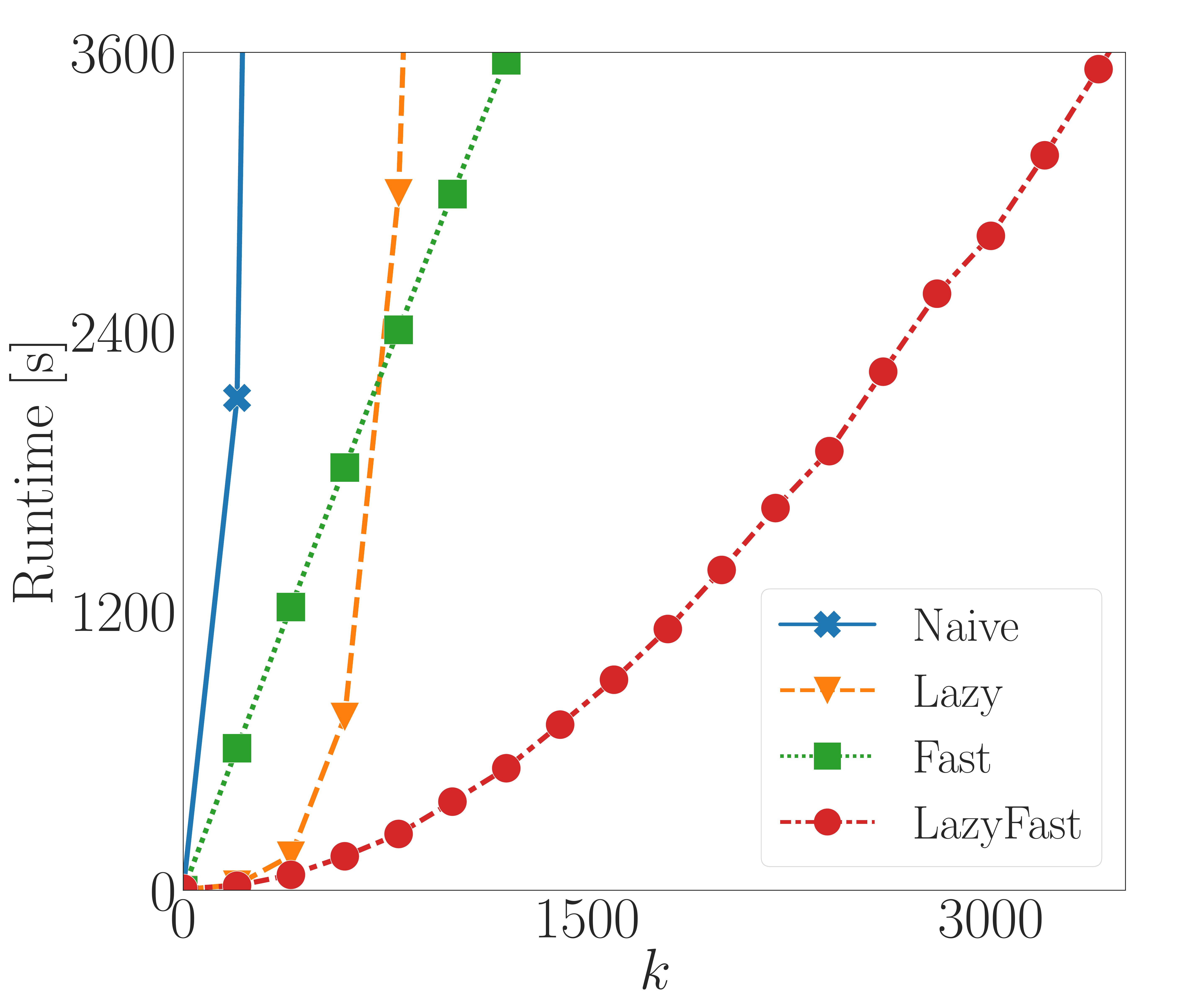}
    \subcaption{MovieLens, Runtime}\label{subfig:random-netflix-k-time}
  \end{minipage}
  \begin{minipage}[b]{0.25\linewidth}
    \vspace{10pt}
    \centering
    \includegraphics[width=\textwidth]{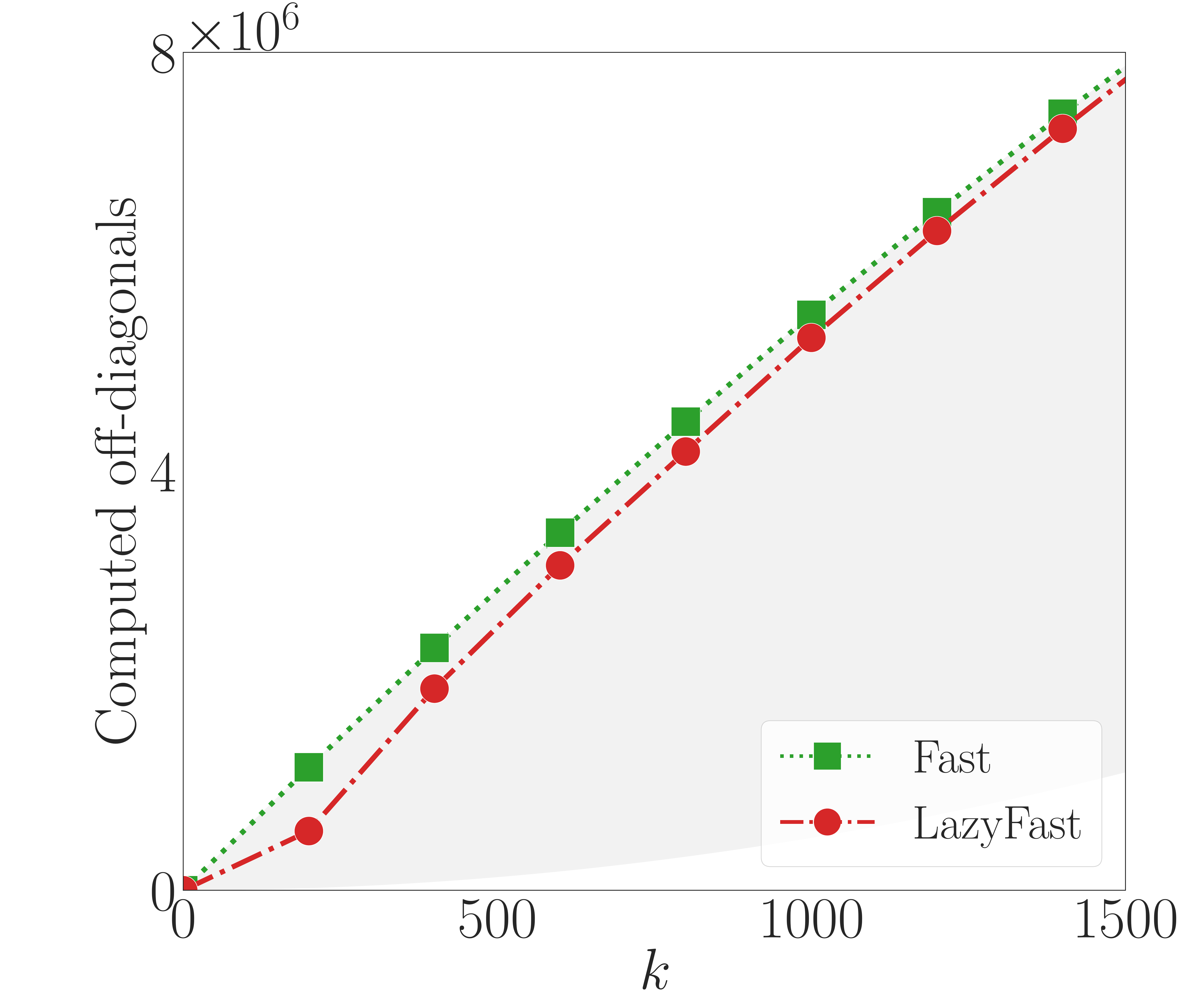}
    \subcaption{$n=6000$, Off-diag.}\label{subfig:random-synth-k-offdiag}
  \end{minipage}%
  \begin{minipage}[b]{0.25\linewidth}
    \vspace{10pt}
    \centering
    \includegraphics[width=\textwidth]{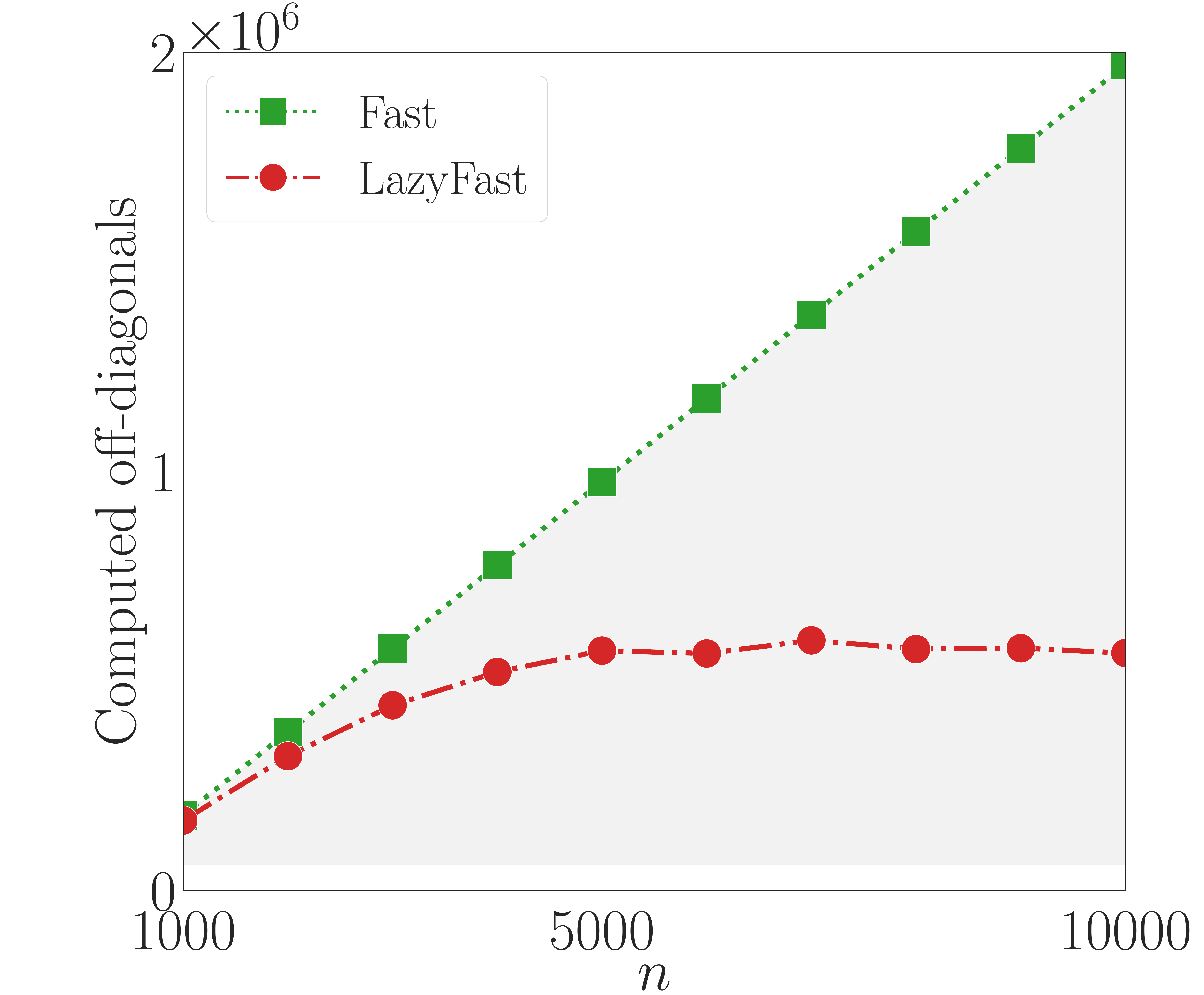}
    \subcaption{$k=200$, Off-diag.}\label{subfig:random-synth-n-offdiag}
  \end{minipage}%
  \begin{minipage}[b]{0.25\linewidth}
    \vspace{10pt}
    \centering
    \includegraphics[width=\textwidth]{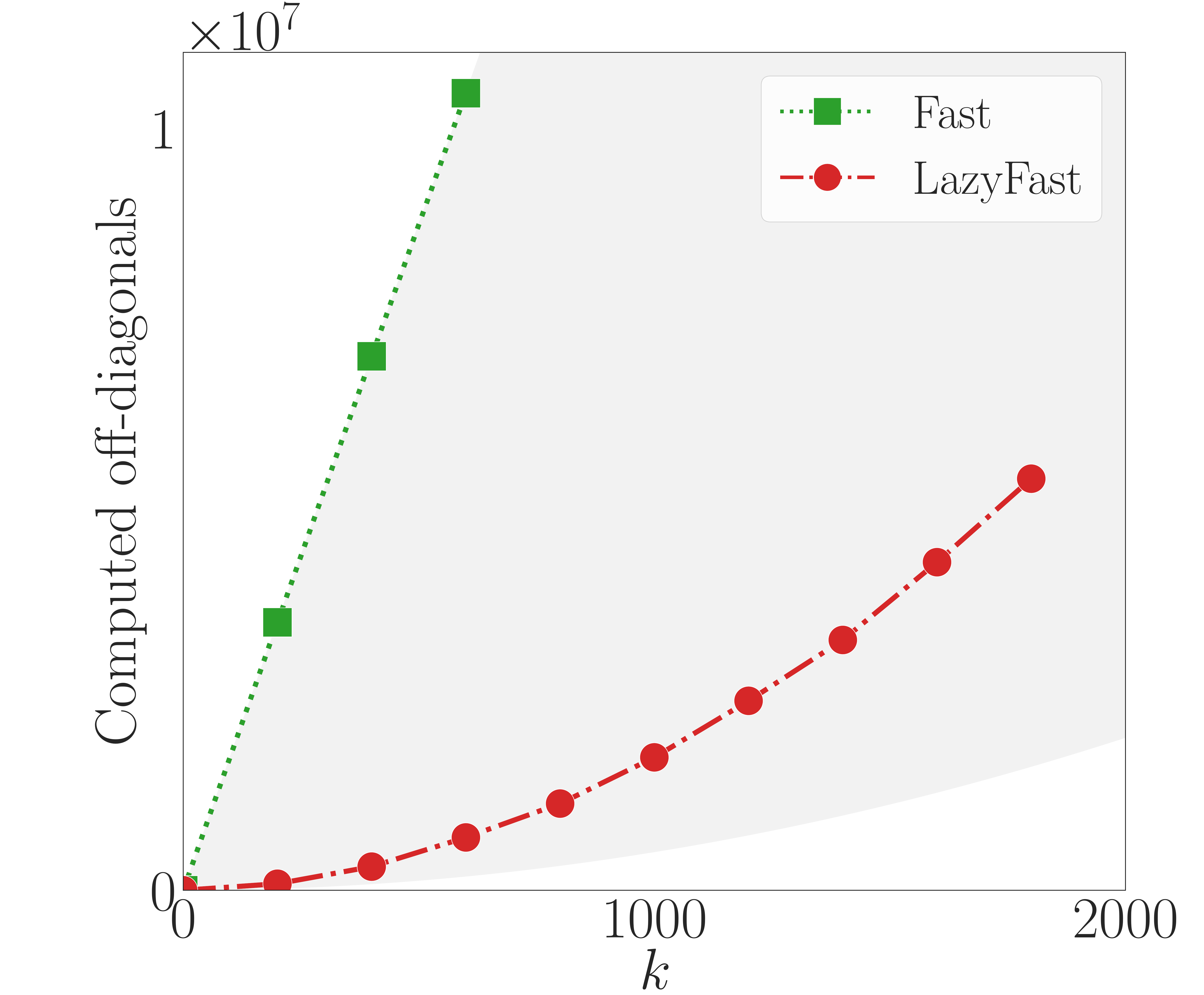}
    \subcaption{Netflix, Off-diag.}\label{subfig:random-movielens-k-offdiag}
  \end{minipage}%
  \begin{minipage}[b]{0.25\linewidth}
    \vspace{10pt}
    \centering
    \includegraphics[width=\textwidth]{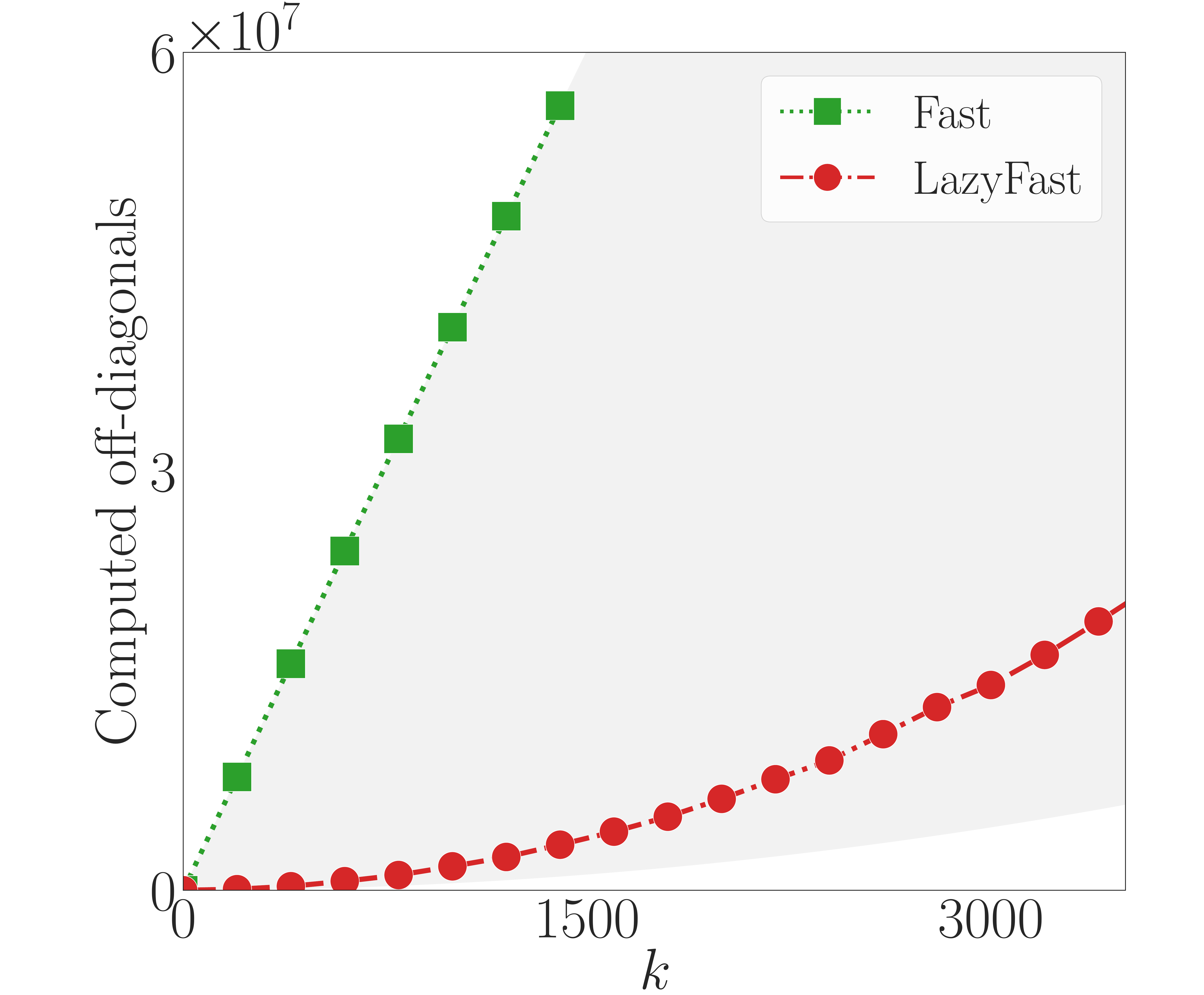}
    \subcaption{MovieLens, Off-diag.}\label{subfig:random-netflix-k-offdiag}
  \end{minipage}
  \caption{Results of \textsc{RandomGreedy}.
  In \cref{subfig:random-synth-k-time}, enlarged views of lower left parts are shown for visibility.
  In the lower figures, the gray band indicates the range of the possible number of computed off-diagonals: $\left[k(k-1)/2, (k-1)(n-k/2)\right]$.}\label{fig:random-greedy}
  \vspace{10pt}
  \begin{minipage}[b]{0.25\linewidth}
    \centering
    \includegraphics[width=\textwidth]{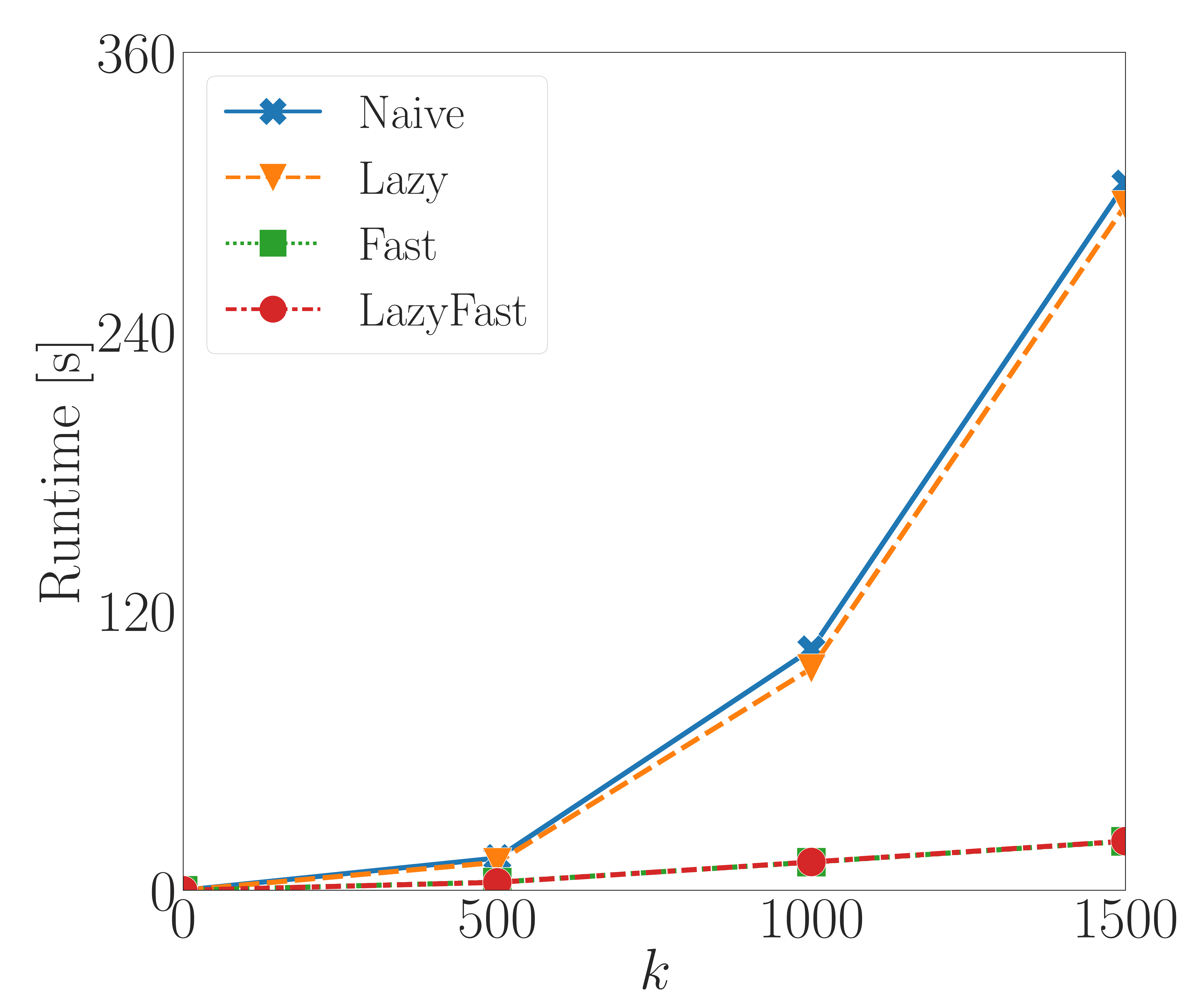}
    \subcaption{$n=6000$, Runtime}\label{subfig:stochastic-synth-k-time}
  \end{minipage}%
  \begin{minipage}[b]{0.25\linewidth}
    \centering
    \includegraphics[width=\textwidth]{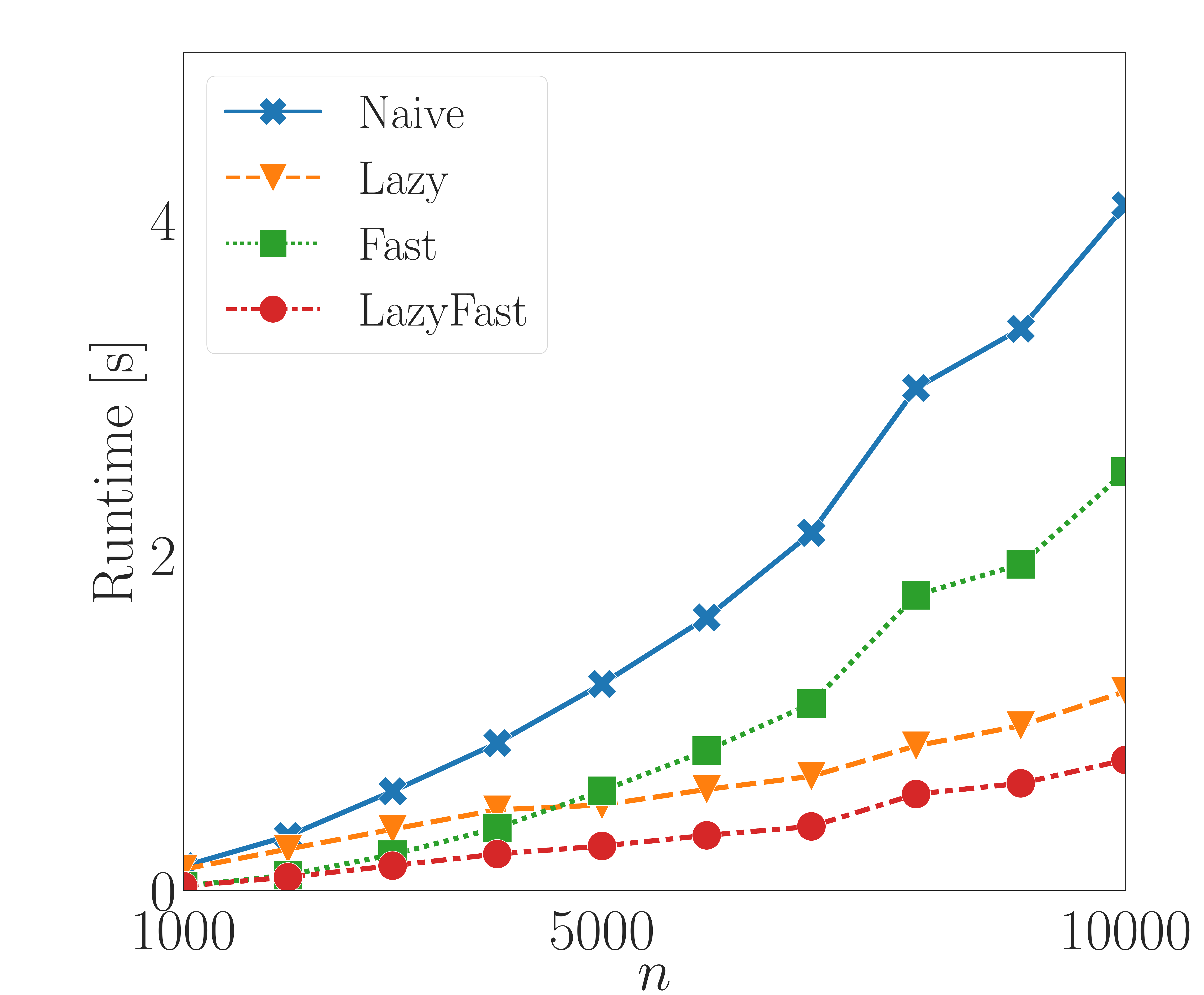}
    \subcaption{$k=200$, Runtime}\label{subfig:stochastic-synth-n-time}
  \end{minipage}%
  \begin{minipage}[b]{0.25\linewidth}
    \centering
    \includegraphics[width=\textwidth]{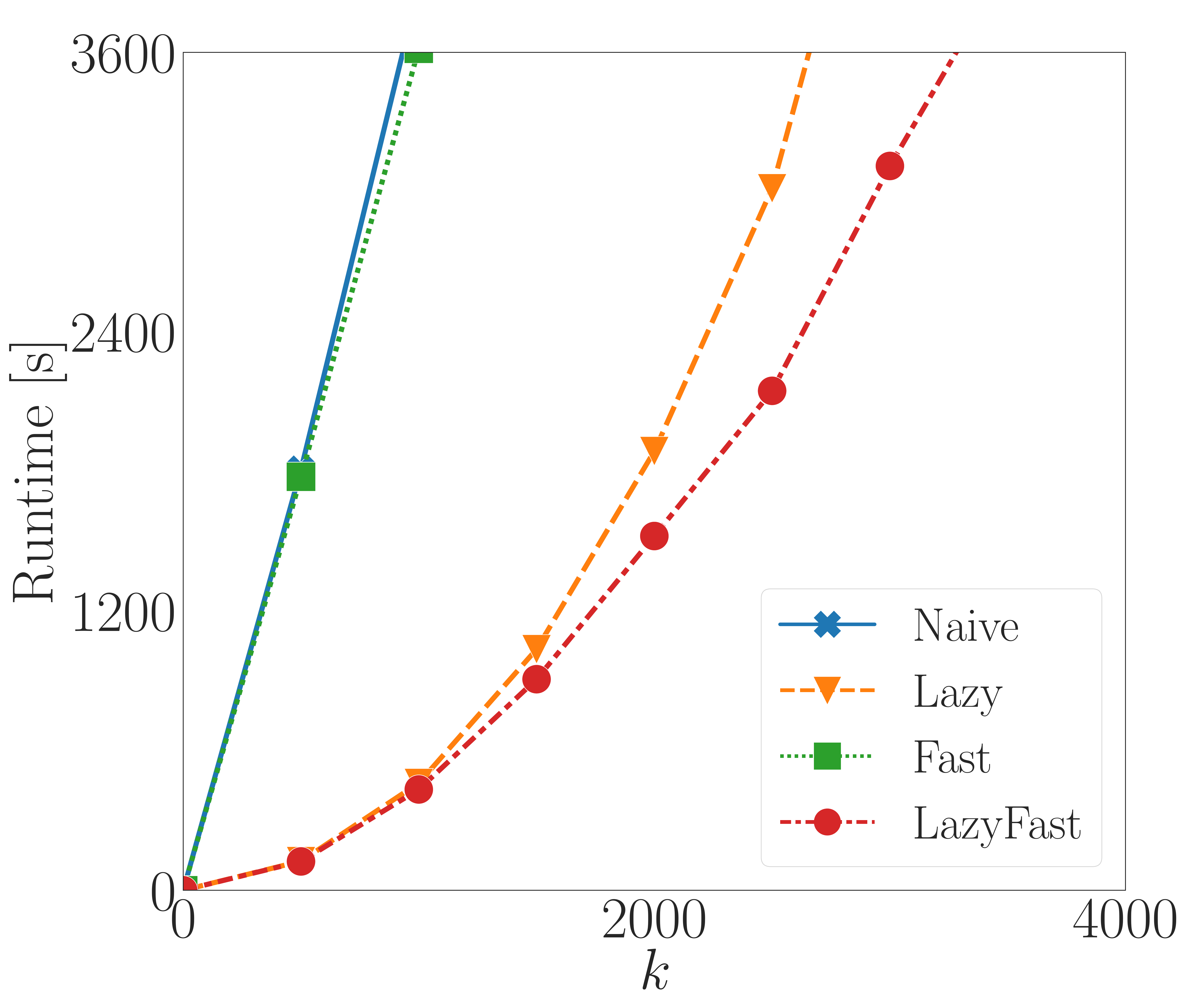}
    \subcaption{Netflix, Runtime}\label{subfig:stochastic-movielens-k-time}
  \end{minipage}%
  \begin{minipage}[b]{0.25\linewidth}
    \centering
    \includegraphics[width=\textwidth]{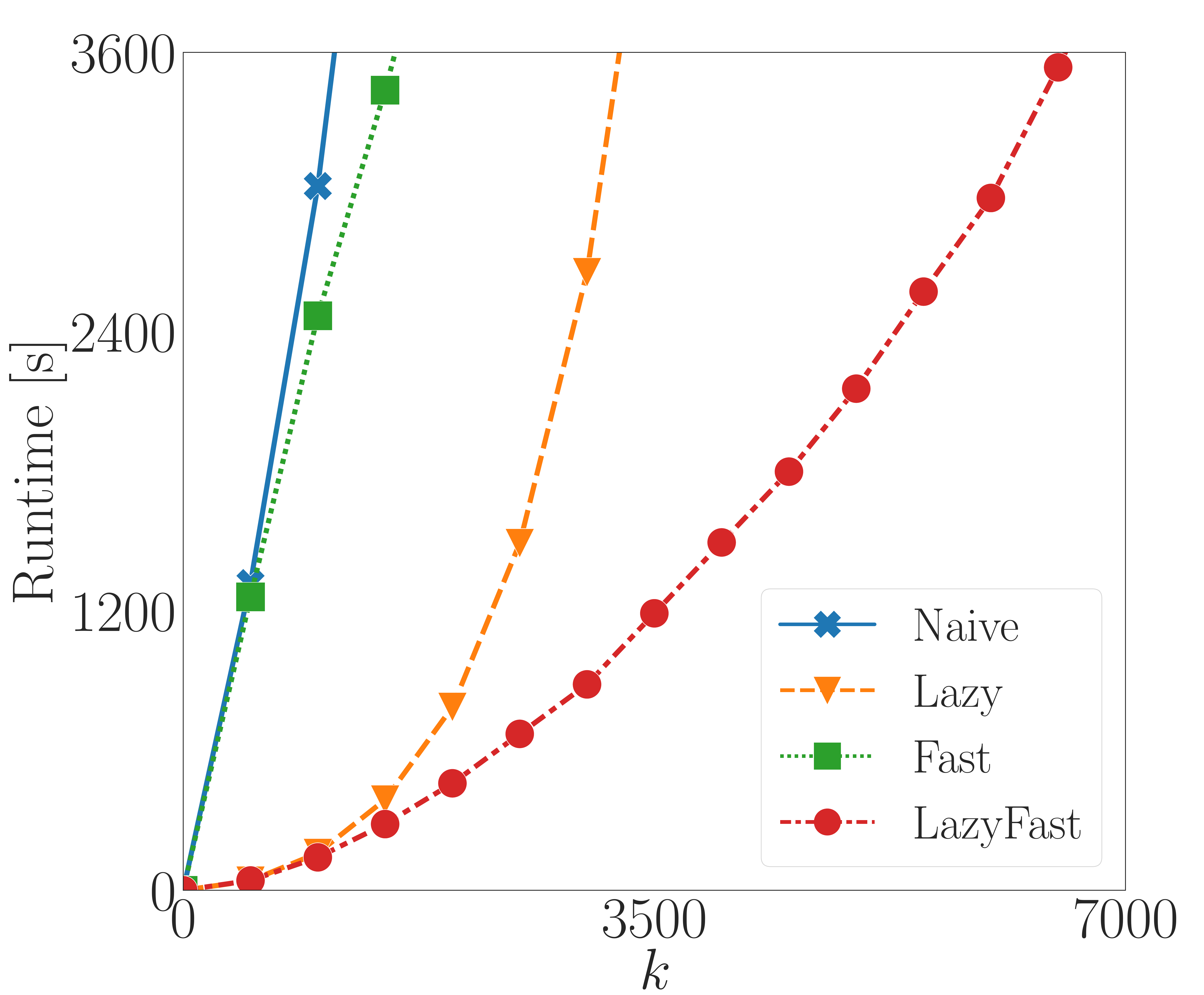}
    \subcaption{MovieLens, Runtime}\label{subfig:stochastic-netflix-k-time}
  \end{minipage}
  \begin{minipage}[b]{0.25\linewidth}
    \vspace{10pt}
    \centering
    \includegraphics[width=\textwidth]{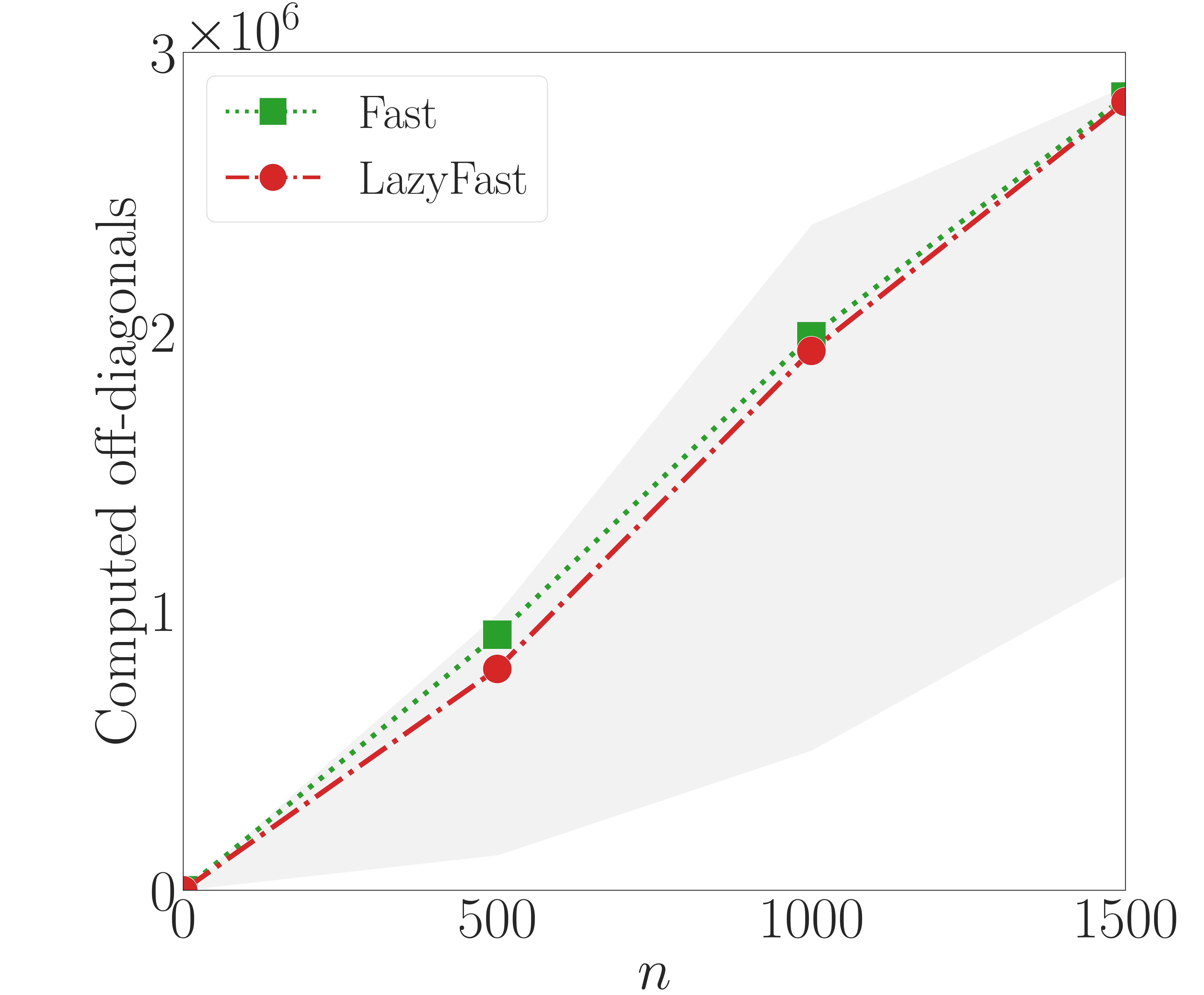}
    \subcaption{$n=6000$, Off-diag.}\label{subfig:stochastic-synth-k-offdiag}
  \end{minipage}%
  \begin{minipage}[b]{0.25\linewidth}
    \vspace{10pt}
    \centering
    \includegraphics[width=\textwidth]{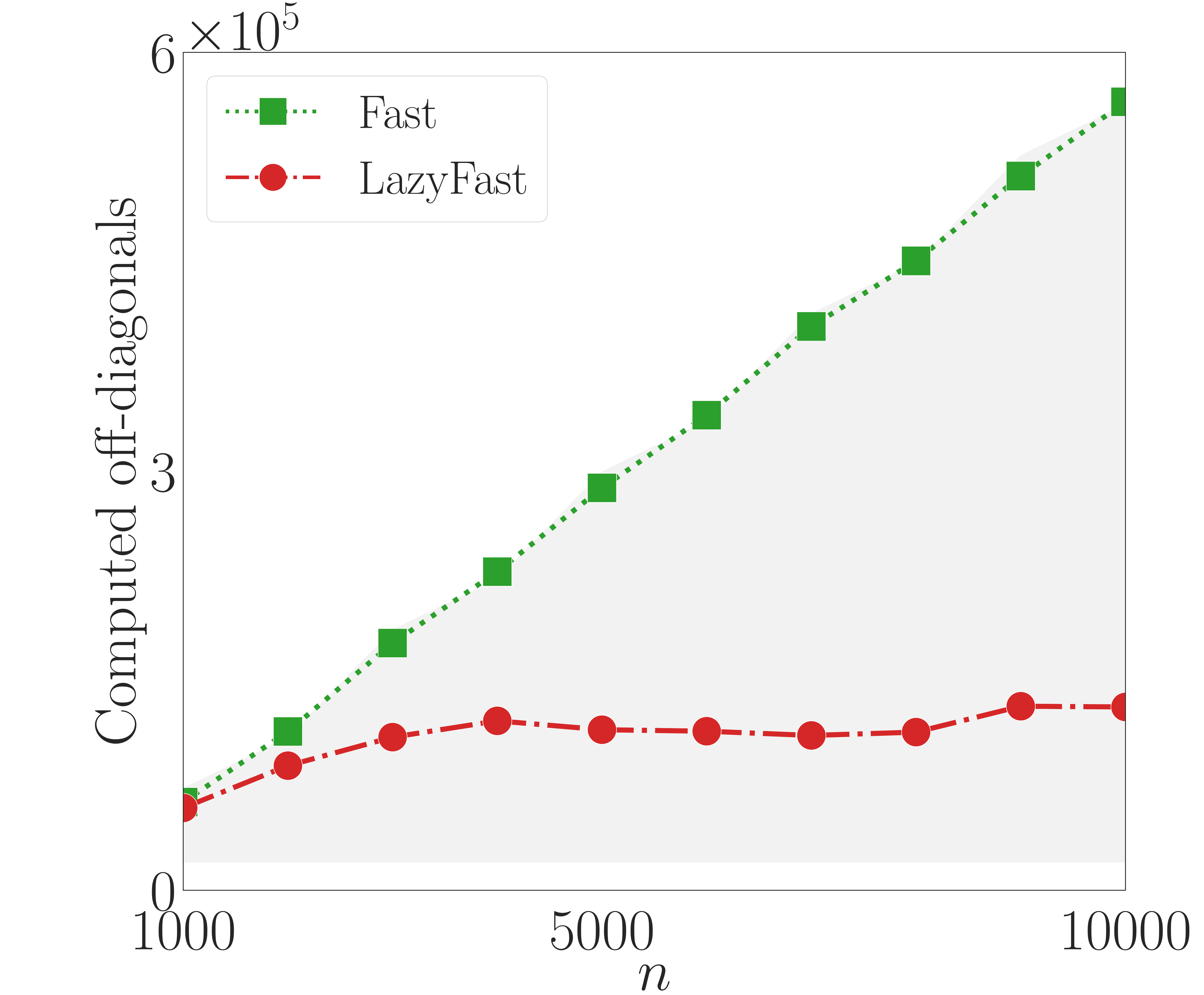}
    \subcaption{$k=200$, Off-diag.}\label{subfig:stochastic-synth-n-offdiag}
  \end{minipage}%
  \begin{minipage}[b]{0.25\linewidth}
    \vspace{10pt}
    \centering
    \includegraphics[width=\textwidth]{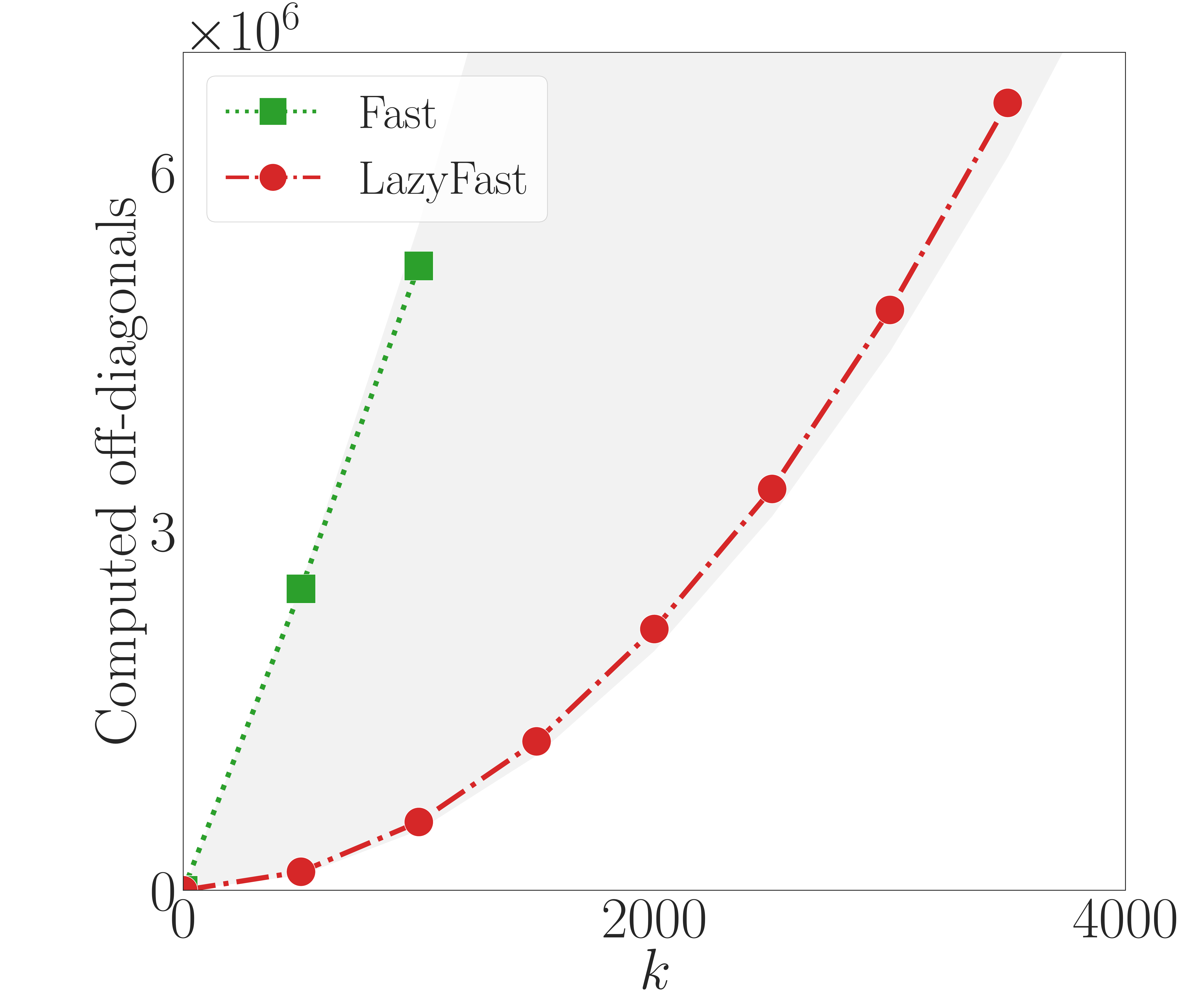}
    \subcaption{Netflix, Off-diag.}\label{subfig:stochastic-movielens-k-offdiag}
  \end{minipage}%
  \begin{minipage}[b]{0.25\linewidth}
    \vspace{10pt}
    \centering
    \includegraphics[width=\textwidth]{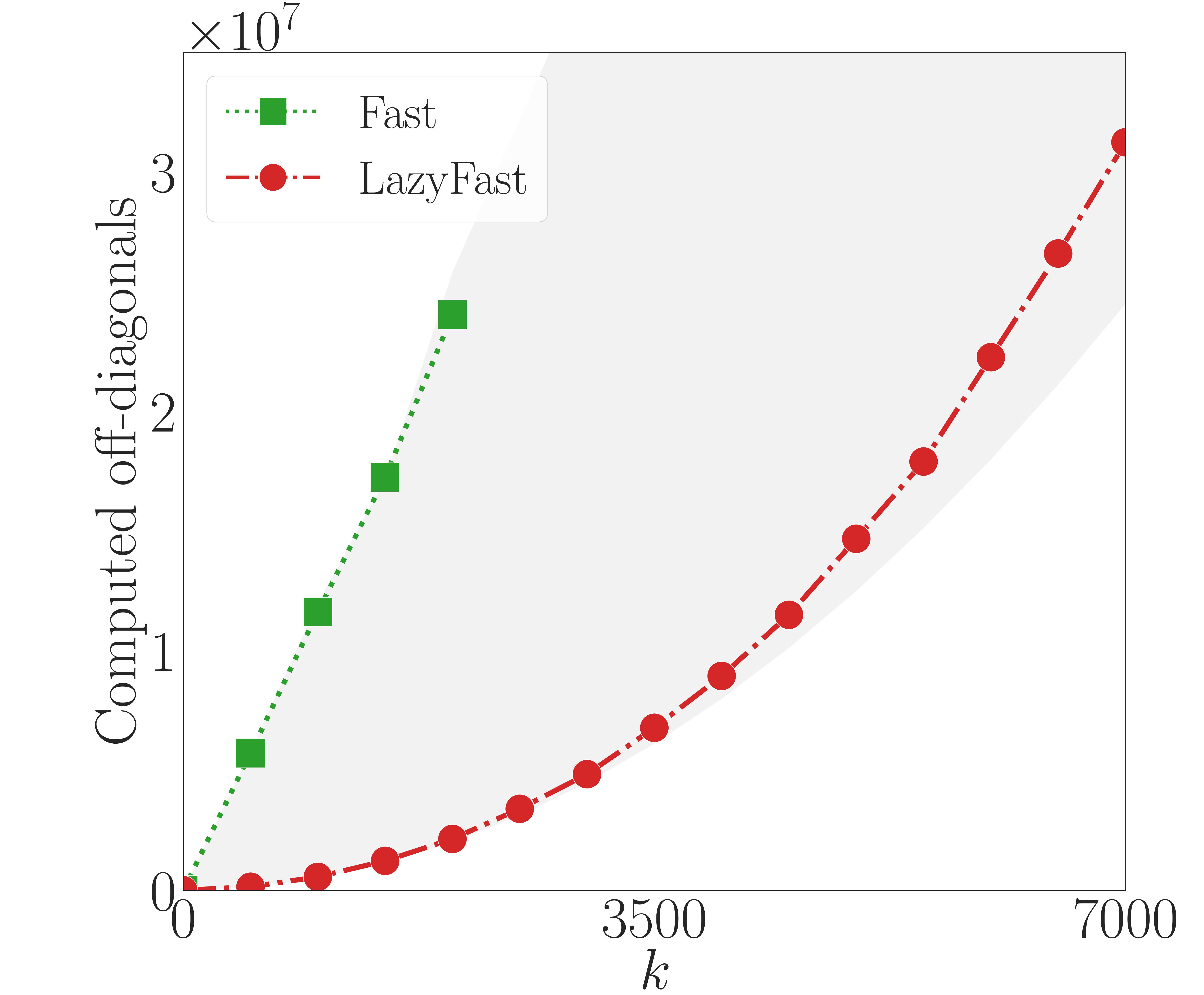}
    \subcaption{MovieLens, Off-diag.}\label{subfig:stochastic-netflix-k-offdiag}
  \end{minipage}
  \caption{Results of \textsc{StochasticGreedy}.
  In the lower figures, the gray band indicates the range of the possible number of computed off-diagonals: $\left[k(k-1)/2, \left(n-k/2\right)(k-q-1)+ kq/2\right]$, where $q\in\N$ is the quotient of $n$ divided by $s$.}\label{fig:stochastic-greedy}
\end{figure}

\begin{figure}[tb]
  \begin{minipage}[b]{0.25\linewidth}
    \centering
    \includegraphics[width=\textwidth]{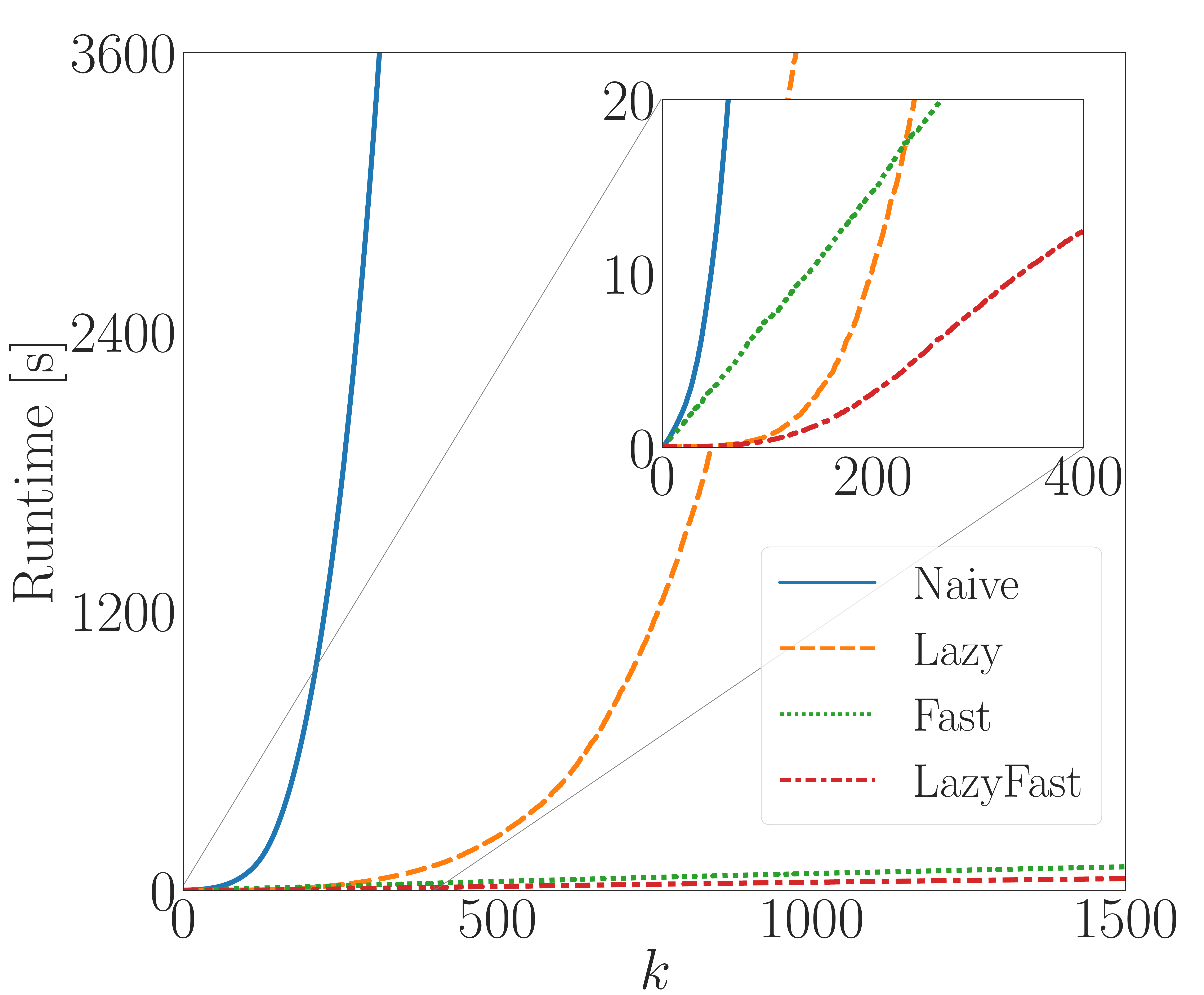}
    \subcaption{$n=6000$, Runtime}\label{subfig:interlace-synth-k-time}
  \end{minipage}%
  \begin{minipage}[b]{0.25\linewidth}
    \centering
    \includegraphics[width=\textwidth]{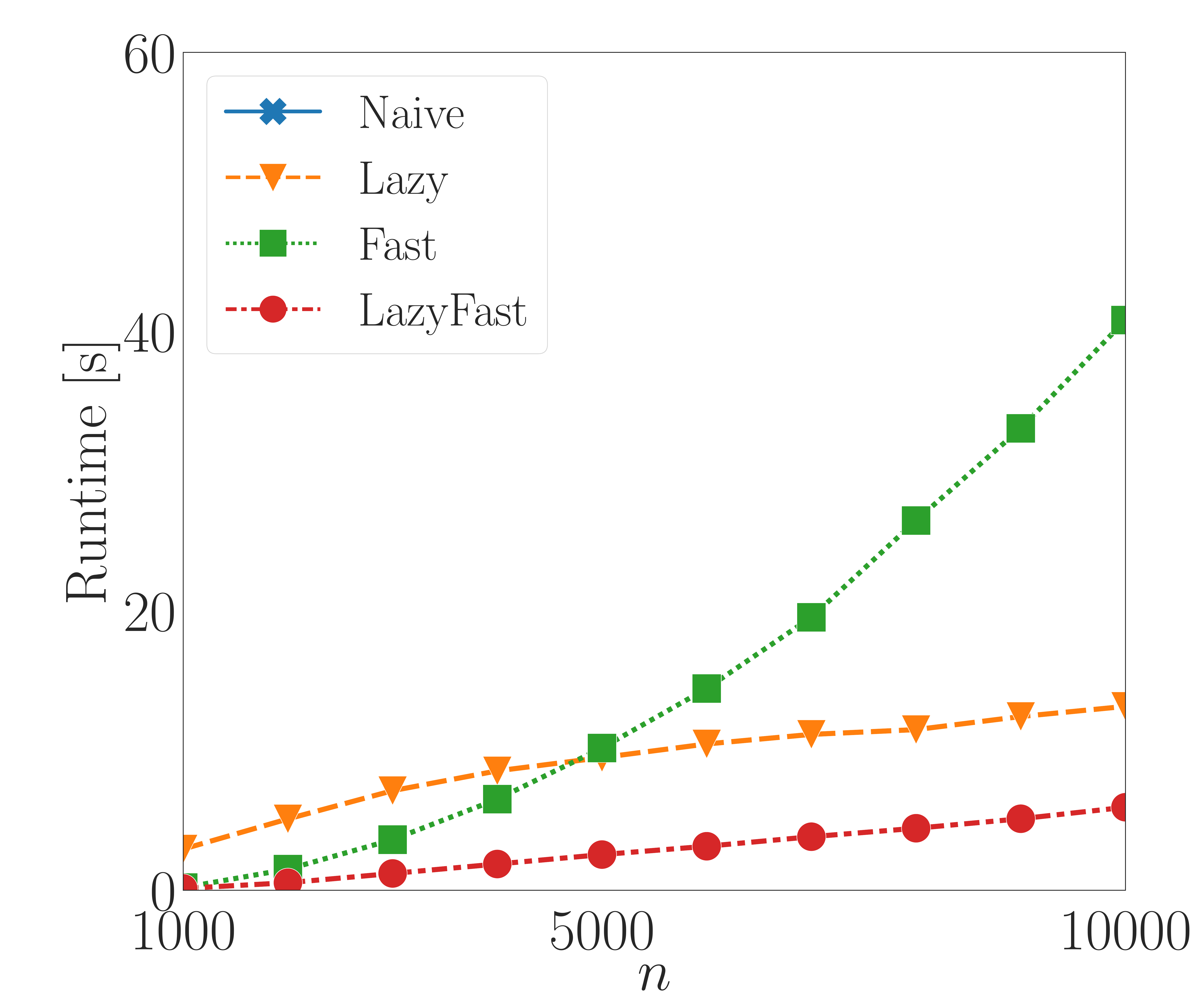}
    \subcaption{$k=200$, Runtime}\label{subfig:interlace-synth-n-time}
  \end{minipage}%
  \begin{minipage}[b]{0.25\linewidth}
    \centering
    \includegraphics[width=\textwidth]{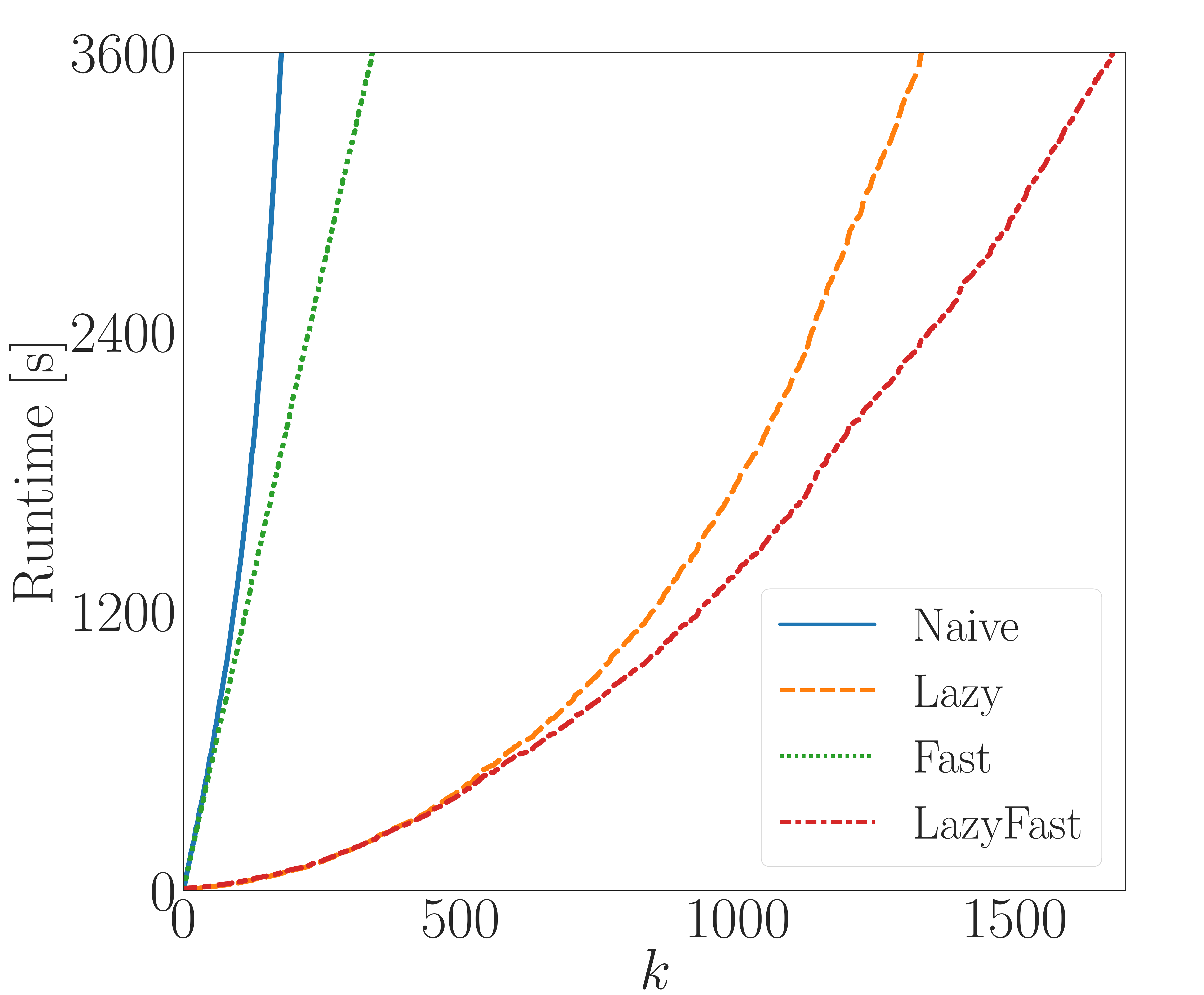}
    \subcaption{Netflix, Runtime}\label{subfig:interlace-movielens-k-time}
  \end{minipage}%
  \begin{minipage}[b]{0.25\linewidth}
    \centering
    \includegraphics[width=\textwidth]{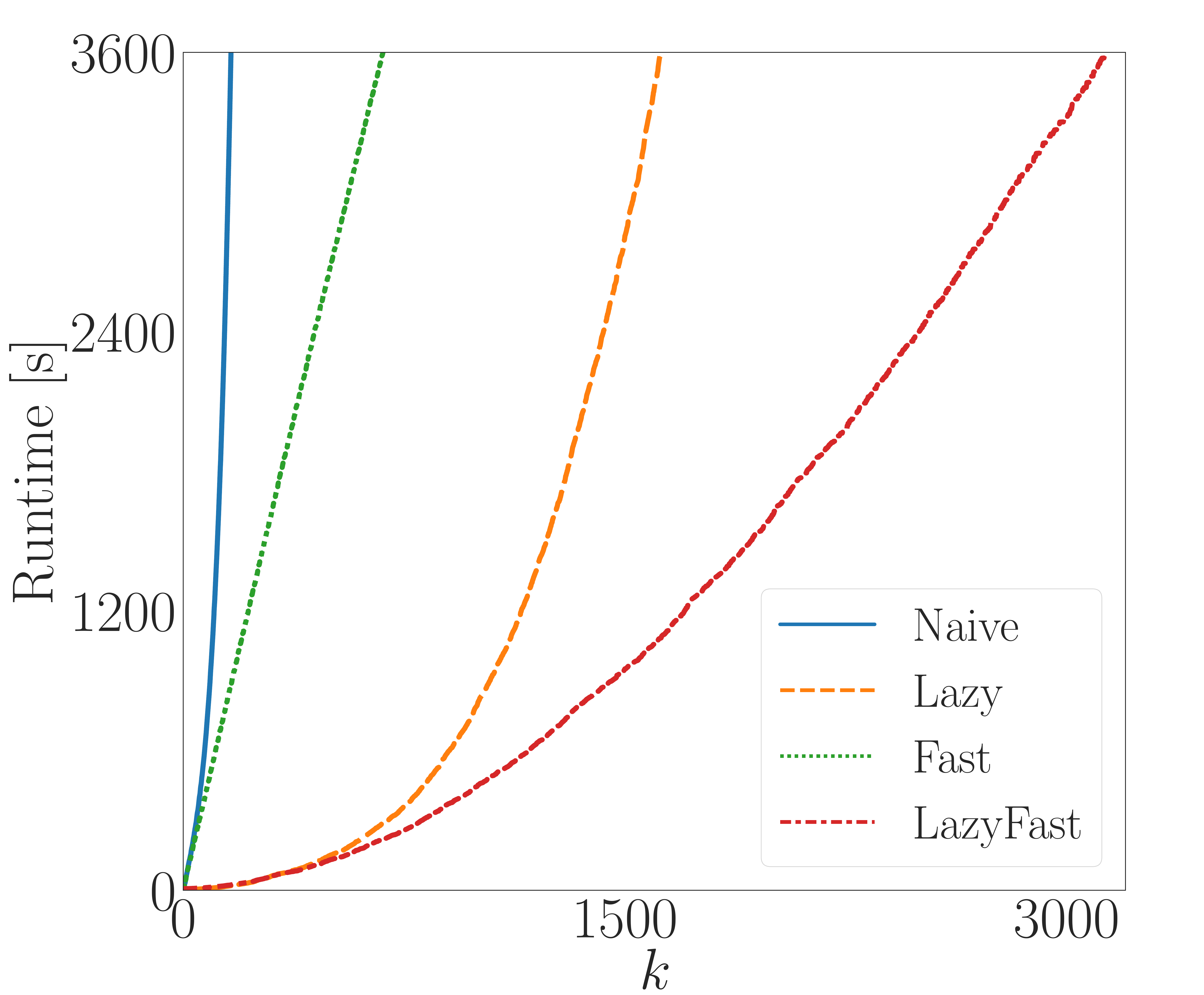}
    \subcaption{MovieLens, Runtime}\label{subfig:interlace-netflix-k-time}
  \end{minipage}
  \begin{minipage}[b]{0.25\linewidth}
    \vspace{10pt}
    \centering
    \includegraphics[width=\textwidth]{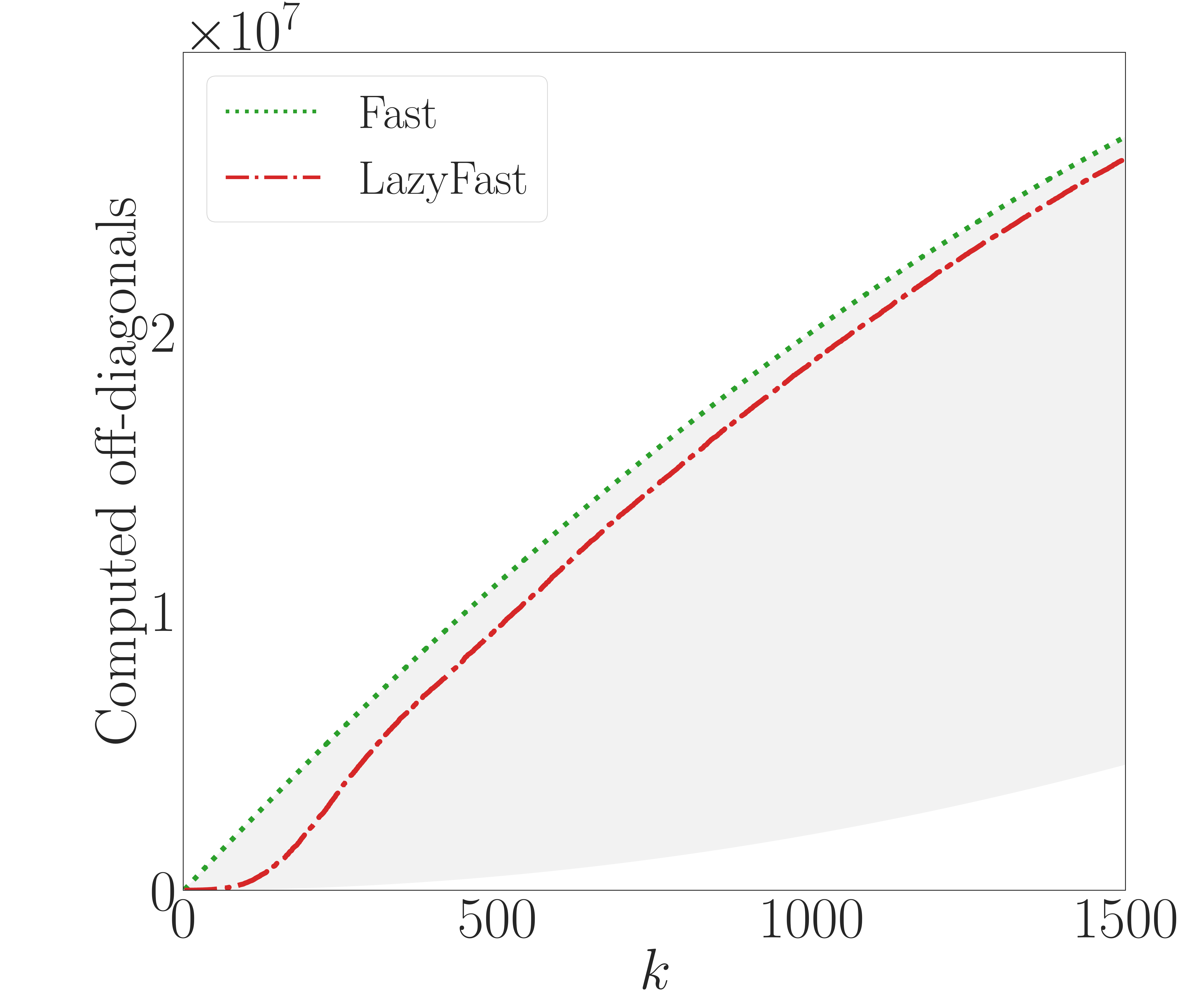}
    \subcaption{$n=6000$, Off-diag.}\label{subfig:interlace-synth-k-offdiag}
  \end{minipage}%
  \begin{minipage}[b]{0.25\linewidth}
    \vspace{10pt}
    \centering
    \includegraphics[width=\textwidth]{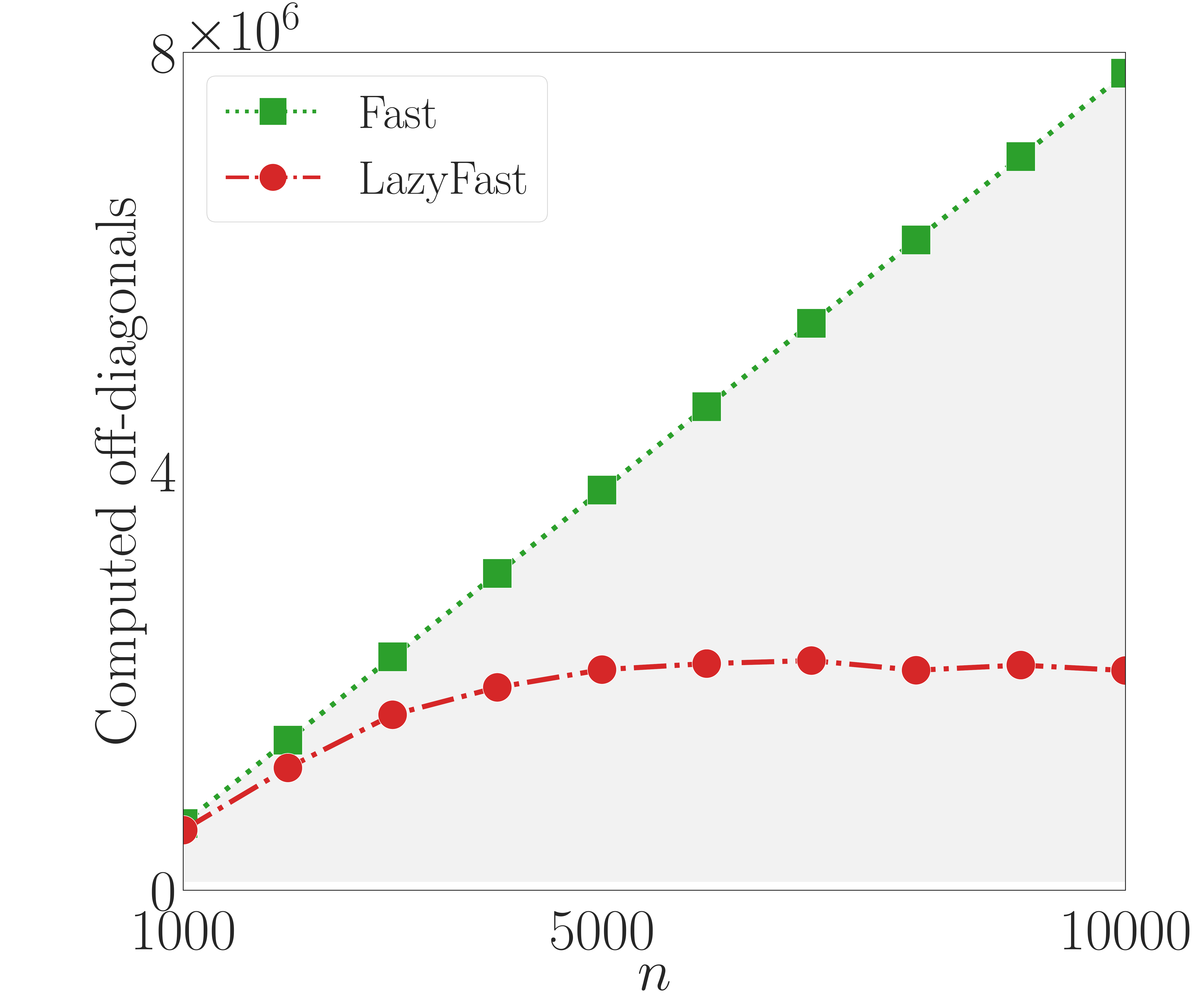}
    \subcaption{$k=200$, Off-diag.}\label{subfig:interlace-synth-n-offdiag}
  \end{minipage}%
  \begin{minipage}[b]{0.25\linewidth}
    \vspace{10pt}
    \centering
    \includegraphics[width=\textwidth]{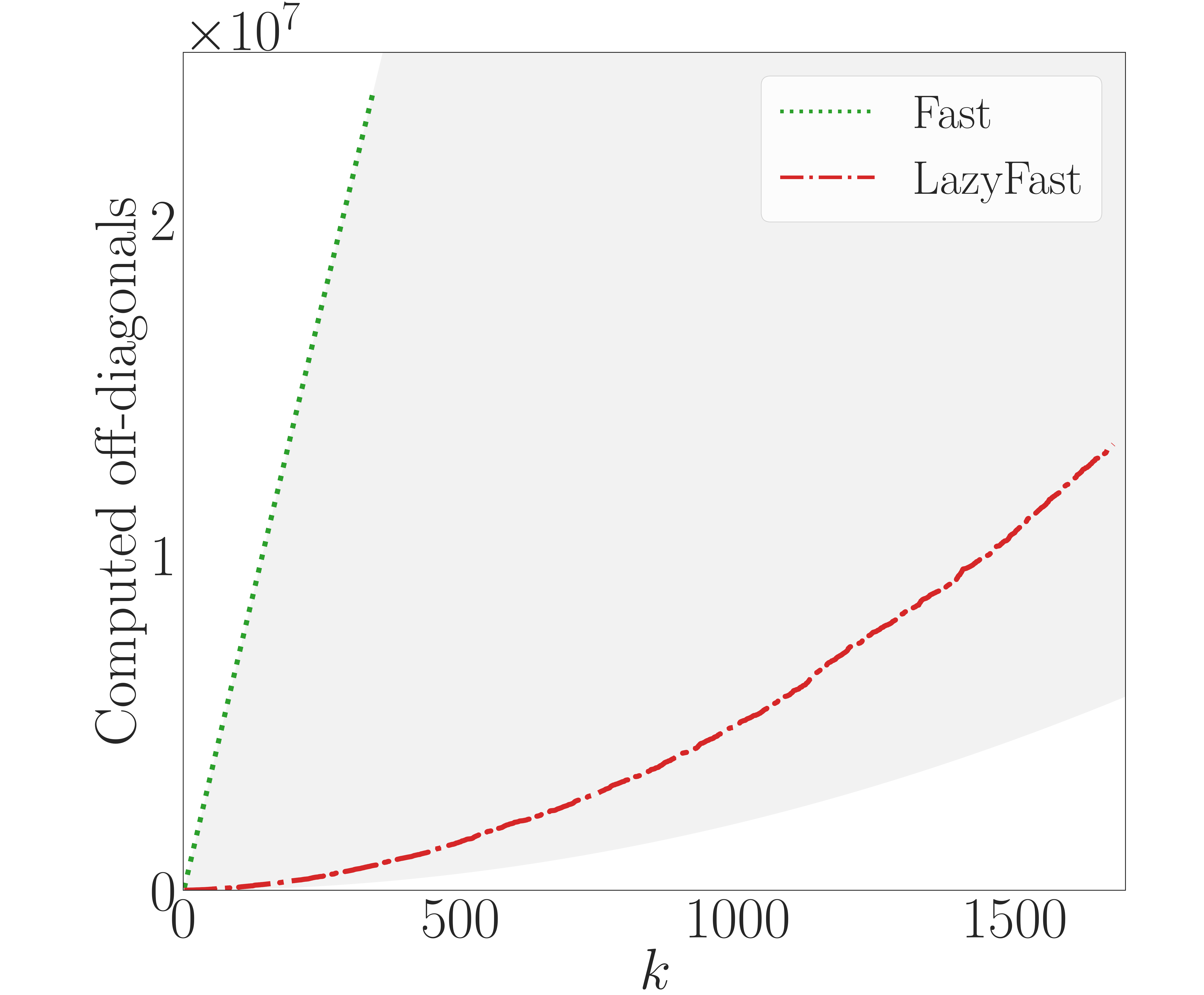}
    \subcaption{Netflix, Off-diag.}\label{subfig:interlace-movielens-k-offdiag}
  \end{minipage}%
  \begin{minipage}[b]{0.25\linewidth}
    \vspace{10pt}
    \centering
    \includegraphics[width=\textwidth]{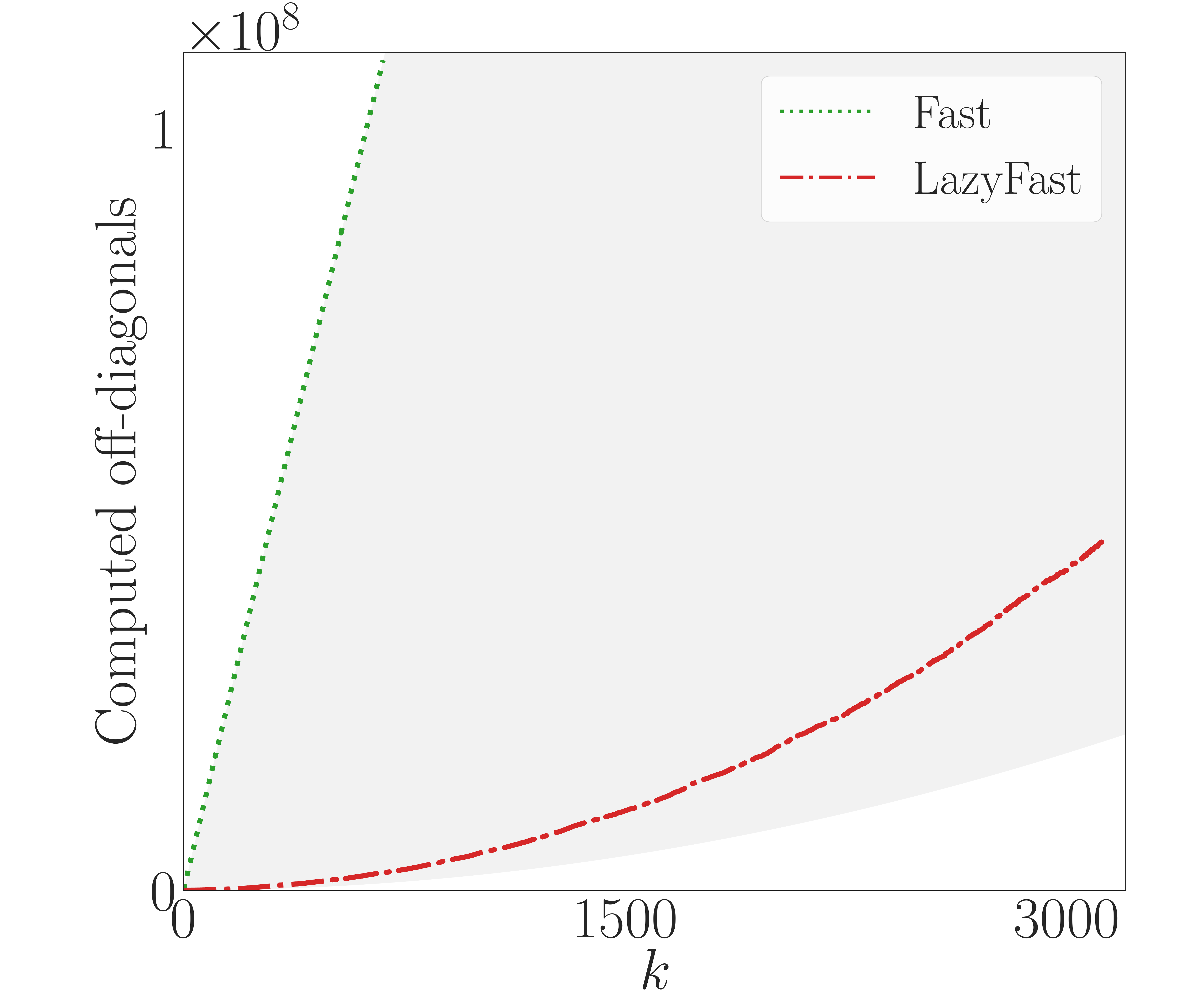}
    \subcaption{MovieLens, Off-diag.}\label{subfig:interlace-netflix-k-offdiag}
  \end{minipage}
  \caption{Results of \textsc{InterlaceGreedy}.
  In \cref{subfig:interlace-synth-k-time}, enlarged views of lower left parts are shown for visibility.
  In \cref{subfig:interlace-synth-n-time}, runtime of Naive is not shown since it did not finish in $60$ seconds even for $n=1000$.
  In the lower figures, the gray band indicates the range of the possible number of computed off-diagonals: $\left[2 k(k-1), 4(n-k)(k-1)\right]$.}\label{fig:interlace-greedy}
  \vspace{10pt}
  \begin{minipage}[b]{0.25\linewidth}
    \centering
    \includegraphics[width=\textwidth]{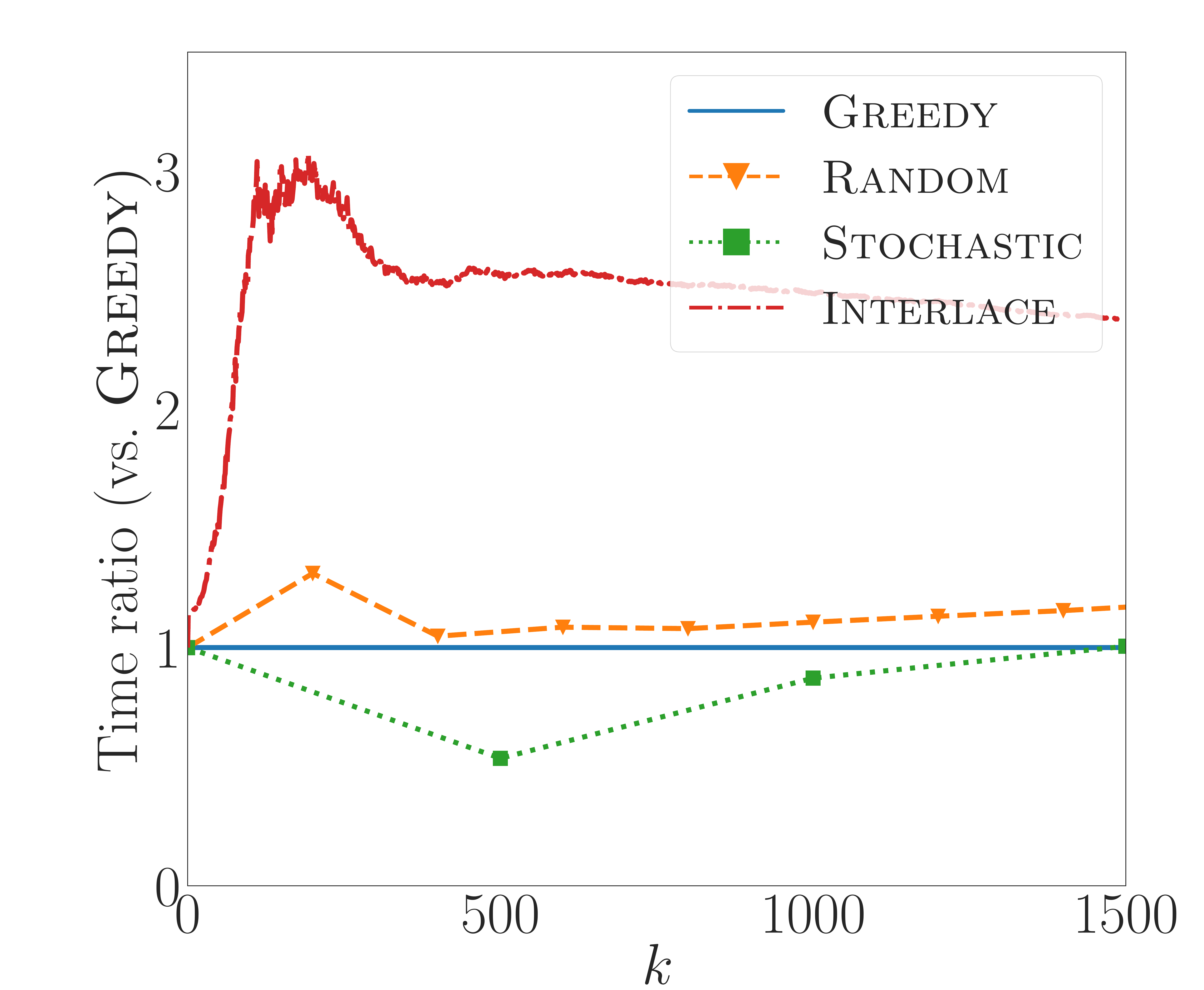}
    \subcaption{$n=6000$, Runtime}\label{subfig:synth-k-function_value}
  \end{minipage}%
  \begin{minipage}[b]{0.25\linewidth}
    \centering
    \includegraphics[width=\textwidth]{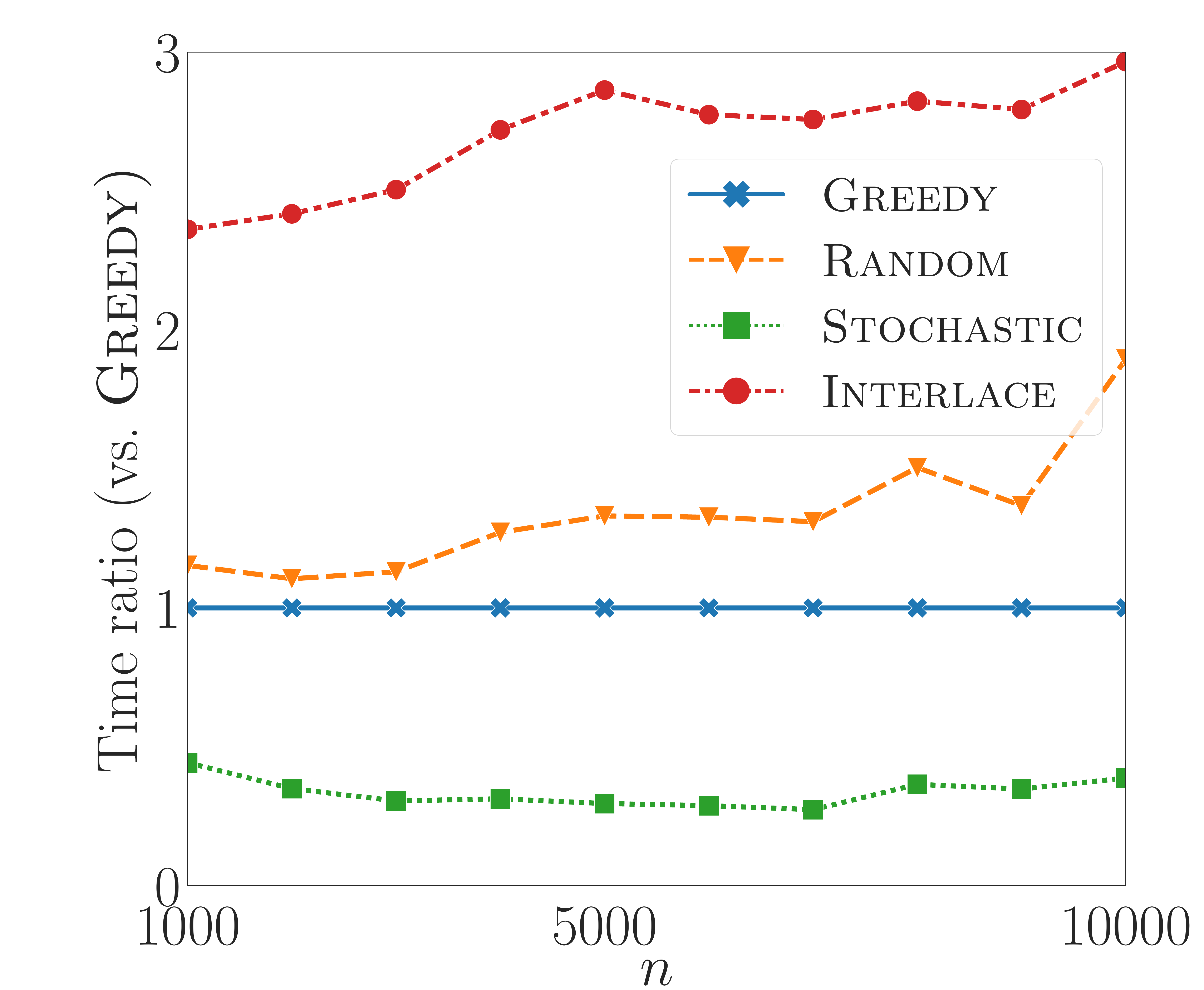}
    \subcaption{$k=200$, Runtime}\label{subfig:synth-n-function_value}
  \end{minipage}%
  \begin{minipage}[b]{0.25\linewidth}
    \centering
    \includegraphics[width=\textwidth]{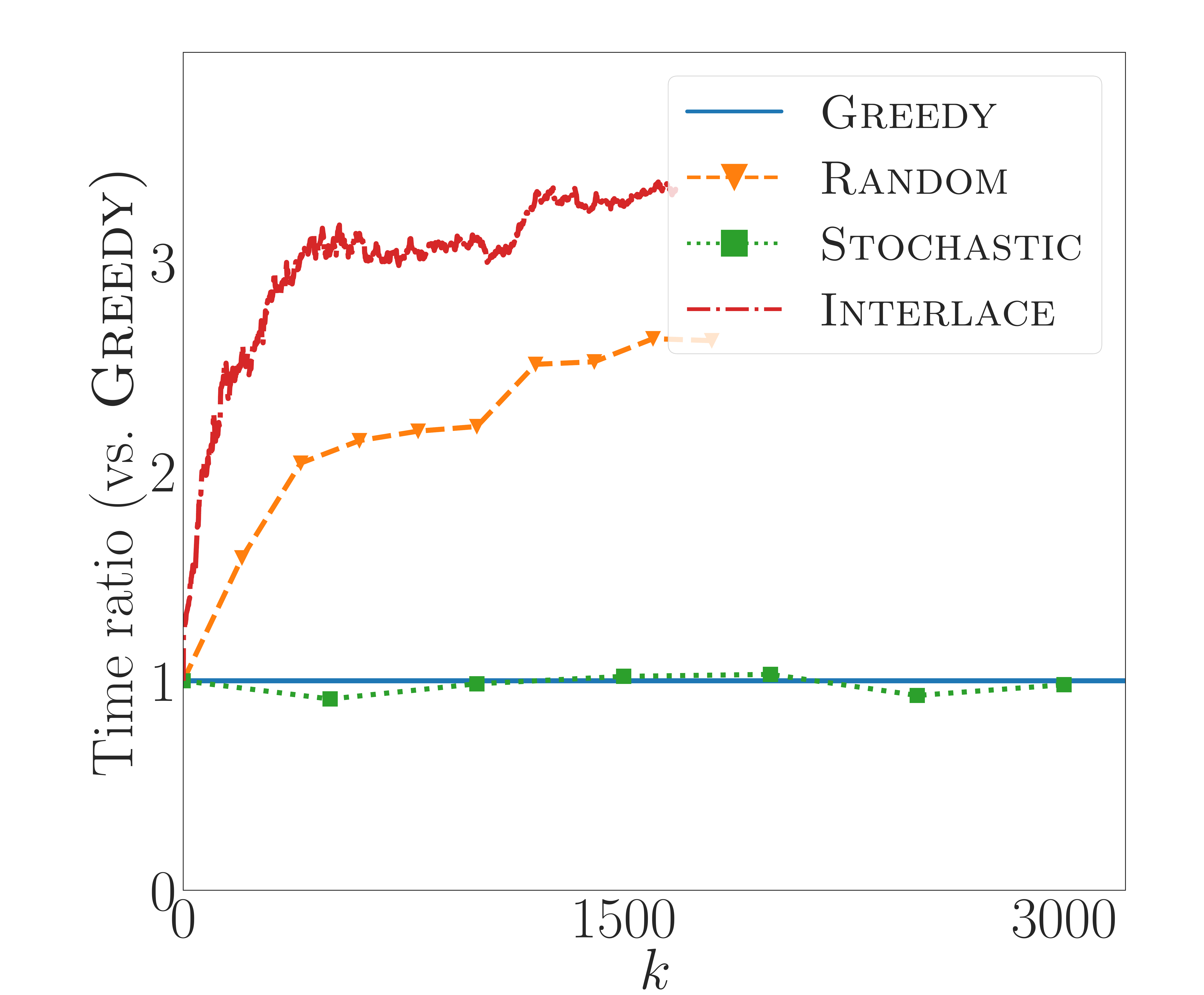}
    \subcaption{Netflix, Runtime}\label{subfig:movielens-k-function_value}
  \end{minipage}%
  \begin{minipage}[b]{0.25\linewidth}
    \centering
    \includegraphics[width=\textwidth]{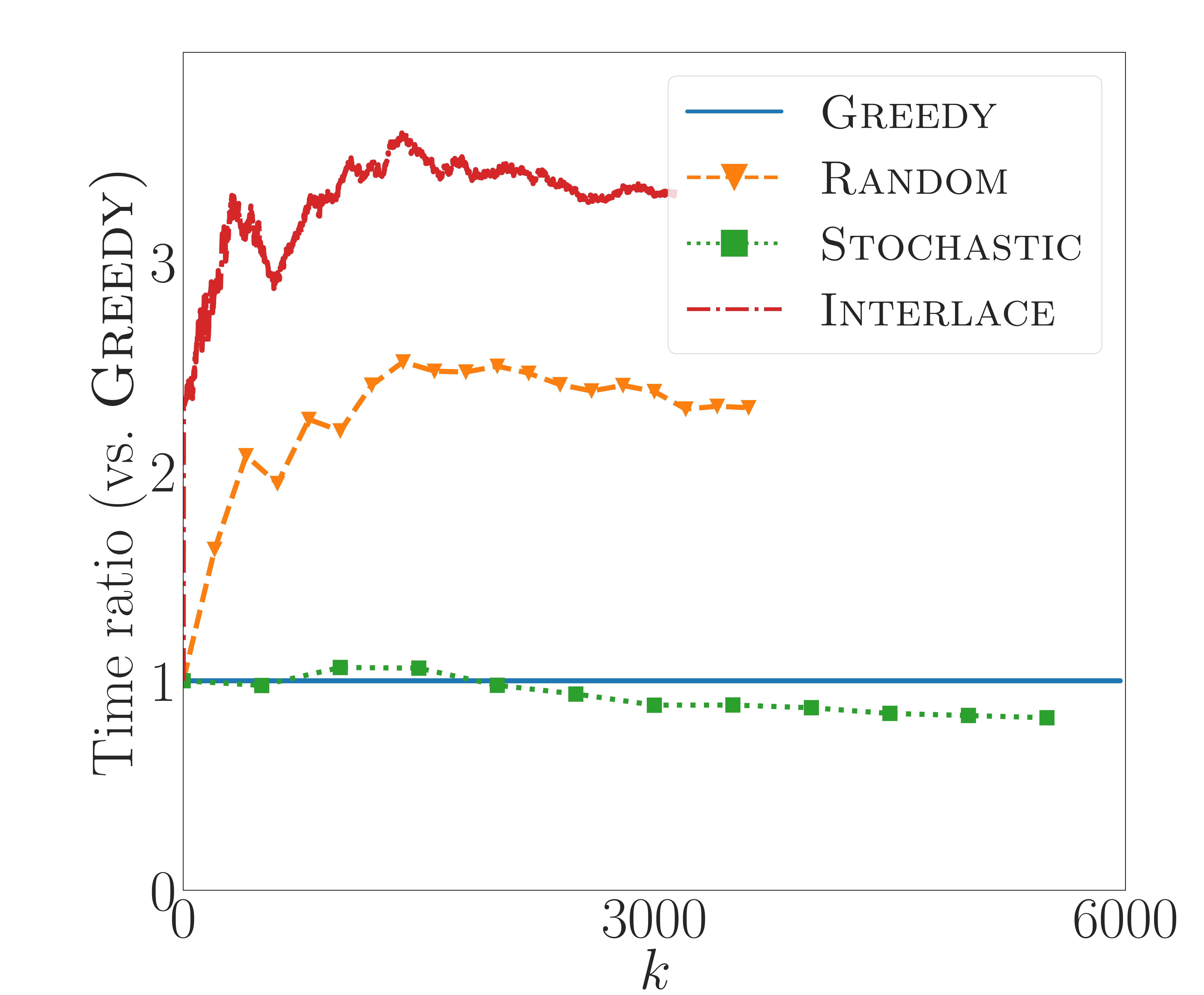}
    \subcaption{MovieLens, Runtime}\label{subfig:netflix-k-function_value}
  \end{minipage}
  \begin{minipage}[b]{0.25\linewidth}
    \centering
    \includegraphics[width=\textwidth]{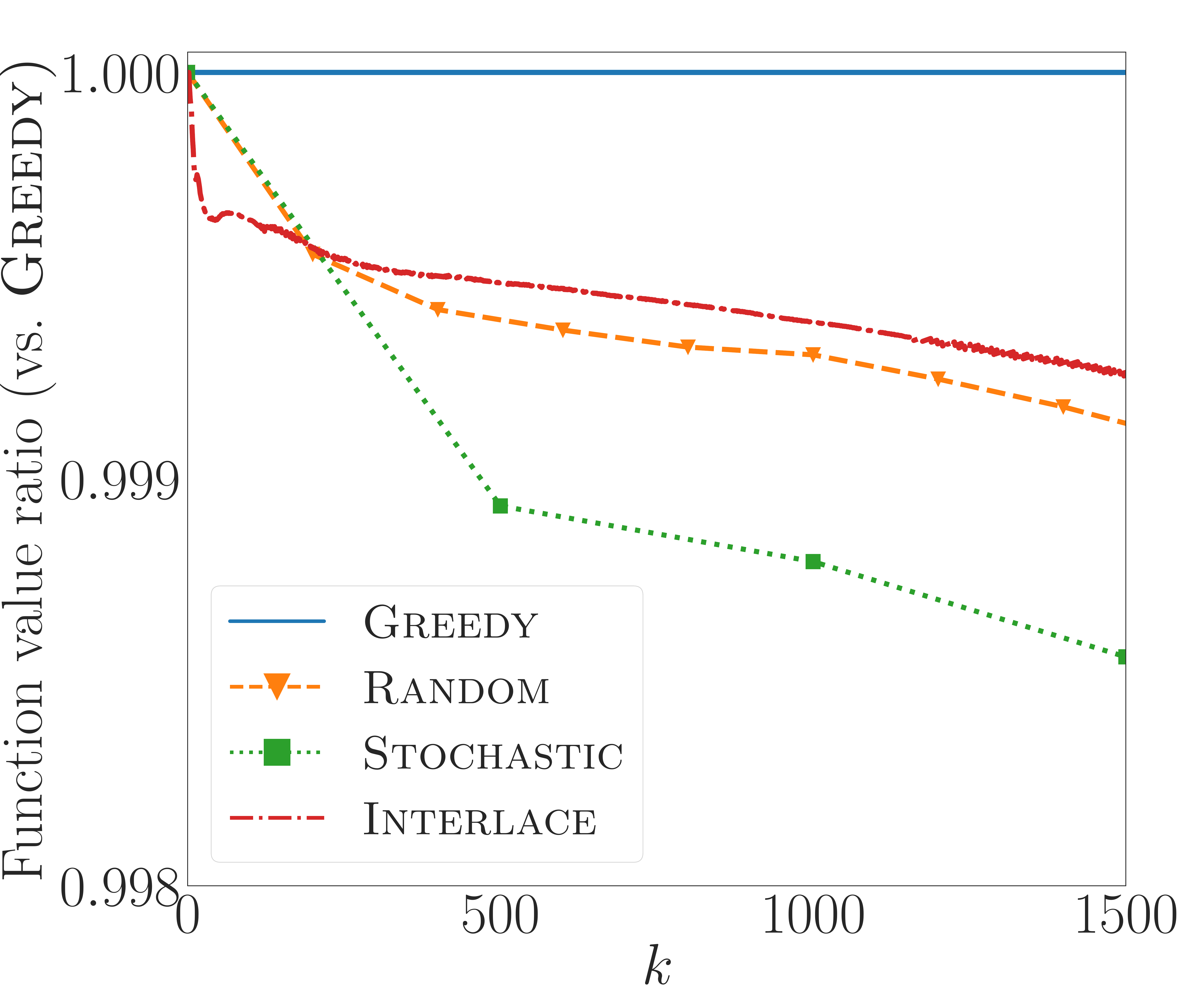}
    \subcaption{$n=6000$, Obj.\ value}\label{subfig:synth-k-function_value}
  \end{minipage}%
  \begin{minipage}[b]{0.25\linewidth}
    \centering
    \includegraphics[width=\textwidth]{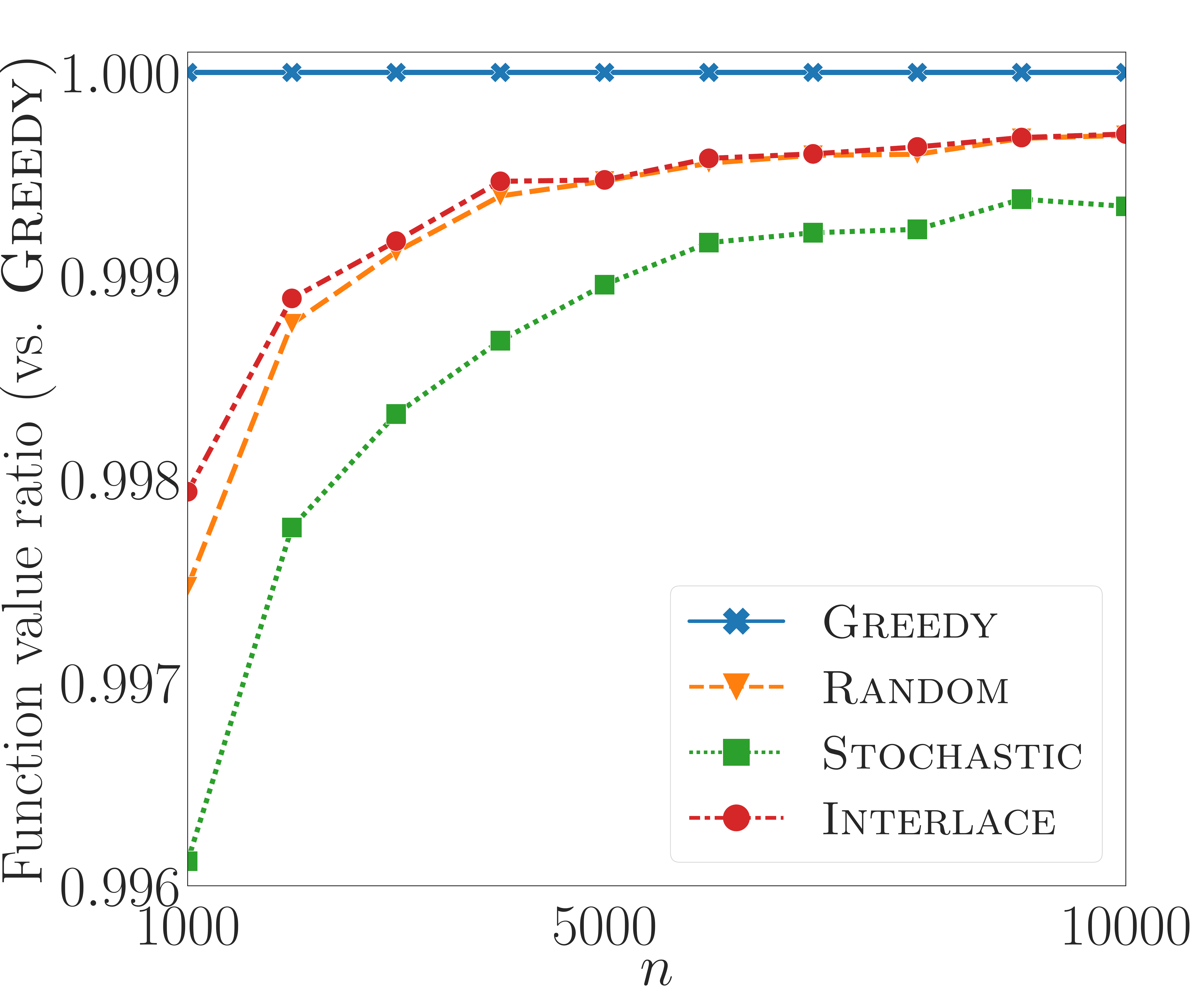}
    \subcaption{$k=200$, Obj.\ value}\label{subfig:synth-n-function_value}
  \end{minipage}%
  \begin{minipage}[b]{0.25\linewidth}
    \centering
    \includegraphics[width=\textwidth]{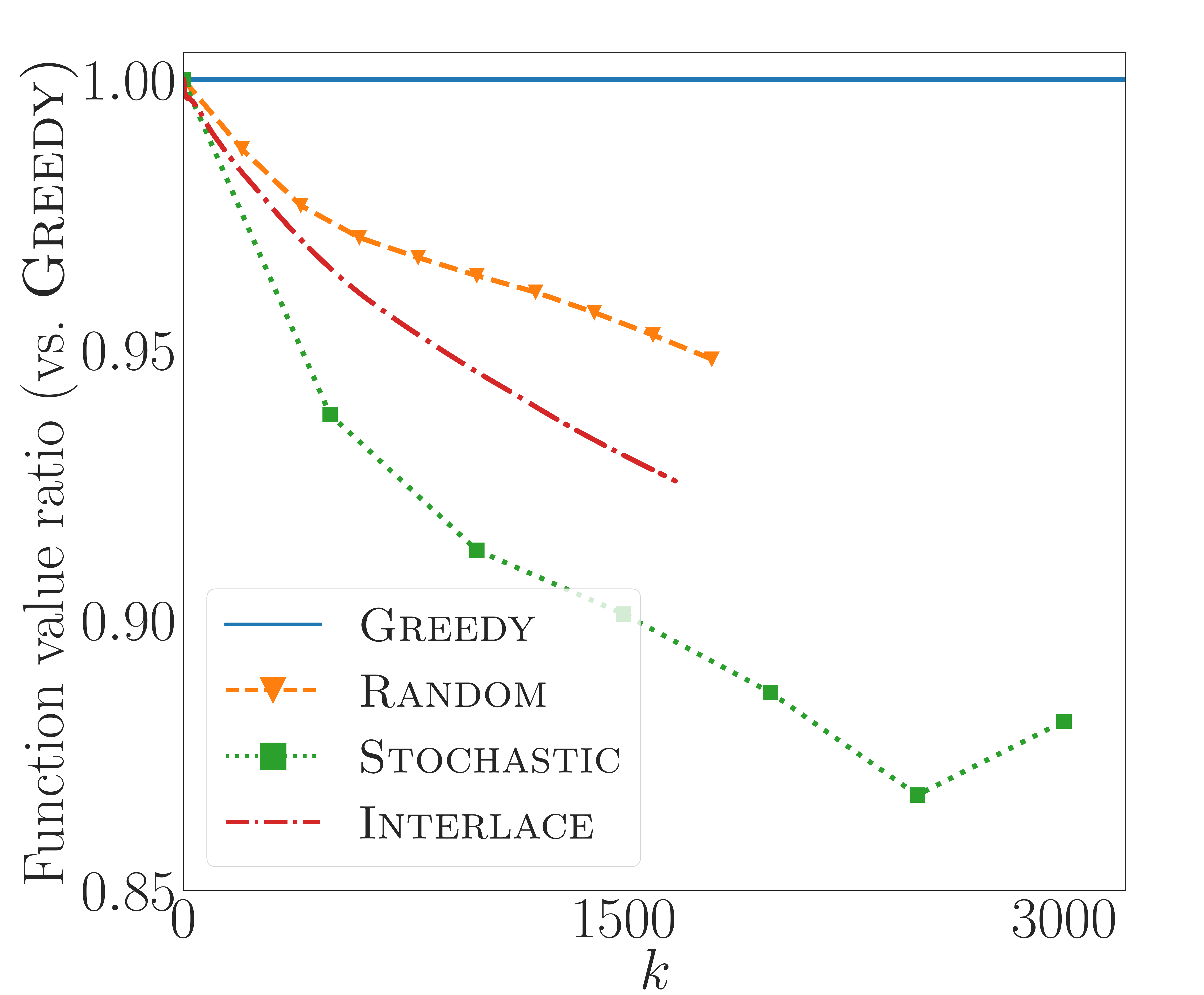}
    \subcaption{Netflix, Obj.\ value}\label{subfig:movielens-k-function_value}
  \end{minipage}%
  \begin{minipage}[b]{0.25\linewidth}
    \centering
    \includegraphics[width=\textwidth]{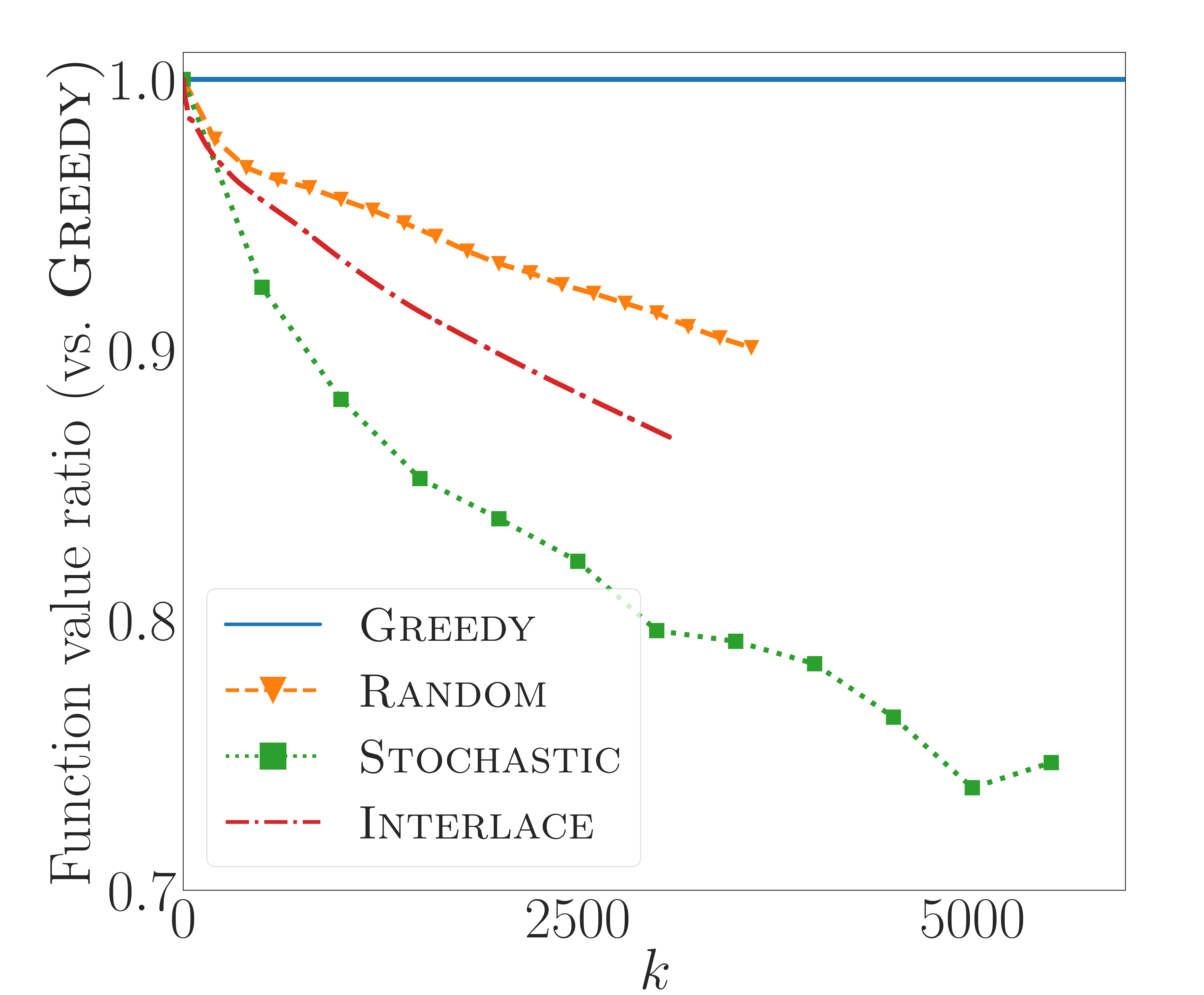}
    \subcaption{MovieLens, Obj.\ value}\label{subfig:netflix-k-function_value}
  \end{minipage}
  \caption{Running-time and objective-value ratios relative to those of \textsc{LazyFastGreedy}.}\label{fig:objective-values}
\end{figure}

\end{document}